\PassOptionsToPackage{dvipsnames}{xcolor}
\documentclass[acmsmall, nonacm, screen]{acmart}

\usepackage{listings}
\usepackage{comment}
\usepackage{wrapfig}
\usepackage{tabularx}
\usepackage[inline]{enumitem}
\usepackage{multicol}
\usepackage{algorithm}
\usepackage{algpseudocode}
\usepackage{mdframed}
\algnewcommand\Failure{\textbf{Failure}}
\algnewcommand\True{\textbf{True}}
\algnewcommand\False{\textbf{False}}
\algnewcommand\Land{\textbf{and }}
\algnewcommand\Or{\textbf{or }}
\algnewcommand\algorithmicsmatch{\textbf{match}}
\algnewcommand\algorithmiccase{\textbf{case}}
\algnewcommand\algorithmicdefault{\textbf{default}}
\algdef{SE}[MATCH]{Match}{EndMatch}[1]{\algorithmicsmatch\ #1:}{\algorithmicend\ \algorithmicsmatch}%
\algdef{SE}[CASE]{Case}{EndCase}[1]{\algorithmiccase\ #1 \ \algorithmicdo}{\algorithmicend\ \algorithmiccase}%
\algdef{SE}[DEFAULT]{Default}{EndCase}{\algorithmicdefault\ \algorithmicdo}{\algorithmicend\ \algorithmicdefault}%
\algtext*{EndMatch}%
\algtext*{EndCase}%

\definecolor{typecolor}{HTML}{B00040}
\definecolor{actioncolor}{rgb}{0.58, 0.0, 0.82}
\newcommand{\typekw}[1]{\textbf{\textcolor{typecolor}{#1}}}
\newcommand{\typeact}[1]{\textbf{\textcolor{actioncolor}{#1}}}
\newcommand{\bind}[0]{\textcolor{brown}{\textsc{Bind}}}
\newcommand{\subhist}[0]{\textcolor{brown}{\textsc{SubHist}}}
\newcommand{\subconc}[0]{\textcolor{brown}{\textsc{SubConc}}}
\newcommand{\histconj}[0]{\textcolor{brown}{\textsc{HistConj}}}
\newcommand{\overa}[1]{\textcolor{red}{\{}#1\textcolor{red}{\}}}
\newcommand{\under}[1]{\textcolor{blue}{[}#1\textcolor{blue}{]}}
\newcommand{\nextd}[0]{\textcolor{brown}{\textsc{new}}}
\newcommand{\bluet}[1]{\textcolor{blue}{#1}}
\newcommand{\oranget}[1]{\textcolor{orange}{#1}}

\theoremstyle{definition}
\newtheorem{definition}{Definition}[section]
\newtheorem{notation}{Notation}[section]
\newtheorem{example}{Example}[section]
\newtheorem{lemma}{Lemma}[section]
\newtheorem{theorem}{Theorem}[section]
\newtheorem{note}{Note}[section]

\AtBeginDocument{%
  }

\lstdefinestyle{OCaml}{
    language=ML,
    numbers=left,
    numbersep=8pt,
    breaklines=true,
    xleftmargin=2em,
    basicstyle=\ttfamily\footnotesize,
    keywordstyle=\color[rgb]{0.0, 0.5, 0.0}\bfseries,
    keywordstyle=[2]\color[rgb]{0.6, 0.0, 0.0}\bfseries,
    stringstyle=\color[rgb]{0.6, 0.0, 0.0},
    commentstyle=\color{gray},
    morekeywords={match,with,let,in,rec,type,fun,function,val,sig,struct,module,open,of,and,or,not},
    morekeywords=[2]{int,string,bool,float,nat,list,array,option,unit,char},
    escapeinside={??},
    alsoletter={'},
    morestring=[b]',
    alsoletter={(},
    literate=
        {()}{{\textcolor[rgb]{0.0, 0.5, 0.0}{()}}}1
        {[]}{{\textcolor[rgb]{0.0, 0.5, 0.0}{[]}}}1
}

\begin{document}\sloppy

\title{Coverage Types for Resource-Based Policies}

\author{Angelo Passarelli}
\email{a.passarelli4@studenti.unipi.it}
\orcid{0009-0009-5714-4922}
\author{Gian-Luigi Ferrari}
\email{gian-luigi.ferrari@unipi.it}
\orcid{0000-0003-3548-5514}
\affiliation{%
  \institution{University of Pisa}
  \department{Computer Science Department}
  \city{Pisa}
  \state{Tuscany}
  \country{Italy}
}


\begin{abstract}
    Coverage Types provide a suitable type mechanism that integrates under-approximation logic to support Property-Based Testing. They are used to type the return value of a function that represents an input test generator. This allows us to statically assert that an input test generator not only produces valid input tests but also generates all possible ones, ensuring completeness.
    
    In this paper, we extend the coverage framework to guarantee the correctness of Property-Based Testing with respect to resource usage in the input test generator. This is achieved by incorporating into Coverage Types a notion of effect, which represents an over-approximation of operations on relevant resources. Programmers can define resource usage policies through logical annotations, which are then verified against the effect associated with the Coverage Type.
\end{abstract}




\maketitle

\section{Introduction}\label{sec:intro}
Resource management and security are essential aspects of software development, as proper handling of these factors can prevent issues ranging from minor bugs to major security vulnerabilities. Here, the term resource refers not only to local resources—such as files, sockets, memory, and network connections—but also, and more importantly, to remote computational entities, such as remote APIs, service entry points, and database connections. The reason for this inclusion is straightforward. For example, even an unopened file could introduce a vulnerability that attackers might exploit. These seemingly simple oversights can be difficult to detect, even at runtime, especially if they occur in rarely accessed execution paths. In such cases, the issue may only become apparent too late, after the software has already been released to users. Implementing robust resource management mechanisms throughout the entire software development lifecycle not only helps prevent bugs and crashes but also strengthens the security and reliability of the software.

To manage resources effectively, software engineering methodologies implement a combination of policies and best practices. Resource management can be approached in two main ways: dynamically \cite{dynamic1-logicguard,dynamic2-java,dynamic3-trace-prop} and statically \cite{static3-types-effects,static1-costa,static2}.

Dynamic resource management involves continuously testing and monitoring software in various environments to identify and address resource-related issues during runtime. This approach exploits tools such as stress testing and runtime monitoring, which are capable of detecting problems like memory leaks, unclosed file handles, and excessive API calls in real-time. By identifying these issues as they occur, developers can promptly address them before they negatively affect the end-user experience. On the other hand, static resource management focuses on analysing the codebase to identify potential violations of resource management policies before the software is executed. By using static analysis tools, developers can catch issues early in the development process, which significantly reduces the risk of bugs and enhances the overall quality of the software \cite{taint}. Several static analyses have been proposed to verify the correct use of program resources \cite{resource-usage,time-regions,enforcing,flow-sensitive,path-sensitive,typestate}. Most of these employ a type-based method to ensure resource safety properties \cite{dynamic4-marriot}. These approaches exploits a type system augmented with suitable annotation about resource usage to guide the analysis.

For the sake of simplicity, we focus on remote APIs. Since the body of a remote function is unknown to us, we cannot predict its behaviour concerning resource usage. Consequently, we cannot define the function reliably. Therefore, it is the responsibility of the API provider to specify not only the parameter and return types but also the function's latent effects. These latent effects must be assumed to reliably characterize the function’s behaviour in relation to the resources it utilizes.

In our approach, we refer to the latent effect of a function as \emph{History Expressions} \cite{history2,history3,history}. Formally, a History Expression is a term within a suitable Basic Process Algebra (BPA) \cite{bpa}. Intuitively, a History Expression represents sequences and non-deterministic choices regarding resource actions, potentially extending infinitely.

In this paper, we focus on a static, type-based verification approach that enables the static analysis of resources by defining usage policies through first-order logic formulas. Our approach integrates two distinct lines of research:

The first research direction involves the definition of a \emph{type and effect system} \cite{types-effects-general,static3-types-effects} for functions, where History Expressions capture the effect by describing the set of possible behaviours of a function concerning resource usage. In other words, a History Expression provides an \emph{over-approximation} of the function's behaviour with respect to resources. Usage policies are then verified against each behaviour in this over-approximation. Unlike the approach in \cite{history}, our method defines usage policies as first-order logic formulas, allowing for efficient validation using specialized theorem provers such as \emph{Z3} \cite{z3}.

The second research direction extends the concept of Coverage Types \cite{coverage} by incorporating History Expressions. Coverage Types capture the set of values that an expression is guaranteed to produce, in contrast to the traditional approach where types represent the potential range of values an expression might yield. Notably, the concept of Coverage Types integrates "must-style" reasoning principles, which are inherently under-approximate \cite{IL1,IL2}, making them a fundamental aspect of the type system.

The extended version of Coverage Types seeks to retain the key characteristics of both approaches outlined above. In particular, we adopt the Coverage Type framework as the foundation for Property-Based Testing (PBT) of programs \cite{PBT1,PBT2}, ensuring that the types reliably capture the behaviour of the program with respect to specified properties. Simultaneously, we introduction of History Expressions facilitates the verification of resource usage within programs, providing a mechanism for tracking and reasoning about resource handling. To demonstrate our approach, we utilize a core programming language as a case study, specifically, a functional programming language that has been enhanced with primitives for executing resource operations. This combination allows us to effectively test properties and verify resource handling in a cohesive and integrated manner.

It is important to note that in our type system, History Expressions are decoupled from the return type logic, whether it involves over- or under-approximation. This separation is a crucial feature, as it allows History Expressions to behave as a versatile formal technique that can be applied across various type systems. However, the goal of verifying properties related to resource usage is not limited to the context of Property-Based Testing, but is applicable to a wide range of programs. Indeed, our approach effectively demonstrates the integration of two distinct logical abstract paradigms, namely over- and under-approximation, within a single type system. Specifically, under-approximation is employed for the return type of an expression (represented by the Coverage Type), while over-approximation governs the type associated with resource behaviour (captured by the History Expression). This dual approach allows for a flexible verification framework, capable of addressing both precise guarantees and broader resource behaviour predictions within the same system.

The paper is organized as follows:
\begin{enumerate}
    \item In Section \ref{sec:overview} we provides an informal, intuitive, step-by-step presentation of the key design choices for the types used, illustrated with an example that effectively demonstrates the semantics of the constructs.
    \item Section \ref{sec:history} formally introduces the concept of \emph{History Expressions} and its formal semantics.
    \item Section \ref{sec:policies} formally presents possible ways of defining policies on a History Expression. 
    \item Sections \ref{sec:tuple_type} and \ref{sec:type_system} focus on introducing the formal relationship between the return type and the behaviour type of functions. Specifically, we extend the language $\lambda^{\textbf{TG}}$ and its type system \cite{coverage} to incorporate History Expressions. In Section \ref{sec:type_system} we present the operational semantics of the language and the main results regarding the correctness of the type system.
    \item Section \ref{sec:algorithm} extends bidirectional typing algorithms \cite{bidirectional-typing} of $\lambda^{\textbf{TG}}$ \cite{coverage}, enabling properties such as subtyping to be verified in a structural, non-semantics-dependent manner.
    \item Section \ref{sec:discussion} summarizes the main contributions of this paper. Finally, Section \ref{sec:conclusion} discusses potential future developments of this work.
\end{enumerate}

\section{Overview}\label{sec:overview}
Before formally presenting our work and the main results, we start with an informal, example-driven prologue. This will help illustrate the intuitive ideas behind the key constructs we are about to introduce.

\paragraph{History Expressions}
Usually, a function, whether local or remote, has a signature of the following form:
\begin{equation}
    f: \tau_1 \longrightarrow \tau_2
\end{equation}
Where $\tau_1$ represents the type of the formal parameters of $f$, while $\tau_2$ is the return type. When we refer to \emph{local} functions, we naturally have access to the function’s body, allowing us to directly infer its type. However, if $f$ is a \emph{remote} service, the body is immaterial and the type cannot be inferred by us. Instead, it is provided by the (API's) provider, and we assume that the information about the function’s type given by the provider is reliable and trusted.

Our goal here is to incorporate additional information about a function (or, more generally, an expression) that are related to how the resources within it are used.
The added information will be statically encoded as a History Expression denoted by $H$. 
The type associated with a function, as shown below,
\begin{equation}
    f: \tau_1 \overset{H}{\longrightarrow} \tau_2
\end{equation}
includes $H$, which is referred to as the latent effect of the function. In other words, we are working within a type and effect system \cite{types-effects-general,static3-types-effects}.

History Expressions were originally introduced in \cite{history} to handle the secure orchestration of properties of functional services and to manage secure resource usage. Our formulation of History Expressions differs from the original version presented in \cite{history}, as we incorporate the use of type qualifiers within them. Type qualifiers are first-order logic formulas that provide an abstract description of the values associated with a type. The concept of type qualifiers in History Expressions is derived from the notion of Refinement Types \cite{refinement}, which typically allow a type qualifier to be paired with a basic type (such as $int$, $bool$, $int\;list$, etc.) to convey additional information about the type's values. Coverage Types \cite{coverage} offer a specific instance of Refinement Types. Note that the introduction of type qualifiers enables more compact representations of histories. Furthermore, because History Expressions capture the effects of a function's behaviour on its resources, it is essential to link these effects on resources with the values returned by the function (the return type).

We have already pointed out that History Expressions provide an over-approximation of the actual behaviour of a function with respect to resources. Therefore, the introduction of type qualifiers allows for a more refined, lower over-approximation, enabling a more precise "skimming" or selection of specific values within a type to be considered.

\paragraph{Coverage Types} 
Coverage Types \cite{coverage} are built on the foundation of Refinement Types \cite{refinement}, which enable the creation of custom types by associating qualifiers with base types. These qualifiers provide an over-approximation of the values that an expression, linked to the type, may return.
Coverage Types were introduced to address the challenges of Property-Based Testing (PBT), where suitable auxiliary functions, known as input test generators, are used to generate input cases that satisfy specific properties (usually constraints on the program's inputs). 

The output of these generators is then fed as input into the program undergoing PBT. In a testing scenario, it is desirable to cover as many input cases as possible, including edge cases and rare scenarios. A key requirement is the ability to statically verify the completeness of an input test generator: it must not only produce correct input cases but also potentially \emph{all} relevant cases.

Rather than focusing on proving the correctness of a program, Coverage Types shift the emphasis towards proving the completeness of the generator functions in relation to the data they handle. Coverage Types achieve this by employing an under-approximation approach. The type associated with an expression provides a guarantee about the values it returns. Specifically, it ensures that all values satisfying the predicate associated with the type will be generated in at least one execution of the expression. It is important to note that the predicate does not necessarily represent all values returned by each execution. However, we are guaranteed that all values satisfying the predicate are covered, in line with the principles of under-approximation logic \cite{IL1}.

Here, we introduce \emph{Extended Coverage Types}, which are characterized by the inclusion of two type qualifiers. The first qualifier corresponds to the standard Coverage Type, while the second provides additional information that specifies \emph{exactly} the values returned by the expression: both \emph{all and only} those values. This extension stems from the need to incorporate \emph{correctness} in resource usage, which is tracked through History Expressions. From a technical standpoint, Extended Coverage Types involves handling under-approximation within an otherwise over-approximated type.

Our goal is to statically verify the correct usage of resources through a type that directly corresponds to a History Expression. In a type system that follows the idea of over-approximation for the return type, this would pose no issues; it would simply require pairing the return type with the behavioural type.

\begin{wrapfigure}{r}{0.4\textwidth}
    \centering
    \begin{lstlisting}
match bool_gen () with
| 0 -> 'a'
| 1 -> 'b'
    \end{lstlisting}
    \caption{}
    \label{fig:exmp2}
\end{wrapfigure}

In Coverage Types, however, the type represents the guarantee of coverage for certain values, while other values may exist along paths that were disregarded during the typing process. For example, consider the example in Figure \ref{fig:exmp2}, the \verb|patter-matching| can be typed by the Coverage Type:

\begin{equation}
    \under{v: char \;|\; v = \text{'a'}}
\end{equation}

The issue arises from excluding the path where the expression reduces to $b$. If this value is later used in an action or an external function call recorded in a History Expression, the correctness property of the type system would be compromised. This is because the expected history would no longer provide a valid over-approximation.

We address this issue by introducing a new Coverage Type, enhanced with a qualifier that ensures not only completeness (coverage) but also correctness. While this might seem like a significant change that could disrupt the type system and its denotations, this is not the case. The only rule affected is the one governing the subtyping relationship, as that is where paths are discarded.

To keep consistency, in this new type—referred to as the \textbf{Extended Coverage Type} (ECT), the second qualifier in the sub-typing rule must remain semantically identical. As we will demonstrate, if an expression is typed according to the type system rules of $\lambda^{\textbf{TG}}$ \cite{coverage} without applying the sub-typing rule, the type qualifier will precisely represent the values returned at runtime.

For instance, the \emph{ECT} associated with the expression of the previous example is:

\begin{equation}
    \under{v: char \;|\; v = \text{'a'} \;|\; v = \text{'a'} \;\lor\; v = \text{'b'}}
\end{equation}

\vspace{0pt}
\begin{wrapfigure}{l}{0.60\textwidth}
    \begin{lstlisting}
let f = ?\typeact{new}$_{file}$? ()
in
let _ = ?\typeact{open}? f
in
let rec g (n: int) (f: ?\typekw{file}?) : int list =
    if n < 0 then [] else (
        if bool_gen () then [] else (
            let x = int_range (0, n)
            in
            let _ = ?\typeact{write}? f x
            in
            x :: (g (x - 1) f)
        )
    )
in
g 10 f
    \end{lstlisting}
    \caption{Code that creates and opens a new file, and writes in a series of randomly generated numbers between 0 and 10 including extremes.}
    \label{fig:exmp1}
\end{wrapfigure}

We now consider a more expressive example. The code in Figure \ref{fig:exmp1} uses OCaml-like syntax and begins by creating and opening a new file. After that, a recursive function, \verb|g|, is defined. This function takes two parameters: an integer \verb|n| and a file \verb|f|, and writes to the file the number generated by the primitive \verb|int_range|. The \verb|int_range| function accepts a pair of numbers and randomly generates, following a uniform probability distribution, a number between the two given values. The generated number is then concatenated with the results of subsequent recursive calls, until either \verb|n| becomes less than \verb|0|, or early termination occurs via the \verb|bool_gen| primitive, which randomly generates a boolean value.This block of code produces \emph{all} lists in descending order without repeated elements, using numbers between \verb|0| and \verb|10|.

To ensure completeness in generation, early termination is essential; otherwise, lists that satisfy the criteria, such as \verb|[10]| and \verb|[7, 5]|, would never be generated. Finally, the function is applied to \verb|10| and the newly created file.

The return type will obviously be a \verb|int list| type. What we are now interested in is determining the program's behaviour with respect to resource usage. The History Expression, the latent effect, that our type system will be able to associate with the code on the left is as follows:

\begin{equation}
    \begin{gathered}
        new_{file}(X) \cdot open(X) \cdot \\
        \cdot \mu G\bluet{(}n{:}(int{:}v \geq 0), \;f{:}(file{:}\top_{file})\bluet{)}\oranget{(}\epsilon + write(v = f, \; v \geq 0 \;\land\; v \leq n) \cdot \\
        \cdot G\bluet{(}n{:}(\exists x. x \geq 0 \;\land\; v \leq n \;\land\; v = x - 1), \;f{:}(v = f)\bluet{)}\oranget{)} \cdot G(v = 10, \;v = X)
    \end{gathered}
\end{equation}

The effect described above is an over-approximation of the behaviour of the code in Figure 
\ref{fig:exmp1}. First, a file is created and associated with an arbitrary identifier, $X$. Then, the file is opened, and the latent effect of the recursive function \verb|g| is defined. This is achieved using recursion (via the $\mu$ construct), allowing us to express the following:

\begin{enumerate}
    \item The identifier $G$ to be used for recursive calls.
    \item The list of parameters that the function takes as input, with associated types.
    \item The latent effect of the recursive function includes the possibility of either terminating the recursion through the special empty action $\epsilon$, or writing the generated number $x$ to the file $f$. Since $x$ can be any number between $0$ and $n$, we achieve an over-approximation by incorporating the qualifier associated with the variable $x$ into the \verb|write| action. The recursive call, $G$, takes the same file and $x - 1$ as parameters. Because the exact value of $x$ is unknown, the qualifier has the role of representing all the possible values that can be obtained, meaning all values that satisfy the predicate.
\end{enumerate}

As expected, the application of the recursive function $G$ takes place at the end of the history, since the $\mu$ primitive defines a recursive function but does not apply it. The appearance of the identifier $G$ with its associated parameters, upon which the function is invoked, marks the starting point for the unfolding of the recursion.

The Coverage Type associated with the previous code fragment in Figure \ref{fig:exmp1} is as follows:
\begin{equation}
    \begin{gathered}
        \under{v: int\;list \;|\; \forall u, mem(v, u) \Longrightarrow 0 \leq u \leq 10 \;\land\; dec\_sorted(v) \;\land\; \\
        \neg(\exists w, \forall u, mem(v, u) \Longleftrightarrow mem(w, u) \;\land\; shorter(w, v))}
    \end{gathered}
\end{equation}

We now comment on the Coverage Type above. Following \cite{coverage}, type qualifiers are enriched with \emph{method predicates} such as $mem$, $dec\_sorted$ and $shorter$, namely certain \emph{uninterpreted functions} of which we only know the name and arity. The predicate $mem(v, u)$ returns \verb|true| if the element $u$ belongs to the data structure $v$; $dec\_sorted$ checks whether the structure is ordered in the descending direction; whereas $shorter$ returns \verb|true| if the length in terms of elements of the first structure is less than the second.

Intuitively, all values that satisfy the predicate are returned by the program. In fact, it is clear that any list meeting the constraints outlined below will certainly be generated in at least one execution of the code in Figure \ref{fig:exmp1}:
\begin{enumerate}
    \item All items in the list are between 0 and 10;
    \item The list must be ordered in descending order;
    \item There must not exist a list $w$ that contains exactly the same elements as the list in question but is shorter in length. If such a list exists, it would mean that the current list contains at least one duplicate. However, as previously mentioned, the code does not generate lists with repeated elements.
\end{enumerate}

A main issue with Coverage Types, however, is their under-approximation flavour. This allows certain paths that are actually reachable to be dropped, precisely because they are defined as an under-approximation of the values returned by the program.

For instance, the code in Figure \ref{fig:exmp1} can also be associated with the Coverage Types:
\begin{equation}
    \begin{gathered}
        \under{v: int\;list \;|\; \forall u, mem(v, u) \Longrightarrow u = 5 \;\lor\; v = 9 \;\land\; dec\_sorted(v) \;\land\; \\
        \neg(\exists w, \forall u, mem(v, u) \Longleftrightarrow mem(w, u) \;\land\; shorter(w, v))}
    \end{gathered}
\end{equation}
\begin{equation}
    \begin{gathered}
        \under{v: int\;list \;|\; \forall u, mem(v, u) \Longrightarrow 0 \leq u \leq 7 \;\land\; dec\_sorted(v) \;\land\; \\
        \neg(\exists w, \forall u, mem(v, u) \Longleftrightarrow mem(w, u) \;\land\; shorter(w, v))}
    \end{gathered}
\end{equation}
Note that the lists satisfying the predicate of the first type above are only \verb|[]|, \verb|[5]|, \verb|[9]| and \verb|[9, 5]|. These lists are guaranteed to be generated in at least one execution of the code, and therefore, the type provides a correct under-approximation of the program's behaviour. Indeed, according to the structural properties of Coverage Types \cite{coverage}, the following subtype relation applies:
\begin{equation}
    \begin{gathered}
        \under{v: int\;list \;|\; \forall u, mem(v, u) \Longrightarrow 0 \leq u \leq 10 \;\land\; dec\_sorted(v) \;\land\; \\
        \neg(\exists w, \forall u, mem(v, u) \Longleftrightarrow mem(w, u) \;\land\; shorter(w, v))} \\
        <: \\
        \under{v: int\;list \;|\; \forall u, mem(v, u) \Longrightarrow u = 5 \;\lor\; v = 9 \;\land\; dec\_sorted(v) \;\land\; \\
        \neg(\exists w, \forall u, mem(v, u) \Longleftrightarrow mem(w, u) \;\land\; shorter(w, v))}
    \end{gathered}
\end{equation}
However, if we delete paths that are nevertheless valid, the linkage with the over-approximation provided by the effect, is no longer valid! This is why we introduce Extended Coverage Types which essentially adds an additional qualifier where no path has been dropped, and this qualifier will be used to bind variables within History Expressions.

\section{History Expressions}\label{sec:history}
In this section, we introduce History Expressions and outline their semantics.

\begin{definition}[History Expression]
    A \textbf{\emph{History Expression}}, denoted $H$, is defined by the following grammar:
    \begin{equation}
        \begin{split}
            H, H' ::=& \;Val_H \;|\; Exp_H \;|\; H \cdot H' \;|\; H + H' \\
            Exp_{H} ::=& \;\alpha(\overline{b{:}\phi}) \;|\; F(\overline{a{:}(b{:}\phi})) \;|\; call(\phi; \;\overline{a{:}(b{:}\psi})) \;|\; \mu F(\overline{a{:}(b{:}\phi}))(H_F) \\
            Val_{H} ::=& \;\epsilon \;|\; \alpha(\overline{v}) \;|\; F(\overline{v}) \;|\; new_r(X) \;|\; get(F)
        \end{split}
        \label{eq:hist_gramm}
    \end{equation}
\end{definition}

The notation $\alpha(\overline{b{:}\phi})$ represents the application of the action named $\alpha$ on a resource, operating over a set of values with base type $b$, specifically those that satisfy the predicate $\phi$. A similar notation is used for the application of an external function named $F$. In this case, each parameter (represented by the base type-qualifier pair) is linked to the corresponding argument in the function signature.

We first comment on expressions. The expression $call$ represents the operation of invoking external functions. However, this expression supports a form of call by property, meaning that the function being called is one that satisfies the predicate $\phi$. Lastly, the $\mu$ notation denotes the declaration of a recursive function, with $H_F$ representing its latent effect.

We move move to values. Values consists of the following:
    \begin{itemize}
    \item The empty history $\epsilon$, indicating the occurrence of no relevant action;
    \item The application of an external function on a set of values; 
    \item The creation of a new resource of type $r$ with the identifier $X$ bound to it; 
    \item The value $get(F)$, which is a constant bound to the API described by the name $F$.
\end{itemize}

Intuitively, the inclusion of type qualifiers in actions and API calls allows for a more refined over-approximation. This topic will be explored further when we present the type system in Section \ref{sec:type_system}.

We can introduce the substitution operator $[\cdot/\cdot]$ on identifiers in History Expressions. This will be mainly useful in local \emph{non} recursive functions, as in these the creation of resources will be allowed, but when the latent effect of a function becomes active it will be necessary to change all the identifiers of the resources created within them, because if there were another call to that same function, the same identifiers would be reused!

\newpage

The substitution $[\cdot/\cdot]$ is then defined as follows:

\begin{equation}
    \begin{split}
        (H \cdot H')[Y/X] =&\; H[Y/X] \cdot H'[Y/X] \\
        (H + H')[Y/X] =&\; H[Y/X] + H'[Y/X] \\
        \alpha(\overline{b_i{:}\phi_i})[Y/X] = &\; \alpha(\overline{b_i{:}\phi_i[Y/X]}) \\
        call(\phi; \;\overline{a_i{:}(b_i{:}\psi_i)})[Y/X] = &\; call(\phi[Y/X]; \;\overline{a_i{:}(b_i{:}\psi_i[Y/X])}) \\
        (\mu F(\overline{a_i{:}({b_i{:}\phi_i})})(H_F) =& \;\mu(F[Y/X])(\overline{a_i{:}({b_i{:}\phi_i[Y/X]})})(H_F[Y/X]) \\
        F(\overline{a_i{:}(b_i{:}\phi_i}))[Y/X] =& \; F[Y/X](\overline{a_i{:}(b_i{:}\phi_i[Y/X]})) \\
        \epsilon[Y/X] =& \;\epsilon \\
        new_r(Z)[Y/X] = &\; new_r(Z[Y/X]) \\
        \alpha(\overline{v})[Y/X] =& \;\alpha(\overline{v[Y/X]}) \\
        F(\overline{v})[Y/X] =& \;F[Y/X](\overline{v[Y/X]}) \\
        get(F)[Y/X] =& \;get(F[Y/X])
    \end{split}
    \label{eq:subst}
\end{equation}

The $\mu$ construct serves as a variable binder for recursive function identifiers, hence this requires defining the identity of History Expressions up-to alpha conversion.

\subsection{$\alpha$-conversion History Expression}\label{subsec:alpha}

We can define the following two rules for the $\alpha$-conversion of History Expressions:
\begin{enumerate}
    \item $H' \cdot \mu F(\overline{a{:}(b{:}\phi)})(H_F) \cdot H'' \;=_{\alpha}\; H' \cdot \mu G(\overline{a{:}(b{:}\phi)})(H_F[G/F]) \cdot H''[G/F]$ \\
    $\text{if} \quad H' \cdot \mu F(\overline{a{:}(b{:}\phi)})(H_F) \cdot H''$ is the longest History Expression containing $\mu F$ and \\
    $G \notin \textsc{Bound}(H') \cup \textsc{Bound}(H_F) \cup \textsc{Bound}(H'')$ where $\textsc{Bound}(H)$ represents all the identifiers used within $H$.
    \item $H' \cdot \mu F(\overline{a{:}(b{:}\phi)})(H_{rec}) \cdot H'' \cdot H''' \;=_{\alpha}\;  H' \cdot H'' \cdot \mu F(\overline{a{:}(b{:}\phi)})(H_{rec}) \cdot H'''$ \\
    $\text{if} \quad \forall \;call(\psi; \;\overline{c{:}(b{:}\theta)}) \in H'', \; \neg \psi[v \mapsto F] \;\land\; \forall\; G(\overline{d{:}(b{:}{\sigma})}) \in H'', \; G \neq F$
\end{enumerate}

Rule (1) allows the identifier of a recursive function to be changed if it has not been used globally in the history in which it is contained.
Rule (2), on the other hand, permits a recursive function declaration to be shifted to the right, provided that the subsequent history ($H''$) contains no API type calls to the function.

\subsection{Equality Relation}
We introduce an equational theory of History Expressions. The equality relation $=$ over History Expressions is the least congruence including $\alpha$-conversion such that:

\begin{figure}[ht]
    \begin{equation*}
        \begin{gathered}
            \epsilon \cdot H = H = H \cdot \epsilon \quad \quad H = H + H \quad \quad H + H' = H' + H \\
            H \cdot (H' \cdot H'') = (H \cdot H') \cdot H'' \quad \quad H + (H' + H'') = (H + H') + H'' \\
            H \cdot (H' + H'') = (H \cdot H') + (H \cdot H'') \quad \quad (H' + H'') \cdot H = (H' \cdot H) + (H'' \cdot H)
        \end{gathered}
    \end{equation*}
    \caption{Equality Axioms of History Expressions}
    \label{fig:hist_eq}
\end{figure}

\subsection{Denotation}

We now focus on the denotation of History Expressions. We start by noting that the management of resource creations within recursive functions and APIs will not be taken into account in the definition of denotation. These constraints will be introduced in Section \ref{sec:type_system}, where the type system is presented, but we will discuss the motivations at the end of this paragraph.

We first introduce an auxiliary function that allows a variable to be bound within the qualifiers of the history. The function, \bind, is inductively defined on the structure of $H$, and its definition is shown in Figure \ref{eq:bind}.

With $\Phi(\tau_x)$ we indicate that only the qualifier of the passed type is taken. In the definition of this function, we note that the discourse about type independence made earlier applies: when we bind a variable in the history, we do not care if the type associated with it is an under- or over-approximation, but only the qualifier is taken into account.

\begin{figure}[H]
    \begin{equation*}
        \begin{split}
            \bind(x{:}\tau_x, \;H \cdot H') =& \; \bind(x{:}\tau_x, \;H) \cdot \bind(x{:}\tau_x, \;H') \\
            \bind(x{:}\tau_x, \;H + H') =& \; \bind(x{:}\tau_x, \;H) + \bind(x{:}\tau_x, \;H') \\
            \bind(x{:}\tau_x, \;\alpha(\overline{b_i{:}\phi_i})) =& \; \alpha(\overline{b_i{:}\exists x. \Phi(\tau_x)[v \mapsto x] \;\land\; \phi_i}) \\
            \bind(x{:}\tau_x, \;call(\phi; \;\overline{a_i{:}(b_i{:}\psi_i}))) =& \; call(\exists x. \Phi(\tau_x) \;\land\; \\
            & \phi; \;\overline{a_i{:}(b_i{:}\exists x. \Phi(\tau_x)[v \mapsto x] \;\land\; \psi_i})) \\
            \bind(x{:}\tau_x, \;\mu F(\overline{a_i{:}({b_i{:}\phi_i})})(H_F)) =& \;\mu F(\overline{a_i{:}({b_i{:}\exists x. \Phi(\tau_x)[v \mapsto x] \;\land\; \phi_i})}) \\
            & (\bind(x{:}\tau_x, \;H_F)) \\
            \bind(x{:}\tau_x, \;F(\overline{a_i{:}(b_i{:}\phi_i)})) =& \;F(\overline{a_i{:}(b_i{:}\exists x.\phi_i)}) \\
            \bind(x{:}\tau_x, \;\epsilon) =& \;\epsilon \\
            \bind(x{:}\tau_x, \;new_r(X)) =& \; new_r(X) \\
            \bind(x{:}\tau_x, \;\alpha(\overline{v})) =& \;\alpha(\overline{v}) \\
            \bind(x{:}\tau_x, \;F(\overline{v_i})) =& \;F(\overline{v_i}) \\
            \bind(x{:}\tau_x, \;get(F)) =& \;get(F)
        \end{split}
        \end{equation*}
    \caption{Definition of function \bind\ which aims to bind a variable within the qualifiers in the history.}
    \label{eq:bind}
\end{figure}

Furthermore, it is possible to execute the \bind\ function on several types by introducing the following equality:

\begin{equation}
    \bind(x_1{:}\tau_1, \dots, x_n{:}\tau_n, \;H) \doteq \bind(x_1{:}\tau_1, \dots, x_{n-1}{:}\tau_{n-1}, \;\bind(x_n{:}\tau_n, \;H))
\end{equation}

Below we are going to present a number of definitions and notations that will be much used throughout the presentation of the paper.

\begin{notation}[$Rid$]
    With the set $Rid$, we are going to denote the set of all possible resource identifiers, both local and remote.
\end{notation}

\begin{definition}[Resource Context]
    The \emph{resource context} is defined as follows:
    \begin{equation*}
        \Delta \subset Rid \cup (Rid \mapsto \tau)
    \end{equation*}
    This set will play two roles:
    \begin{enumerate*}[label=(\roman*)]
        \item makes it possible to keep track of resource identifiers already in use, that is, those that are already associated with resources that have already been created, so that when a new resource is created (through the \verb|new|), it has a new identifier;
        \item the second function is to contain the association between the resource identifiers used to represent the API, and the signature, that is, the $\tau$ type of the external function.
    \end{enumerate*}
    Obviously, in normal contexts, at first $\Delta$ is populated only by the mapping $Rid \mapsto \tau$ for each API that will be available within the program.
\end{definition}

\begin{notation}[$\eta$]
    With the symbol $\eta$, we shall indicate a history that is no longer reducible, which can only be composed of concatenations of the values $Val_H$ presented in \ref{eq:hist_gramm}.
\end{notation}

\begin{notation}[$\uparrow$]
    The notation $\uparrow$, associated with a History Expression $H$, allows each $new_r(X)$ present in $H$ to be assigned the symbol $\uparrow$: this is necessary because when calculating the denotation of $H$ we will have to take into account, in a set, the identifiers of the resources on which the \verb|new| was actually called, and since this in our language is a terminal value, we use this special symbol to indicate that we must first add the identifier to the set and then stop with the reduction.
    For example, if $H$ were equal to:
    \begin{equation}
        (new_r(X) + new_r(Y)) \cdot \alpha(r: v = X \;\lor\; v = Y)
    \end{equation}
    Applying $\uparrow$ to $H$ would return:
    \begin{equation}
        (new_r(X)^{\uparrow} + new_r(Y)^{\uparrow}) \cdot \alpha(r: v = X \;\lor\; v = Y)
    \end{equation}
    In this way, having to choose between the creation of $X$ and that of $Y$, assuming we choose the former, the identifier $X$ will be inserted in our set, and then during the reduction of the expression $\alpha$ we will know that only $X$ and not $Y$ will have to be taken into account, thus generating only the value $\alpha(X)$.

    The symbol $\uparrow$ is also associated with any expression indicating a call of an external function, of the form $F(\overline{a_i{:}(b_i{:}\phi_i}))$: in this case, however, the semantics of the symbol will indicate that the latent effect of $F$ has not yet been inserted into the main history. 
    
    The application of the $\uparrow$ symbol does not take place within the effects of the recursive functions linked to the $\mu$ construct: $\uparrow$ will be applied each time the effect becomes active.
\end{notation}

\begin{definition}[History Expression's Denotation without a Context]\label{def:history_den}
    The denotation $\llbracket \cdot \rrbracket$ of an \emph{History Expression} $H$, in a context of resources $\Delta$, is the set of all histories $\eta$ such that $H^{\uparrow}$ is reduced (in one or more steps) to $\eta$, and $\eta$ is terminal, i.e. it is no longer reducible.
    \begin{equation}
        \llbracket H \rrbracket = \{ \eta \;|\; (\varnothing, \;\varnothing, \;H^{\uparrow}) \rightarrow^* (\Omega, \;\Upsilon, \;\eta) \;\land\; (\Omega, \;\Upsilon, \;\eta) \nrightarrow \}
    \end{equation}
\end{definition}

We note how the reduction consists of a triple of the form:
\begin{equation}
    (\Omega, \;\Upsilon, \;H)
\end{equation}
Where $\Omega$ will denote the set in which the identifiers of the resources actually created during the calculation of a single history $\eta$ will be held. While with $\Upsilon$ we shall denote a list of substitutions specified by the notation $\{\cdot/\cdot\}$ that concern only the identifiers of external functions that actually represent recursive functions. This list will be the construct that will allow true recursion, as it will allow each identifier representing a recursive call to be substituted for the behaviour of the recursive function itself.

The relation $\rightarrow$ for $H$ is defined in Figure \ref{fig:denot_history}. Rule No. 6 makes it possible to add to the set $\Omega$ the identifier of a resource that has just been created and which may be used in actions following its creation. The two rules below constraint that in a concatenation, the whole of the term on the left must first be reduced and then the term on the right. This is necessary because histories have a property of temporal ordering between actions and, referring back to the previous rule concerning the \verb|new|, in order to be able to use a resource it will be necessary for it to have first been created and then added to the $\Omega$ set.

\begin{figure}[H]
    \begin{equation*}
        \begin{gathered}
            (\Omega, \;\Upsilon, \;\epsilon) \nrightarrow \quad \quad 
            (\Omega, \;\Upsilon, \;\alpha(\overline{v})) \nrightarrow \quad \quad
            (\Omega, \;\Upsilon, \;get(F)) \nrightarrow \quad \quad
            (\Omega, \;\Upsilon, \;F(\overline{v})) \nrightarrow \\ \\
            (\Omega, \;\Upsilon, \;new_r(X)) \nrightarrow \quad \quad
            \frac{
                (\Omega, \;\Upsilon, \;\eta) \nrightarrow \quad (\Omega, \;\Upsilon, \;\eta') \nrightarrow
            }{
                (\Omega, \;\Upsilon, \;\eta \cdot \eta') \nrightarrow
            } \\ \\
            (\Omega, \;\Upsilon, \;new_r(X)^{\uparrow}) \rightarrow (\Omega \cup \{X\}, \;\Upsilon, \;new_r(X)) \\ \\
            \frac{
                (\Omega, \;\Upsilon, \;H) \rightarrow (\Omega', \;\Upsilon', \;H'')
            }{
                (\Omega, \;\Upsilon, \;H \cdot H') \rightarrow (\Omega', \;\Upsilon', \;H'' \cdot H')
            } \quad \quad
            \frac{
                (\Omega, \;\Upsilon, \;\eta) \nrightarrow \quad (\Omega, \;\Upsilon, \;H) \rightarrow (\Omega', \;\Upsilon', \;H')
            }{
                (\Omega, \;\Upsilon, \;\eta \cdot H) \rightarrow (\Omega', \;\Upsilon', \;\eta \cdot H')
            } \\ \\
            (\Omega, \;\Upsilon, \;H + H') \rightarrow (\Omega, \;\Upsilon, \;H) \quad \quad
            (\Omega, \;\Upsilon, \;H + H') \rightarrow (\Omega, \;\Upsilon, \;H') \\ \\
            \frac{
                \begin{gathered}
                    \forall i, Val_i = \{u: b_i \;|\; \phi_i[v \mapsto u] \;\land\; (u \in Rid \Longrightarrow u \in \Omega \;\lor\; \Delta(u)\downarrow)\} \\
                    H = \begin{cases}
                        \underset{\overline{u} \in \underset{i}{\prod} Val_i}{\bigoplus} \alpha(\overline{u}) & \text{if}\;\forall i, Val_i \neq \varnothing \\
                        \epsilon & \text{otherwise}
                    \end{cases}
                \end{gathered}
            }{
                (\Omega, \;\Upsilon, \;\alpha(\overline{b_i{:}\phi_i})) \rightarrow (\Omega, \;\Upsilon, \;H)
            } \\ \\
            \frac{
                Api = \{ F \;|\; \phi[v \mapsto F] \;\land\; \Delta(F)\downarrow \} \quad H = \begin{cases}
                    (\underset{F \in Api}{\bigoplus} F(\overline{a_i{:}(b_i{:}\psi_i}))^{\uparrow}) & \text{if}\;Api \neq \varnothing \\
                    \epsilon & \text{otherwise}
                \end{cases}
            }{
                (\Omega, \;\Upsilon, \;call(\phi; \;\overline{a_i{:}(b_i{:}\psi_i))}) \rightarrow (\Omega, \;\Upsilon, \;H)
            } \\ \\
            \frac{
                \Delta(F) = \overline{\tau} \rightarrow (\tau_F, \; H_F) \quad H_F^{\star} = \bind(\overline{a_i: \under{v: b_i \;|\; \psi_i}}, \;H_F)
            }{
                (\Omega, \;\Upsilon, \;F(\overline{a_i{:}(b_i{:}\psi_i}))^{\uparrow}) \rightarrow (\Omega, \;\Upsilon, \;(F\Upsilon)(\overline{a_i{:}(b_i{:}\psi_i})) \cdot {H_F^{\star}}^{\uparrow})
            } \\ \\
            \frac{
                \begin{gathered}
                    \forall i, Val_i = \{u: b_i \;|\; \phi_i[v \mapsto u] \;\land\; (u \in Rid \Longrightarrow u \in \Omega \;\lor\; \Delta(u)\downarrow)\} \\
                    H = \begin{cases}
                        \underset{\overline{u} \in \underset{i}{\prod} Val_i}{\bigoplus} \alpha(\overline{u}) & \text{if}\;\forall i, Val_i \neq \varnothing \\
                        \epsilon & \text{otherwise}
                    \end{cases}             
                \end{gathered}
            }{
                (\Omega, \;\Upsilon, \;F(\overline{a_i{:}(b_i{:}\phi_i)})) \rightarrow (\Omega, \;\Upsilon, \;H)
            } \\ \\
            (\Omega, \;\Upsilon, \;H(\overline{a_i{:}(b_i{:}\psi_i)}) \rightarrow (\Omega, \;\Upsilon, \;{\bind(\overline{a_i{:}(b_i{:}\psi_i)}, \;H)}^{\uparrow}) \\ \\
            (\Omega, \;\Upsilon, \;\mu F(\overline{a_i{:}({b_i{:}\phi_i})})(H_F)) \rightarrow (\Omega, \;\Upsilon\{H_F/F\}, \;\epsilon)
        \end{gathered}
    \end{equation*}
    \caption{Reduction Relation for History Expressions}
    \label{fig:denot_history}
\end{figure}

After the two rules for nondeterministic choice, we present the reduction rule for invoking an action on \textbf{all} those values satisfying each predicate $\phi_i$. We note how not only are put into a set, for each argument $i$ of the action, all those values of type $b_i$ that satisfy $\phi_i$; but in the case where the value under consideration is an identifier of a resource - and thus $b_i$ will be a resource type - this either must be present in $\Omega$ - it was created by means of a \verb|new| - or in the case where the identifier represents an external function, this must be defined in $\Delta$. Finally, if each parameter $i$ has associated at least one value on which to perform the action (so each $Val_i$ is different from the empty set), we create a vector $\overline{u}$ for each possible combination between the values of each parameter $i$ and place them in nondeterministic choice. Otherwise, if there is at least one parameter that does not have any of the admissible values associated with it, it means that the action within the programme can never occur (this is guaranteed by the type qualifiers, which represent a correct over-approximation for each parameter). The same mechanism was used for APIs call reduction as can be seen in rule No. 15.

The next rule allows the calls of a set of external functions invoked on the same parameters to be unpacked into as many non-deterministic choices as possible. The APIs taken into consideration, in addition to having to satisfy the predicate $\phi$, must also be defined in the context of the resources $\Delta$.

Below, we find the rule that allows the latent effect of an external function to be active: the effect is then taken from $\Delta$ and all the parameters passed to the call are bound to it. In the case where the identifier $F$ is in fact associated with a recursive function, one (and only one) of the substitutions in $\Gamma$ will be successful (and as we will see in the type system rules in Section \ref{sec:type_system} the latent effect associated with a recursive function in $\Delta$ will be $\epsilon$).

The last two rules handle recursion. The term to be reduced in the first of the two is a direct consequence of the reduction of external functions with no active latent effect: thus $H$ will be the latent effect of a recursive function, inserted thanks to one of the reductions in $\Gamma$, and will be consequently bound to the parameters on which the recursive function was called. The last rule simply adds the substitution between the latent effect and the identifier of the recursive function declared through the $\mu$ construct to the $\Gamma$ list.

\begin{definition}[History Expression's Denotation under a Context]
    The denotation of a history $H$, taking into account already having a context $\Gamma$ to bind the non-quantified variables in the action qualifiers, is inductively defined on $\Gamma$ as follows:
    \begin{equation}
        \begin{gathered}
            {\llbracket H \rrbracket}_{\varnothing} = \llbracket H \rrbracket \\
            {\llbracket H \rrbracket}_{x: \tau_x,\Gamma} = {\llbracket \bind(H, \;x{:}\tau_x) \rrbracket}_{\Gamma}
        \end{gathered}
    \end{equation}
\end{definition}

That is, for the empty context, the denotation is the same as the context-free one defined above since the variables will all be bound; whereas, for the inductive step, if the context consists of an association $x:\tau_x$ and another set of associations contained in $\Gamma$, we evaluate the denotation in $\Gamma$, but of the history in which all the qualifiers in the actions are bound existentially to a variable $x$ of type $\tau_x$.

\begin{theorem}[Correctness of Denotation to Type Qualifiers]
    Given a History Expression $H$ and an interpretation $\mathcal{I}$ that maps each qualifier in $H$ with the value assigned by the interpretation itself, it holds that:
    \begin{equation}
        \forall \eta, \;\Gamma, \;H. \;\eta \in \llbracket H \rrbracket_{\Gamma} \Longrightarrow \exists \mathcal{I}. \;\mathcal{I} \models \Phi(H) \;\land\; \eta \in \llbracket H(\mathcal{I}) \rrbracket_{\Gamma}
    \end{equation}
    This property tells us that for each terminal history $\eta$ that belongs to a History Expression $H$, there exists at least one interpretation $\mathcal{I}$ that satisfies the qualifiers taken into account by $H$ such that, by applying the substitution within $H$, $\eta$ belongs to the denotation of $H$ with inside instead of type qualifiers the values, surely correct since they belong to a valid interpretation. Consequently, the values in $\eta$ will also be correct with respect to the type qualifiers as this belongs to a denotation (the second) which will not use the reduction rules on the calculation of values, this because there are directly within $H$. 
    
    \paragraph{Interpretation} Since the qualifiers have no variable name or explicit identifier associated with them, we will use a number as an identifier in the definition of the interpretations, which will represent the index of the position of the qualifier itself within the History Expression with respect to the others, starting from left to right.
    
    Applying an interpretation $\mathcal{I} = \{\overline{i = v_i}\}$ to a History Expression $H$ replaces each qualifier indexed by $i$ with the expression $v = v_i$. This actually contradicts what was said earlier about the rules on calculating qualifiers not being used, because in the end $v = v_i$ is still a qualifier. In fact in reality these rules will be used, but the type qualifier becomes trivial, the only value that satisfies $v = v_i$ is $v_i$ itself!
    
    Finally, it should be pointed out that the only qualifiers that will not require mapping are those encapsulated within the $\mu$ recursion construct, since as already mentioned this is not used for the semantic purposes of History Expressions, but only to verify well-formedness in the type system rules as we shall see later.
\end{theorem}

\begin{proof}
    Having used inference rules for the definition of denotation, the demonstration can proceed by induction:
    \begin{itemize}
        \item For the first five rules, the demonstration is immediate as we are already operating on values that are not altered. For the next rule (concatenation of two terminal histories) the same.
        \item The event of creation of a resource is not of interest to us, as it does not encapsulate any type qualifiers, but directly presents a value (the identifier of the new resource) that is preserved in the reduction, and is added to the set $\Omega$. This last piece of information is important for some of the following points.
        \item For the following two concatenation rules we simply use the inductive hypothesis on the reduction of the History Expressions $H$ into $H''$ and $H'$ respectively, since they are in the premises. We conclude by stating that the property on the qualifiers is respected.
        \item On the other hand, for non-deterministic choice rules, here we are simply dropping a history so within the chosen interpretation $\mathcal{I}$ for the qualifiers in it can be assigned any value, still respecting the type qualifier, but independent of the terminal history $\eta$.
        \item Let us move on to the first notable rule, which concerns the reduction of the action $\alpha$. We must first differentiate between the case in which the type of values is a resource and the case in which it is not. In the latter case, what we compute in the premises is precisely the set of \textbf{all} values satisfying each predicate, and then we place \textbf{any} of these values (actually a vector of values in the case where the action has more than one parameter, but the transition to this case is w.l.g.) in non-deterministic choice. Proceeding in the reduction we will then have to choose only one of these values $v_i$ (because of the rules on the operator $+$) and this will belong to the terminal history $\eta$ (encapsulated in the event $\alpha$). Within the interpretation $\mathcal{I}$ we can then choose as the value for the qualifier at the position of $\alpha$ just $v_i$ since it satisfies it. For resource types, when we go to compute the set of values, we simply add another condition to the qualifier $\phi$ (which allows us to have identifiers that actually exist or APIs that have been defined), only restricting the starting set.

        In the case where the value generated is $\epsilon$, it will mean that there will be at least one qualifier that represents a contradiction, so there can be no value that satisfies it: as the value for the interpretation $\mathcal{I}$ we will choose the special symbol $\bot$, which in the semantics of the interpretation allows us to skip checking the qualifier in question.
        
        \paragraph{Note} In practice, when we will integrate History Expressions with the Coverage Types type system, the variables can never have a type with a contradictory qualifier associated with them, as this would lead to a context not being in good form \cite{coverage}. On the merits, such a qualifier (e.g. trivially $v = 1 \;\land\; v = 2$) is possible to find during the typing phase (just think of the function application of a parameter that does not respect the type of the argument), but since it will represent the presence of a type-matching error, it will lead to an error by failing to infer a type. \\

        It is worth noting, as this last discussion makes us reflect on why in the theorem the implication does not also hold the other way around, i.e. that for every valid interpretation there is a terminal history that belongs to the denotation: this is not possible with resource types in mind, in fact one could take in $\mathcal{I}$ identifiers that satisfy the type qualifier, but which have not been defined or created.
        \item For the remaining rules using qualifiers, such as the one on the single API call, the demonstration is identical, as the premises are the same. For the multiple call of APIs the calculation is also similar, as done above we restrict the set of values that satisfy the qualifier, in this case with those that also represent defined APIs. Basically, as a value for each qualifier in $\mathcal{I}$ we always take the one chosen during the reduction.
        \item Finally, for the last three rules that do not perform calculations on the qualifiers, we note how these, regardless of the operations they perform, which are irrelevant for the choice of values in $\mathcal{I}$, always leave the type qualifiers unchanged by only dragging them into other constructs or sets, such as $\Omega$ for recursion, or by binding them into other histories through \bind.
    \end{itemize}
\end{proof}

\begin{note}[Creation of resources in recursive functions and APIs]
    As previously announced, there will be constraints on the latent effects of APIs and recursive functions for which the creation of any type of resource within them will not be permitted. This implementation choice was made for two reasons:
    \begin{enumerate}[label=(\roman*)]
        \item The former is mainly related to recursive functions, and the reason for this is that unfolding a recursive function that at each iteration \emph{could} create new resources is dangerous: we do not know how many times the recursive call could be invoked, statically the number of invocations is potentially infinite! Despite the fact that in the type system of $\lambda^{\textbf{TG}}$ there is a control on recursion that imposes a decreasing ordering relation between the first parameter of the current function and the one present in each recursive call \cite{coverage}, however, the resources we have at our disposal are always limited in quantity, so even if the recursive function after $n$ calls terminates, this $n$ could be sufficiently large to create, for example, a disproportionate number of files, or network connections, realising, deliberately or otherwise, a \emph{DOS} \cite{dos} attack. The same goes for APIs: due to the principle of compositionality, an API could call other APIs within it, perhaps using recursive functions, and encounter the same problem.
        Let us imagine that we have the following recursive History Expression $H$ associated with any function\footnote{In our type system this syntax is \textbf{illegal} both because in Section \ref{sec:type_system} we will see that a recursive function with this history cannot be typed, and also because if we were to calculate the denotation of this history this \textbf{not} would be a correct over-approximation of the function.}:
        \begin{equation}
            H \equiv \mu F(n{:}(v > 0))(new_{file}(X) \cdot F(n{:}(v = n - 1)))
        \end{equation}
        Although, as already mentioned, the rules guarantee that the recursion is finite - in fact $n$ descends by $1$ at each iteration and when it is less than or equal to $0$ the recursion stops - from the point of view of resources and their limited availability this represents a problem. Consider a potential call of $F$ with $n$ equal to $10,000$: although this is feasible, it would lead to the creation of $10,000$ files, which could cause problems on the memory device!
        \item The second reason, is related to a \emph{ROP} (Return Oriented Programming) \cite{rop} attack, which a potential attacker could carry out on our code. Should this be able to find an entry point in the program to inject and execute malicious code, such as through \emph{buffer overflow} \cite{buffer}, by checking the stack it can manage and change the address of the return function, subsequently creating a fairly long chain of function calls, in which for instance it always returns to the same function, which may be the one that creates one or more resources. In spite of this, in order not to limit the expressiveness of the language too much, we will still allow the creation of resources within local non-recursive functions, while for recursive functions, or for API calls, instead of creating them directly within them, it will be possible to use the \verb|new| before their invocation and pass the desired resources as parameters.
    \end{enumerate}
\end{note}

\begin{proof}[Correctness of $\alpha$-conversion rules] After having introduced the denotation of History Expressions, we can demonstrate the two rules concerning $\alpha$-conversion presented above:
    \begin{itemize}
        \item For rule (1), it is necessary to make the assumption that the identifiers associated with the recursive functions \textbf{not} can be used explicitly within programmes, but must only be used by the compiler. This constraint will actually be present in our language and will be formally defined in the well-formedness rules of the type system in Section \ref{sec:type_system}. Thus, it is not possible for these identifiers to be present within type qualifiers, except in the qualifier of the expression $call$ and consequently also as an identifier in the single API call $F(\overline{a_i{:}(b_i{:}\phi_i)})$.
        The first part of the History Expression, $H'$, is identical. The second part, the recursive construct, although the identifier is changed from $F$ to $G$, this will be reduced to $\epsilon$. But there is the side effect that adds the substitution $\{H_F/F\}$ to the list $\Omega$. So according to the denotation in Figure \ref{fig:denot_history}, if $F$ in $H''$ is a recursive function, then the substitution in $\Omega$ will be applied by replacing the identifier $F$ with the History Expressions associated with the body of the recursive function, and in any terminal history $\eta$ it will not be possible to find any reference to $F$ in the values, due to the assumption made at the beginning. Consequently it is possible to change the identifier $F$ in all occurrences in $H''$ - and in the body of $H_F$ since it can become active in $H''$ through substitutions - to another $G$ such that it is globally fresh. For $H'$, however, if it is correct, there can be no reference to $F$.
        \item For rule (2), the demonstration is simple. Since the recursive construct is reduced to $\epsilon$, we will have that:
        \begin{equation}
            H' \cdot \epsilon \cdot H'' \cdot H''' = H' \cdot H'' \cdot \epsilon \cdot H'''
        \end{equation}
        According to the rules of equality in Figure \ref{fig:hist_eq}. The only denotation that might change, however, is only that of H‘’ since in its evaluation it will now no longer have in $\Omega$ the recursive function identified by $F$. But if there are no calls to $F$ in it, this substitution will never be used.
    \end{itemize}
\end{proof}

\section{Policies}\label{sec:policies}

In this section, we explore different approaches to defining policies and introduce two widely applicable relationships.

\begin{definition}[Ordering between events]
    The relation $<_{\eta}$ defines an ordering between two events in the terminal history $\eta$, The relation is defined as follows:
    \begin{equation}
        \alpha(X) <_{\eta} \beta(Y) \Longleftrightarrow \exists \; \eta_1, \eta_2, \eta_3. \; \eta = \eta_1 \cdot \alpha(X) \cdot \eta_2 \cdot \beta(Y) \cdot \eta_3
    \end{equation}
\end{definition}

\begin{definition}[Ownership of an event]
    The ownership relation of an event in a terminal history $\eta$ is defined by:

    \begin{equation}
        \alpha(X) \in \eta \Longleftrightarrow \exists \; \eta_1, \eta_2. \; \eta = \eta_1 \cdot \alpha(X) \cdot \eta_2
    \end{equation}
\end{definition}

\begin{example}
    For example, the policy, discussed at the beginning of the introduction on reading and writing files, is defined by making use of the relationships presented above:
    
    \begin{equation}
        \begin{gathered}
            \forall \; \eta \in \llbracket H \rrbracket .\; \forall \; read(X) \in \eta .\; \exists \; open(X) . \; open(X) <_{\eta} read(X) \;\land\; \\ 
            \nexists \; close(X) \in \eta .\; \; open(X) <_{\eta} close(X) <_{\eta} read(X)
        \end{gathered}
    \end{equation}

  One can use the same pattern to represent the same policy but on write operations (\verb|write|).
\end{example}

\section{Language \& May-Types with Must-Types}\label{sec:tuple_type}
In this section, we introduce our programming language, a functional language, along with its type system, which is based on Coverage Types. Figure \ref{fig:language} illustrates the syntax of the language $\lambda^{\textbf{TG}}$ originally introduced in \cite{coverage}, which has been extended incorporating enhancements in both terms and values, as well as in types.

\begin{figure}[ht]
    \centering
    \begin{tabularx}{\textwidth}{rcX}
        \textbf{Variables} &  & $x, f, u, \dots$ \\
        \textbf{Identifiers} &  & $X, Y, Z, \dots$ \\
        \textbf{Data Constructors} & $d ::= $ & $()$ \textsf{\;\textbar\; true \;\textbar\; false \;\textbar\; O \;\textbar\; S \;\textbar\; Cons \;\textbar\; Nil \;\textbar\; Leaf \;\textbar\; Node} \\
        \textbf{Constants} & $c ::= $ & $\mathbb{B} \;|\; \mathbb{N} \;|\; \mathbb{Z} \;|\; \dots \;|\; d\;\overline{c}$ \\
        \textbf{Operators} & $op ::= $ & $d$ \;\textbar\; $+$ \;\textbar\; $==$ \;\textbar\; $<$ \;\textbar\; \textsf{mod} \;\textbar\; \textsf{nat\_gen} \;\textbar\; \textsf{int\_gen} \;\textbar\; $\dots$ \\
        \textbf{Actions} & $\alpha ::= $ & $read$ \;\textbar\; $write$ \;\textbar\; $open$ \;\textbar\; $close$ \;\textbar\; $\dots$ \\
        \textbf{Values} & $v ::= $ & $c \;|\; op \;|\; x \;|\; \lambda x{:}t{.}e \;|\; \textsf{fix}\;f{:}t{.}\lambda x{:}t{.}e$ \\
        \textbf{Terms} & $e ::= $ & $v$ \;\textbar\; \textsf{err} \;\textbar\; $\texttt{\small{let}} \; x = e \; \texttt{\small{in}} \; e$ \;\textbar\; $\texttt{\small{let}} \; x = op\;\overline{v} \; \texttt{\small{in}} \; e$ \;\textbar\; $\texttt{\small{let}} \; x = v\;v \; \texttt{\small{in}} \; e$ \;\textbar\; $\texttt{\small{match}}\;v\;\texttt{\small{with}}\;\overline{d\;\overline{y} \rightarrow e}$ \;\textbar\; $\texttt{\small{let}}\;x = \alpha \; \overline{v}\;\texttt{\small{in}}\;e \;|\; \texttt{\small{let}}\;x = \texttt{\small{new}}_r \; ()\;\texttt{\small{in}}\;e \;|\; \texttt{\small{let}}\;x = \texttt{\small{get}}\; F \;\texttt{\small{in}}\;e \;|\; \texttt{\small{let}}\;x = f\;\overline{v}\;\texttt{\small{in}}\;e$ \\
        \textbf{Resource Types} & $r ::= $ & $file$ \;\textbar\; $socket$ \;\textbar\; $api$ \;\textbar\; $\dots$ \\
        \textbf{Base Types} & $b ::= $ & $r$ \;\textbar\; $unit$ \;\textbar\; $bool$ \;\textbar\; $nat$ \;\textbar\; $int$ \;\textbar\; $b\;list$ \;\textbar\; $b\;tree$ \;\textbar\; $\dots$ \\
        \textbf{Basic Types} & $t ::= $ & $b$ \;\textbar\; $t \rightarrow t$ \\
        \textbf{Method Predicates} & $mp ::= $ & $emp$ \;\textbar\; $hd$ \;\textbar\; $mem$ \;\textbar\; $\dots$ \\
        \textbf{Literals} & $l :: = $ & $c$ \;\textbar\; $x$ \\
        \textbf{Propositions} & $\phi ::= $ & $l$ \;\textbar\; $\bot$ \;\textbar\; $\top_b$ \;\textbar\; $op(\overline{l})$ \;\textbar\; $mp(\overline{x})$ \;\textbar\; $\neg \phi$ \;\textbar\; $\phi \;\land\; \phi$ \;\textbar\; $\phi \;\lor\; \phi$ \;\textbar\; $\phi \Longrightarrow \phi$ \;\textbar\; $\forall u{:}b . \phi$ \;\textbar\; $\exists u{:}b . \phi$ \\
        \textbf{Refinement Types} & $\tau ::= $ & $\under{v: b \;|\; \phi \;|\; \phi}$ \;\textbar\; $\overa{v: b \;|\; \phi}$ \;\textbar\; $x{:}\tau \rightarrow \kappa$ \\
        \textbf{Function Types} & $\kappa ::= $ & $x{:}\tau \rightarrow \kappa \;|\; \pi$ \\
        \textbf{History Types} & $\pi ::= $ & ($\tau, \;H)$ \\
        \textbf{Type Contexts} & $\Gamma ::= $ & $\varnothing$ \;\textbar\; $\Gamma,x{:}\tau$\\
    \end{tabularx}
    \caption{Syntax of $\lambda^{\textbf{TG}}$ extended for resource use.}
    \label{fig:language}
\end{figure}

Identifiers have been introduced as elements of the the set $Rid$, which is an infinite, countable set (although resources are always finite in any given context). Actions, represented by the symbol $\alpha$, are essentially functions. This means that their types must be predefined. It is crucial that the sets of identifiers and actions remain disjoint, as the inferable history in the type system rules will differ depending on the case.

The set of terms includes new forms of \verb|let-bindings|:
\begin{enumerate*}[label=(\roman*)]
    \item the application of an action
    \item the creation of a new resource
    \item the association to a variable of an external function
    \item the ability to apply a set of values to a variable. This is necessary for the application of APIs after they have been assigned to a variable using the keyword \verb|get|. 
\end{enumerate*}

Also types have been enriched. First, types for resources denoted by $r$ were added to the base types. In the refinement types, the return type of an arrow type has been modified to a function type. $\kappa$. This either recursively enforces the addition of an argument to the function or supports termination inserting the return type, which is now a pair type called a history type, denoted by $\pi$.
Here, $H$ represents the function's latent effect, and $\pi$ becomes the type that must be associated with each expression.

We introduce a number of properties and demonstrations on this extension. First we would like to demonstrate that if when typing an expression the rule \textsc{TSub} is not used then the associated Coverage Type will represent \textbf{\emph{all and only}} the values returned by the expression.

\begin{lemma}
    In the function \emph{Ty}, which associates a constant and an operator with its own type - and used in the rules \textsc{TConst} and \textsc{TOp} - the Coverage Type represents not only some of the values to which the expression reduces, but also represents all of them.
    \label{lemma:ty}
\end{lemma}

\begin{proof}
    The demonstration is divided for both cases:
    \begin{itemize}
        \item A constant $c$ will obviously be a value, so the Coverage Type associated with it will be of the form: $\under{v: b \;|\; v = c}$ for some base type $b$. Trivially, the qualifier represents all and only the values returned by the expression $c$, which being a value will only be $c$ itself.
        \item The type associated with an operator will be an arrow type. Again, the return type will correspond exactly to all and only the values returned by the function. In fact, the Coverage Type corresponding to the return type of an operator can be defined in two ways:
        \begin{itemize}
            \item Or by using the operator within the qualifier, e.g. $\text{Ty}(+) = a{:}\overa{v: int \;|\; \top_{int}} \rightarrow b{:}\overa{v: int \;|\; \top_{int}} \rightarrow \under{v: int \;|\; v = a + b}$.
            \item Or inductively on the base type structure of the returned value, e.g. for binary trees we define the type for leaves: $\text{Ty}(Leaf) = \under{v: b\;tree \;|\; emp(v)}$; and for nodes: $\text{Ty}(Node) = r{:}\overa{v: b \;|\; \top_{b}} \rightarrow l{:}\overa{v: b\;tree \;|\; \top_{b\;tree}} \rightarrow l{:}\overa{v: b\;tree \;|\; \top_{b\;tree}} \rightarrow \under{v: int \;|\; root(v, \;r) \;\land\; left(v, \;l) \;\land\; right(v, \;r)}$.
        \end{itemize}
        In conclusion, we can state that the return type of operators and constructors also depends exactly on the parameters passed, and the values satisfying its qualifier are \emph{\textbf{all and only}} those returned.
    \end{itemize}
\end{proof}

\begin{theorem}
    Typing an expression $e$ with a Coverage Type $\tau$, according to the type system rules of $\lambda^{\textbf{TG}}$, \emph{without the use of the \textsc{TSub} rule} and \emph{taking each branch into account when typing a} \verb|pattern-matching|; either leads to an error situation (even though using all the rules it might be possible to associate $\tau$ with $e$), or if the typing phase ends successfully, it means that $\tau$ not only represents some values that will surely be returned by $e$, but also describes \textbf{all} them.

    The same applies to the typing of a function $f$ with an arrow type $\overline{\tau} \rightarrow \under{v: b \;|\; \phi}$: in this case the property is shifted from the entire function to just the return type.
\end{theorem}

\begin{proof}
    The proof is presented inductively on the rule structure of the type system of $\lambda^{\textbf{TG}}$ \cite{coverage}:
    \begin{itemize}
        \item For the rules \textsc{TConst} and \textsc{TOp}, the demonstration is straightforward as it has already been proven in the previous lemma that the function \emph{Ty} satisfies the desired property. Axiomatically, we state that the rule \textsc{TErr} also conforms to this property since the qualifier $\bot$ can only be associated with divergent computations.
        \item For the rules \textsc{TVarBase}, \textsc{TVarFun}, \textsc{TFun} and \textsc{TFix}, the proof is straightforward since if the property holds for the expression in the premises it will also hold for the expression in the consequences.
        \item For the rule \textsc{TEq}, if the property holds for the type $\tau$, it will also hold for $\tau'$. Because it holds that the first is a subtype of the other, and vice versa, so the qualifiers of the two types are semantically equal, and denote the same set of expressions. Nothing is proved for the rule $\textsc{TSub}$ because it is not used in the assumptions of the theorem.
        \item For the rule \textsc{TApp}, if the qualifiers in the type of $a$ and $v_2$ match without using the rule \textsc{TSub}, the qualifier will represent all and only the values returned by $v_2$. Since we assume that the property is valid for the function type of $v_1$, substitution, for \emph{composition}, will allow the property to be unaltered for $\tau_x$. The last step is to assume that this also holds for $e$ and the type $\tau$, and under these conditions it will also hold. The same reasoning can be done for the rules \textsc{TAppOp}, \textsc{TAppFun} and \textsc{TLetE}.
        \item Let us therefore assume that the \textsc{TMerge} rule is cascaded over \textbf{all} the branches of a \verb|pattern-matching|, typed individually via the \textsc{TMatch} rule. The latter is the only rule for which the property will not apply. But it cannot be the first one to be applied because under the assumptions of the theorem it is not allowed; consequently, to type a \verb|pattern-matching| it will be necessary to start from the rule \textsc{TMerge}, and having to apply it on all the branches, the types $\tau_1$ and $\tau_2$ associated with each $e$ will not respect the property of covering all and only the values returned, but will cover the branches taken into consideration. Since the cascading rule must therefore be applied to all branches, the \verb|pattern-matching| will be typed with a Coverage Type whose qualifier will indicate \textbf{\emph{all and only}} the values it returns.
    \end{itemize}

    The proof for the typing of a function is similar, it is only necessary to first apply a series of \textsc{TFun} to populate the context with the parameters of the function ($\overline{\tau_i}$) and then proceed in an identical manner to the above with the body of the function.
\end{proof}

The next section will extend the type system of $\lambda^{\textbf{TG}}$ in order to:
\begin{enumerate}
    \item Add the \emph{History Expressions} to the type of each expression.
    \item Replace the Coverage Types with the \emph{Extended Coverage Types} so that the \textsc{TSub} rule can be used, but at the same time knowing the full set of values returned by an expression. This will allow for correct over-approximation when binding the type of a variable into an effect type (\emph{History Expressions}).
\end{enumerate}

\section{A Static Type System for Resource Policies}\label{sec:type_system}
Some type system rules will be based on the relationships between type denotations. As we shall see later, although new types are used, these are simply extensions or compositions of others that already exist. That is, the type pairs is composed of an Extended Coverage Type, which in turn is based on Coverage Types. Our choice has been not to create new denotations for these types - this has only been done with History Expressions as already seen in Section \ref{sec:history} - but to use existing ones. The motivation for this choice is to be able to prove the correctness of the type system in a faster and simpler way, guaranteeing the preservation of the properties on the Coverage Types. Therefore, since denotations on Coverage Types are based on a context of only function types with no history and traditional Coverage Types \cite{coverage}, we will make use of two definitions that allow us to transform the context:

\begin{definition}[Context Changes $\phi$]
    \begin{equation}
        \begin{gathered}
            (x{:}\tau_x, \;\Gamma)_{\phi} \equiv x{:}(\tau_x)_{\phi}, \;\Gamma_{\phi} \\
            (\tau, \;H)_{\phi} \equiv \tau_{\phi} \\
            \under{v: b \;|\; \phi \;|\; \psi}_{\phi} \equiv \under{v: b \;|\; \phi} \\
            \overa{v: b \;|\; \phi}_{\phi} \equiv \overa{v: b \;|\; \phi} \\
            (\overa{v: b \;|\; \phi} \rightarrow \kappa)_{\phi} \equiv \overa{v: b \;|\; \phi} \rightarrow (\kappa_{\phi}) \\
            ((a{:}\tau_a \rightarrow \kappa_b) \rightarrow \kappa)_{\phi} \equiv (a{:}(\tau_a)_{\phi} \rightarrow (\kappa_b)_{\phi}) \rightarrow \kappa_{\phi} \\
        \end{gathered}
    \end{equation}
\end{definition}

With these two definitions, we can modify the reference context by adapting it to the definitions of the denotations on the original Coverage Types, working either on the first qualifiers of the Extended Coverage Types ($\phi$), or on the second ones ($\psi$), as required.

\begin{definition}[Context Changes $\psi$]
    \begin{equation}
        \begin{gathered}
            (x{:}\tau_x, \;\Gamma)_{\psi} \equiv x{:}(\tau_x)_{\psi}, \;\Gamma_{\psi} \\
            (\tau, \;H)_{\psi} \equiv \tau_{\psi} \\
            \under{v: b \;|\; \phi \;|\; \psi}_{\psi} \equiv \under{v: b \;|\; \psi} \\
            \overa{v: b \;|\; \phi}_{\psi} \equiv \overa{v: b \;|\; \phi} \\
            (\overa{v: b \;|\; \phi} \rightarrow \kappa)_{\psi} \equiv \overa{v: b \;|\; \phi} \rightarrow (\kappa_{\psi}) \\
            ((a{:}\tau_a \rightarrow \kappa_b) \rightarrow \kappa)_{\psi} \equiv (a{:}(\tau_a)_{\psi} \rightarrow (\kappa_b)_{\psi}) \rightarrow \kappa_{\psi} \\
        \end{gathered}
    \end{equation}
\end{definition}

\subsection{Auxiliary Rules}

In Figure \ref{fig:well-form} we find the extended rules of well-formedness. Two rules have been added: the \textsc{WfHistory} that separately checks the property for the return type and the behaviour type; and the \textsc{WfHistExp} that checks for a History Expression that the qualifiers in each of its component expressions are bound formulae in the context $\Gamma$. For recursive function declarations, the property is checked recursively on the extended context with the parameters defined in the $\mu$ construct.

Furthermore, for the well-formedness of $H$, we require that each identifier of a recursive function be used only in the History Expression representing its body and here only for recursive calls (via $F$ or $call$). In fact $F$ is not really an API! Additionally, it is required as already explained that no new resources are created in the recursive functions.

A further check is made on the APIs defined in $\Delta$ which also must not contain the creation of resources within them.

The rule \textsc{WfBase}, for Extended Coverage Types only adds the closed formula constraint for the second predicate $\psi$.

Here we note the reason why the list of arguments with associated full type (base type and qualifier) was also presented in the definition of recursive functions in History Expressions: as we have seen in the calculation of denotation in Figure \ref{fig:denot_history}, this information is useless, but it is necessary, in the rules just discussed, to verify the well-formedness relationship (\textbf{WF}).

Figure \ref{fig:subtyping} shows the extensions for the rules concerning the verification of subtyping relations. The rule \textsc{SubUBase} is the only one worthy of comment: we keep the denotations for the traditional Coverage Types, and thus verify as in the original rule that the denotation of the type with qualifier $\phi_1$ is subtyped by that with qualifier $\phi_2$. We verify the same relation on the histories $H_1$ and $H_2$. Instead, the denotations on the Coverage Types constructed on $\psi_1$ and $\psi_2$ we want to be equal, i.e. semantically $\psi_1$ is equal to $\psi_2$. This intuitively reflects the role of this added qualifier, which represents all the values actually returned by an expression: we can neither add new ones as we would be inserting guarantees that are not respected, nor can we remove them as we would lose useful values for the correct over-approximation of the histories.

Finally, in Figure \ref{fig:disjunction}, we find the rules useful for joining types belonging to several branches of the same \verb|pattern-matching|. We note that compared to the original disjunction rule \cite{coverage}, here we did not have to use the denotations of the Coverage Types for the function types and preferred, as with the subtyping rules, to proceed structurally on the individual types that make up the function. This choice is constrained by the fact that a function may have as an argument another function that therefore has a latent effect. Separating in this case the History Expressions from the various return types becomes more difficult. Here, however, we have also changed the rules for Coverage Types, as we have to prove that the rules in Figure \ref{fig:disjunction} are correct.

\begin{figure}[H]
    \textbf{Well-Formedness}
    
    \begin{equation}
        \frac{
            \Gamma \vdash^{\textbf{WF}} \tau \quad \Gamma \vdash^{\textbf{WF}} H
        }{
            \Gamma \vdash^{\textbf{WF}} (\tau, \; H)
        }
        \tag{\textsc{WfHistory}}
        \label{eq:wfhistory}
    \end{equation}
    
    \begin{equation}
        \frac{
            \begin{gathered}
                \Gamma \equiv \overline{x_i:\overa{v: b_{x_i} \;|\; \phi_{x_i}}, y_j:\under{v: b_{y_j} \;|\; \phi_{y_j} \;|\; \psi_{y_j}}, z:(a:\tau_a \rightarrow \tau_b)} \\
                \forall \; \alpha(\overline{r_k{:}\phi_k}) \in H, \; \forall k, \;(\overline{\;\exists x_i:b_{x_i}, \exists y_j : b_{y_j}}, \exists v : r_k, \phi_k) \; \text{is a Boolean predicate} \\
                \forall \; call(\phi; \;\overline{a_h{:}(b_h{:}\psi_h)}) \in H, \;(\overline{\;\exists x_i:b_{x_i}, \exists y_j : b_{y_j}}, \exists v : api, \phi) \;\text{is a Boolean predicate} \;\land \\
                \land\; \forall h, (\overline{\;\exists x_i:b_{x_i}, \exists y_j : b_{y_j}}, \exists v : b_h, \psi_h) \;\text{is a Boolean predicate} \\
                \forall \; F(\overline{a_m{:}(b_m{:}\sigma_m)}) \in H, \; \forall m, (\overline{\;\exists x_i:b_{x_i}, \exists y_j : b_{y_j}}, \exists v : b_m, \sigma_m) \;\text{is a Boolean predicate} \\
                \forall \;\mu F(\overline{a_l{:}({b_l{:}\theta_l})})(H_F) \in H, \; \Gamma, \overline{a_l: \overa{v: b_l \;|\; \theta_l}} \vdash^{\textbf{WF}} H_F \;\land \\
                \land\; \forall \;G(\overline{a{:}(b{:}\phi)}) \in H, \;G = F \;\Longrightarrow\; G(\overline{a{:}(b{:}\phi)}) \in H_F \;\land \\ \land\; \forall \;call(\phi; \;\overline{a{:}(b{:}\psi)}) \in H, \;\phi[v \mapsto F] \;\Longrightarrow\; call(\phi; \;\overline{a{:}(b{:}\psi)}) \in H_F \;\land \\
                \land\; \forall \;Val_H \in H_F. \;\forall r. \;\forall \;X \in Rid. \;\neg(Val_H = \texttt{\small{new}}_r(X)) \;\land\; \forall \;\alpha(\overline{r{:}\theta}) \in H, \;\forall \;\theta, \;\neg(\theta[v \mapsto F]) \\
                \forall \;(F, \;\overline{\tau} \rightarrow (\tau_F, \;H_F)) \in \Delta, \;\forall \;Val_H \in H_F, \;\forall r. \;\forall \;X \in Rid, \;\neg(Val_H = \texttt{\small{new}}_r(X)) \\
            \end{gathered} 
        }{
            \Gamma \vdash^{\textbf{WF}} H
        }
        \tag{\textsc{WfHistExp}}
        \label{eq:wfhexp}
    \end{equation}
    
    \begin{equation}
        \frac{
            \begin{gathered}
                \Gamma \equiv \overline{x_i:\overa{v: b_{x_i} \;|\; \phi_{x_i}}, y_j:\under{v: b_{y_j} \;|\; \phi_{y_j} \;|\; \psi_{y_j}}, z:(a:\tau_a \rightarrow \tau_b)} \\
                (\overline{\forall x_i: b_{x_i}, \exists y_j: b_{y_j}}, \forall v: b, \phi) \; \text{is a Boolean predicate} \quad \forall j, \text{err} \notin \llbracket \under{v: b_{y_j} \;|\; \phi_{y_j}} \rrbracket_{\Gamma_{\phi}} \\
                (\overline{\forall x_i: b_{x_i}, \exists y_j: b_{y_j}}, \forall v: b, \psi) \; \text{is a Boolean predicate}
            \end{gathered}
        }{
            \Gamma \vdash^{\textbf{WF}} \under{v: b \;|\; \phi \;|\; \psi}
        }
        \tag{\textsc{WfBase}}
        \label{eq:wfbase}
    \end{equation}
    \begin{multicols}{2}
        \raggedcolumns
        \begin{equation}
            \frac{
                \Gamma, x:\overa{v: b \;|\; \phi} \vdash^{\textbf{WF}} \kappa
            }{
                \Gamma \vdash^{\textbf{WF}} x:\overa{v: b \;|\; \phi} \rightarrow \kappa
            }
            \tag{\textsc{WfArg}}
            \label{eq:wfarg}
        \end{equation}
    
        \columnbreak
    
        \begin{equation}
            \frac{
                \Gamma \vdash^{\textbf{WF}} (a: \tau_a \rightarrow \kappa_b) \quad \Gamma \vdash^{\textbf{WF}} \kappa
            }{
                \Gamma \vdash^{\textbf{WF}} (a: \tau_a \rightarrow \kappa_b) \rightarrow \kappa
            }
            \tag{\textsc{WfRes}}
            \label{eq:wfres}
        \end{equation}
    \end{multicols}
    \caption{Extension of the auxiliary rules of Well-Formedness.}
    \label{fig:well-form}
\end{figure}

\begin{figure}[H]
    \textbf{Subtyping}
    
    \begin{equation}
        \frac{
            \llbracket \under{v: b \;|\; \phi_1} \rrbracket_{\Gamma_{\phi}} \subseteq \llbracket \under{v: b \;|\; \phi_2} \rrbracket_{\Gamma_{\phi}} \quad \llbracket H_1 \rrbracket_{\Gamma_{\psi}} \subseteq \llbracket H_2 \rrbracket_{\Gamma_{\psi}} \quad \llbracket \under{v: b \;|\; \psi_1}\rrbracket_{\Gamma_{\psi}} = \llbracket \under{v: b \;|\; \psi_2} \rrbracket_{\Gamma_{\psi}}
        }{
            \Gamma \vdash (\under{v: b \;|\; \phi_1 \;|\; \psi_1}, \; H_1) <: (\under{v: b \;|\; \phi_2 \;|\; \psi_2}, \; H_2)
        }
        \tag{\textsc{SubUBase}}
        \label{eq:subu}
    \end{equation}

    \begin{equation}
        \frac{
            \llbracket \overa{v: b \;|\; \phi_1} \rrbracket_{\Gamma_{\phi}} \subseteq \llbracket \overa{v: b \;|\; \phi_2} \rrbracket_{\Gamma_{\phi}}
        }{
            \Gamma \vdash \overa{v: b \;|\; \phi_1} <: \overa{v: b \;|\; \phi_2}
        }
        \tag{\textsc{SubOBase}}
        \label{eq:subo}
    \end{equation}

    \begin{equation}
        \frac{
            \Gamma \vdash \tau_{21} <: \tau_{11} \quad \Gamma, x: \tau_{21} \vdash \kappa_{12} <: \kappa_{22}
        }{
            \Gamma \vdash x: \tau_{11} \rightarrow \kappa_{12} <: x:\tau_{21} \rightarrow \kappa_{22}
        }
        \tag{\textsc{SubArr}}
        \label{eq:subarr}
    \end{equation}
    \caption{Extension of the auxiliary rules of subtyping.}
    \label{fig:subtyping}
\end{figure}

\begin{figure}[h]
    \textbf{Disjunction}
    
    \begin{equation}
        \frac{
            \begin{gathered}
                \llbracket \under{v: b \;|\; \phi_1} \rrbracket_{\Gamma_{\phi}} \cap \llbracket \under{v: b \;|\; \phi_2} \rrbracket_{\Gamma_{\phi}} = \llbracket \under{v: b \;|\; \phi_3} \rrbracket_{\Gamma_{\phi}} \\
                \llbracket \under{v: b \;|\; \psi_1} \rrbracket_{\Gamma_{\psi}} \cap \llbracket \under{v: b \;|\; \psi_2} \rrbracket_{\Gamma_{\psi}} = \llbracket \under{v: b \;|\; \psi_3} \rrbracket_{\Gamma_{\psi}}
            \end{gathered}
        }{
            \Gamma \vdash \under{v: b \;|\; \phi_1 \;|\; \psi_1} \;\lor\; \under{v: b \;|\; \phi_2 \;|\; \psi_2} = \under{v: b \;|\; \phi_3 \;|\; \psi_3}
        }
        \tag{\textsc{Disjunction Under}}
    \end{equation}

    \begin{equation}
        \frac{
            \llbracket \overa{v: b \;|\; \phi_1} \rrbracket_{\Gamma_{\phi}} \cup \llbracket \overa{v: b \;|\; \phi_2} \rrbracket_{\Gamma_{\phi}} = \llbracket \overa{v: b \;|\; \phi_3} \rrbracket_{\Gamma_{\phi}}
        }{
            \Gamma \vdash \overa{v: b \;|\; \phi_1} \;\lor\; \overa{v: b \;|\; \phi_2} = \overa{v: b \;|\; \phi_3}
        }
        \tag{\textsc{Disjunction Over}}
    \end{equation}

    \begin{equation}
        \frac{
            \llbracket H_1 \rrbracket_{\Gamma_{\psi}} \cup \llbracket H_2 \rrbracket_{\Gamma_{\psi}} = \llbracket H_3 \rrbracket_{\Gamma_{\psi}}
        }{
            \Gamma \vdash H_1 \;\lor\; H_2 = H_3
        }
        \tag{\textsc{Disjunction History}}
    \end{equation}

    \begin{equation}
        \frac{
            \begin{gathered}
                \Gamma \vdash \tau_1 \;\lor\; \tau_2 = \tau_3 \quad \Gamma \vdash H_1 \;\lor\; H_2 = H_3
            \end{gathered}
        }{
            \Gamma \vdash (\tau_1, \;H_1) \;\lor\; (\tau_2, \;H_2) = (\tau_3, \;H_3)
        }
        \tag{\textsc{Disjunction Pair}}
    \end{equation}

    \begin{equation}
        \frac{
            \Gamma \vdash \tau_{a_1} \;\land\; \tau_{a_2} = \tau_{a_3} \quad \Gamma, a{:}\tau_{a_3} \vdash \kappa_1 \;\lor\; \kappa_2 = \kappa_3
        }{
            \Gamma \vdash a{:}\tau_{a_1} \rightarrow \kappa_1 \;\lor\; a{:}\tau_{a_2} \rightarrow \kappa_2 = a{:}\tau_{a_3} \rightarrow \kappa_3
        }
        \tag{\textsc{Disjunction Arrow}}
    \end{equation}

    \begin{equation}
        \frac{
            \begin{gathered}
                \llbracket \under{v: b \;|\; \phi_1} \rrbracket_{\Gamma_{\phi}} \cup \llbracket \under{v: b \;|\; \phi_2} \rrbracket_{\Gamma_{\phi}} = \llbracket \under{v: b \;|\; \phi_3} \rrbracket_{\Gamma_{\phi}} \\
                \llbracket \under{v: b \;|\; \psi_1} \rrbracket_{\Gamma_{\psi}} \cup \llbracket \under{v: b \;|\; \psi_2} \rrbracket_{\Gamma_{\psi}} = \llbracket \under{v: b \;|\; \psi_3} \rrbracket_{\Gamma_{\psi}}
            \end{gathered}
        }{
            \Gamma \vdash \under{v: b \;|\; \phi_1 \;|\; \psi_1} \;\land\; \under{v: b \;|\; \phi_2 \;|\; \psi_2} = \under{v: b \;|\; \phi_3 \;|\; \psi_3}
        }
        \tag{\textsc{Conjunction Under}}
    \end{equation}

    \begin{equation}
        \frac{
            \llbracket \overa{v: b \;|\; \phi_1} \rrbracket_{\Gamma_{\phi}} \cap \llbracket \overa{v: b \;|\; \phi_2} \rrbracket_{\Gamma_{\phi}} = \llbracket \overa{v: b \;|\; \phi_3} \rrbracket_{\Gamma_{\phi}}
        }{
            \Gamma \vdash \overa{v: b \;|\; \phi_1} \;\land\; \overa{v: b \;|\; \phi_2} = \overa{v: b \;|\; \phi_3}
        }
        \tag{\textsc{Conjunction Over}}
    \end{equation}

    \begin{equation}
        \frac{
            \llbracket H_1 \rrbracket_{\Gamma_{\psi}} \cap \llbracket H_2 \rrbracket_{\Gamma_{\psi}} = \llbracket H_3 \rrbracket_{\Gamma_{\psi}}
        }{
            \Gamma \vdash H_1 \;\land\; H_2 = H_3
        }
        \tag{\textsc{Conjunction History}}
    \end{equation}

    \begin{equation}
        \frac{
            \begin{gathered}
                \Gamma \vdash \tau_1 \;\land\; \tau_2 = \tau_3 \quad \Gamma \vdash H_1 \;\land\; H_2 = H_3
            \end{gathered}
        }{
            \Gamma \vdash (\tau_1, \;H_1) \;\land\; (\tau_2, \;H_2) = (\tau_3, \;H_3)
        }
        \tag{\textsc{Conjunction Pair}}
    \end{equation}

    \begin{equation}
        \frac{
            \Gamma \vdash \tau_{a_1} \;\lor\; \tau_{a_2} = \tau_{a_3} \quad \Gamma, a{:}\tau_{a_3} \vdash \kappa_1 \;\land\; \kappa_2 = \kappa_3
        }{
            \Gamma \vdash a{:}\tau_{a_1} \rightarrow \kappa_1 \;\land\; a{:}\tau_{a_2} \rightarrow \kappa_2 = a{:}\tau_{a_3} \rightarrow \kappa_3
        }
        \tag{\textsc{Conjunction Arrow}}
    \end{equation}
    \caption{Extension of auxiliary disjunction rules.}
    \label{fig:disjunction}
\end{figure}

\begin{proof}[Correctness of Disjunction Rules for Coverage Types]
    The correctness of the procedure of unpacking the function and finding a disjunct (or conjunct) type for each individual type refinement and then recomposing the function structurally is actually already demonstrated in the original work when discussing the algorithmic procedures \textsc{\textcolor{blue}{Disj}} and \textsc{\textcolor{blue}{Conj}} \cite{coverage}. Here, however, we do not want to proceed in an algorithmic way (exploiting properties on qualifiers) and so we still rely on denotations showing that the procedure is the same. What we do know is that, for example given two Coverage Types $\under{v: b \;|\; \phi_1}$ and $\under{v: b \;|\; \phi_2}$, and a type context $\Gamma$, the following relationships apply:
    
    \begin{equation}
        \begin{gathered}
            \Gamma \vdash \under{v: b \;|\; \phi_1 \;\lor\; \phi_2} <: \under{v: b \;|\; \phi_1} \\
            \Gamma \vdash \under{v: b \;|\; \phi_1 \;\lor\; \phi_2} <: \under{v: b \;|\; \phi_2}
        \end{gathered}
    \end{equation}

    Thus, by definition of the subtyping relationship, it holds that:
    \begin{equation}
        \begin{gathered}
            \llbracket \under{v: b \;|\; \phi_1 \;\lor\; \phi_2} \rrbracket_{\Gamma_{\phi}} \subseteq \llbracket \under{v: b \;|\; \phi_1} \rrbracket_{\Gamma_{\phi}} \\
            \llbracket \under{v: b \;|\; \phi_1 \;\lor\; \phi_2} \rrbracket_{\Gamma_{\phi}} \subseteq \llbracket \under{v: b \;|\; \phi_2} \rrbracket_{\Gamma_{\phi}}
        \end{gathered}
    \end{equation}

    These two relationships lead to the statement that:
    
    \begin{equation}
         \llbracket \under{v: b \;|\; \phi_1} \rrbracket_{\Gamma_{\phi}} \cap \llbracket \under{v: b \;|\; \phi_2} \rrbracket_{\Gamma_{\phi}} \subseteq \llbracket \under{v: b \;|\; \phi_1 \;\lor\; \phi_2} \rrbracket_{\Gamma_{\phi}}
    \end{equation}

    But we know that $\under{v: b \;|\; \phi_1 \;\lor\; \phi_2}$ is the disjoint type of the two sets, which leads us to state that its denotation corresponds exactly to the intersection of the two sets:

    \begin{equation}
        \llbracket \under{v: b \;|\; \phi_1} \rrbracket_{\Gamma_{\phi}} \cap \llbracket \under{v: b \;|\; \phi_2} \rrbracket_{\Gamma_{\phi}} = \llbracket \under{v: b \;|\; \phi_1 \;\lor\; \phi_2} \rrbracket_{\Gamma_{\phi}}
    \end{equation}

    A dual demonstration can be performed on the conjunction of two Coverage Types leading to the following relations:

    \begin{equation}
        \begin{gathered}
            \Gamma \vdash \under{v: b \;|\; \phi_2} <: \under{v: b \;|\; \phi_1 \;\land\; \phi_2} \quad \quad \Gamma \vdash \under{v: b \;|\; \phi_1} <: \under{v: b \;|\; \phi_1 \;\land\; \phi_2} \\
            \llbracket \under{v: b \;|\; \phi_2} \rrbracket_{\Gamma_{\phi}} \subseteq \llbracket \under{v: b \;|\; \phi_1 \;\land\; \phi_2} \rrbracket_{\Gamma_{\phi}} \quad \quad \llbracket \under{v: b \;|\; \phi_1} \rrbracket_{\Gamma_{\phi}} \subseteq \llbracket \under{v: b \;|\; \phi_1 \;\land\; \phi_2} \rrbracket_{\Gamma_{\phi}} \\
            \llbracket \under{v: b \;|\; \phi_1} \rrbracket_{\Gamma_{\phi}} \cup \llbracket \under{v: b \;|\; \phi_2} \rrbracket_{\Gamma_{\phi}} = \llbracket \under{v: b \;|\; \phi_1 \;\land\; \phi_2} \rrbracket_{\Gamma_{\phi}}
        \end{gathered}
    \end{equation}

    We can also verify that these relations also apply to over-approximation types.

    During this demonstration, we have also shown implicitly that when checking the disjunction of a function the addition of the conjoined type of the preceding parameters to the context is correct, since in the subtyping rules this is done, but it is these that lead us to the intersection (or union) of the denotations of the disjoined (or conjoined) types! And these denotations are precisely computed under the same $\Gamma$ context where the conjoined (or disjoint) type will be found.

    In conclusion, we can say that the entire new set of disjunction rules is correct since we have shown that:
    \begin{enumerate*}[label=(\roman*)]
        \item structural separation of types is correct; 
        \item the assumptions concerning the intersection and union of the denotations for disjunction and conjunction (used for function parameters that are in contravariance to the main operator) respectively are correct; 
        \item finally, the evaluation of the denotations of the individual function types in the extended contexts with the previously disjointed (or conjoined) type is also correct.
    \end{enumerate*}
\end{proof}

\subsection{Type System}

Now, in Figures \ref{fig:type-system-1} and \ref{fig:type-system-2} we present the type system rules of the extended $\lambda^{\textbf{TG}}$ language for resource usage. We will show only a few selected ones, however most of them, the most noteworthy. However, it is possible to find the remaining set of rules in Appendix \ref{app:type-system}.

The special expression \verb|err| denoting divergent computations can be typed only with the contradictory qualifier $\bot$ (which is also the qualifier representing the exact values of the expression), while the effect produced is the special action $err()$.

As demonstrated in the Lemma \ref{lemma:ty} the $Ty$ function indicates not only the values it will surely return, but also specifies all of them, so to align with the syntax of Extended Coverage Types we need only replicate the qualifier a second time. The effect produced by just invoking a constant (\textsc{TConst}) or an action (\textsc{TAction}), without applying it, is null. The same is also true for variables. In the rule \textsc{TVarBase} the tautological qualifier $v = x$ is also used for the second qualifier: $x$ will be present in $\Gamma$ with an Extended Coverage Types associated with it, so the first and second $v = x$ will refer to the first and second qualifiers present in the type of $x$ in $\Gamma$, respectively.

For the rules \textsc{TSub} and \textsc{TEq} the only change is the use of the pair type $\pi$ instead of a traditional Coverage Type.

Even the mere declaration of a function is typed with the empty history (\textsc{TFun}). However, we introduce a new rule called \textsc{TFunFlat} that allows us to eliminate empty effects associated with parameters when a function has at least two. For example, it will be possible to change the type of a function declaration from (using only the rule \textsc{TFun}):
\begin{equation}
    (a{:}\tau_a \rightarrow (b{:}\tau_b \rightarrow (c{:}\tau_c \rightarrow (\tau, \;H), \;\epsilon), \;\epsilon), \;\epsilon)
\end{equation}
To one that eliminates all associated empty effects from the second to the last parameter:
\begin{equation}
    (a{:}\tau_a \rightarrow b{:}\tau_b \rightarrow c{:}\tau_c \rightarrow (\tau, \;H), \;\epsilon)
\end{equation}

The \textsc{TMatch} rule now prohibits being able to tip only one branch (or at any rate a smaller number than the total number of branches) of a \verb|pattern-matching|. We are obliged to make this change otherwise the property of representing \emph{all and only} the values returned by the expression in the second qualifier of an Extended Coverage Type would not be respected.
We note how the return type of each $v_i$ and $d_i(\overline{y})$ can be different (because of the $\psi$ qualifier of Extended Coverage Types), but we introduce the notation $\langle \cdot \rangle$ to transform Extended Coverage Types into traditional Coverage Types.

Despite this change, it is still possible to drop branches of a \verb|pattern-matching| by applying the \textsc{TSub} rule after the \textsc{TMatch} rule. In this way we retain the properties of the old type system but we are also able to identify in the second type qualifier \emph{all and only} the values returned to construct a proper over-approximation of the effects.

The rule for defining recursive functions is one that has been changed the most. In this case, in order to implement recursion also at the level of the History Expressions, we treat the function as an API: in this way the API calls (which will represent the recursive calls) will be present in the History associated with the recursive construct $\mu$, and as seen in the definition of denotations (Figure \ref{fig:denot_history}) this strategy allows us to implement a form of recursion. The History of the function associated with the identifier $F$ in $\Delta$ will be $\epsilon$ as it is not to be used in the calculation of the denotation of History Expressions. This is the only case in which the definition of a function will generate an active History, which is precisely the recursive construct $\mu$!

\begin{figure}[H]
    \begin{multicols}{2}
        \begin{equation}
            \frac{
                \Gamma \vdash^{\textbf{WF}} (\under{v: b \;|\; \bot \;|\; \bot}, \; err())
            }{
                \Gamma \vdash err : (\under{v: b \;|\; \bot \;|\; \bot}, \; err())
            }
            \tag{\textsc{TErr}}
            \label{eq:terr}
        \end{equation}
        
        \columnbreak
    
        \begin{equation}
            \frac{
                \Gamma \vdash^{\textbf{WF}} (\text{Ty}(c), \; \epsilon)
            }{
                \Gamma \vdash c : (\text{Ty}(c), \; \epsilon)
            }
            \tag{\textsc{TConst}}
            \label{eq:tconst}
        \end{equation}
    \end{multicols}

    \begin{multicols}{2}
        \begin{equation}
            \frac{
                \Gamma \vdash^{\textbf{WF}} (\text{Ty}(\alpha), \; \epsilon)
            }{
                \Gamma \vdash \alpha : (\text{Ty}(\alpha), \; \epsilon)
            }
            \tag{\textsc{TAction}}
            \label{eq:taction}
        \end{equation}

        \columnbreak

        \begin{equation}
            \frac{
                \Gamma \vdash^{\textbf{WF}} (\under{v: b \;|\; v = x \;|\; v = x}, \; \epsilon)
            }{
                \Gamma \vdash x : (\under{v: b \;|\; v = x \;|\; v = x}, \; \epsilon)
            }
            \tag{\textsc{TVarBase}}
            \label{eq:tvarb}
        \end{equation}
    \end{multicols}

    \begin{multicols}{2}
        \noindent
        \begin{equation}
            \frac{
                \Gamma(x) = (a: \tau_a \rightarrow \kappa) \quad \Gamma \vdash^{\textbf{WF}} (a:\tau_a \rightarrow \kappa, \; \epsilon)
            }{
                \Gamma \vdash x : (a:\tau_a \rightarrow \kappa, \; \epsilon)
            }
            \tag{\textsc{TVarFun}}
            \label{eq:tvarf}
        \end{equation}
        \columnbreak
        \begin{equation}
            \frac{
                \varnothing \vdash \pi <: \pi' \quad \varnothing \vdash e : \pi \quad \Gamma \vdash^{\textbf{WF}} \pi'
            }{
                \Gamma \vdash e : \pi'
            }
            \tag{\textsc{TSub}}
            \label{eq:tsub}
        \end{equation}
    \end{multicols}

    \begin{multicols}{2}
        \begin{equation}
            \frac{
                \begin{gathered}
                    \Gamma \vdash \pi' <: \pi \quad \Gamma \vdash \pi <: \pi' \\
                    \Gamma \vdash e: \pi \quad \Gamma \vdash^{\textbf{WF}} \pi'
                \end{gathered}
            }{
                \Gamma \vdash e: \pi'
            }
            \tag{\textsc{TEq}}
            \label{eq:teq}
        \end{equation}

        \columnbreak

        \begin{equation}
            \frac{
                \Gamma, x: \tau_x \vdash e: \pi \quad \Gamma \vdash^{\textbf{WF}} (x: \tau_x \rightarrow \pi, \; \epsilon)
            }{
                \Gamma \vdash \lambda x : \lfloor\tau_x\rfloor. e : (x: \tau_x \rightarrow \pi, \; \epsilon)
            }
            \tag{\textsc{TFun}}
            \label{eq:tfun}
        \end{equation}
    \end{multicols}

    \begin{equation}
        \frac{
            \Gamma, x: \tau_x \vdash e: (y: \tau_y \rightarrow \kappa, \; \epsilon) \quad \Gamma \vdash^{\textbf{WF}} (x: \tau_x \rightarrow y: \tau_y \rightarrow \kappa, \; \epsilon)
        }{
            \Gamma \vdash \lambda x : \lfloor\tau_x\rfloor. e : (x: \tau_x \rightarrow y: \tau_y \rightarrow \kappa, \; \epsilon)
        }
        \tag{\textsc{TFunFlat}}
        \label{eq:tflat}
    \end{equation}    

    \begin{equation}
        \frac{
            \begin{gathered}
                \forall i, \;\Gamma \vdash v: (\tau_{v_i}, \; H_{v_i}) \quad \Gamma, \overline{y: \tau_y} \vdash d_i(\overline{y}) : (\tau_{d_i}, \;H_{y_i}) \quad \langle \tau_{v_i} \rangle = \langle \tau_{d_i} \rangle \quad \Gamma, \overline{y: \tau_y} \vdash e_i: (\tau_i, \;H_{e_i}) \\
                \tau = \underset{i}{\bigvee} \tau_i \quad H = \underset{i}{\bigvee} (H_{v_i} \cdot H_{y_i} \cdot H_{e_i}) \quad \Gamma \vdash^{\textbf{WF}} (\tau, \;H)   
            \end{gathered}
        }{
            \Gamma \vdash (\texttt{match}\;v\;\texttt{with}\;\overline{d_i\overline{y_j} \rightarrow e_i}): (\tau, \;H)
        }
        \tag{\textsc{TMatch}}
        \label{eq:tmatch}
    \end{equation}

    \begin{equation}
        \frac{
            \Gamma \vdash e : (\tau_1, \; H_1) \quad \Gamma \vdash e : (\tau_2, \; H_2) \quad \Gamma \vdash \tau_1 \;\lor\; \tau_2 = \tau \quad \Gamma \vdash H_1 \;\lor\; H_2 = H \quad \Gamma \vdash^{\textbf{WF}} (\tau, \; H)
        }{
            \Gamma \vdash e : (\tau, \; H)
        }
        \tag{\textsc{TMerge}}
        \label{eq:tmerge}
    \end{equation}

    \begin{equation}
        \frac{
            \begin{gathered}
                \nextd(\Delta, \;x{:}\overa{v: b \;|\; v \prec x \;\land\; \phi} \rightarrow \overline{\tau_i} \rightarrow (\tau, \;\epsilon)) = F \\
                \Gamma, x{:}\overa{v: b \;|\; \phi}, f{:}\under{v: api \;|\; v = F \;|\; v = F} \vdash e : \overline{\tau_i} \rightarrow (\tau, \;H) \\
                A = \{x{:}(b{:}\phi)\} \cup \{a_i{:}(b_i{:}\psi_i) \;|\; \tau_i = a_i{:}\overa{v: b_i \;|\; \psi_i}\} \quad \Gamma \vdash^{\textbf{WF}} (x{:}\overa{v: b \;|\; \phi} \rightarrow \overline{\tau_i} \rightarrow (\tau, \;H), \;\mu F(A)(H))
            \end{gathered}
        }{
            \Gamma \vdash \text{fix}\;f{:}(b \rightarrow \lfloor\kappa\rfloor).\lambda x{:}b.\;e : (x{:}\overa{v: b \;|\; \phi} \rightarrow \overline{\tau_i} \rightarrow (\tau, \;H), \;\mu F(A)(H))
        }
        \tag{\textsc{TFix}}
        \label{eq:tfix}
    \end{equation}

    \begin{equation}
        \frac{
            \begin{gathered}
                \neg(\kappa_x = \pi)
                \quad \Gamma \vdash v_1: (a{:}\overa{v: b \;|\; \phi} \rightarrow \kappa_x, \; H_{v_1}) \quad \Gamma \vdash v_2: (\under{v: b \;|\; \phi \;|\; \psi}, \; H_{v_2}) \\
                \Gamma, x : \kappa_x[a \mapsto v_2] \vdash e: (\tau, \; H_e) \quad \Gamma \vdash^{\textbf{WF}} (\tau, \; H_{v_1} \cdot H_{v_2} \cdot H_e) \\
            \end{gathered}
        }{
            \Gamma \vdash \;\texttt{let}\; x = v_1 \; v_2 \;\texttt{in}\; e: (\tau, \; H_{v_1} \cdot H_{v_2} \cdot H_e)
        }
        \tag{\textsc{TAppMulti}}
        \label{eq:tappmulti}
    \end{equation}
    \caption{Selected type system rules - Part I.}
    \label{fig:type-system-1}
\end{figure}

\begin{figure}[h]
    \begin{equation}
        \frac{
            \begin{gathered}
                \Gamma \vdash v_1: (a{:}\overa{v: b \;|\; \phi} \rightarrow (\tau_{v_1}, \;H_{\tau_{v_1}}), \; H_{v_1}) \quad \Gamma \vdash v_2: (\under{v: b \;|\; \phi \;|\; \psi}, \; H_{v_2}) \\
                \Theta = \{[Y/X] \;|\; new(X) \in H_{\tau_{v_1}} \;\land\; \nextd(\Delta) = Y\} \quad H_{\tau_{v_1}}^{\star} = H_{\tau_{v_1}}(\Theta)[a \mapsto v_2] \\
                \Gamma, x : \tau_{v_1}(\Theta)[a \mapsto v_2] \vdash e: (\tau, \; H_e) \quad \Gamma \vdash^{\textbf{WF}} (\tau, \; H_{v_1} \cdot H_{v_2} \cdot H_{\tau_{v_1}}^{\star} \cdot H_e)
            \end{gathered}
        }{
            \Gamma \vdash \;\texttt{let}\; x = v_1 \; v_2 \;\texttt{in}\; e: (\tau, \; H_{v_1} \cdot H_{v_2} \cdot H_{\tau_{v_1}}^{\star} \cdot H_e)
        }
        \tag{\textsc{TAppLast}}
        \label{eq:tapplast}
    \end{equation}

    \begin{equation}
        \frac{
            \Gamma, x: \under{v: r \;|\; v = X \;|\; v = X} \vdash e : (\tau, \; H_e) \quad \nextd(\Delta) = X \quad \Gamma \vdash^{\textbf{WF}} (\tau, \; \texttt{\small{new}}_r(X) \cdot H_e)
        }{
            \Gamma \vdash \texttt{\small{let}}\;x = \texttt{\small{new}}_r\;()\;\texttt{\small{in}}\;e : (\tau, \;\texttt{\small{new}}_r(X) \cdot H_e)
        }
        \tag{\textsc{TNew}}
        \label{eq:tnew}
    \end{equation}

    \begin{equation}
        \frac{
            \Gamma, x : \under{v: api \;|\; v = F \;|\; v = F} \vdash e : (\tau, \; H_e) \quad \Delta(F) = \tau_{F} \quad \Gamma \vdash^{\textbf{WF}} (\tau, \; \texttt{\small{get}}(F) \cdot H_e)
        }{
            \Gamma \vdash \texttt{\small{let}}\;x = \texttt{\small{get}}\;F\;\texttt{\small{in}}\;e : (\tau, \; \texttt{\small{get}}(F) \cdot H_e)
        }
        \tag{\textsc{TGet}}
        \label{eq:tgetapi}
    \end{equation}

    \begin{equation}
        \frac{
            \begin{gathered}
                \Gamma \vdash \alpha: (\overline{a_i{:}\overa{v: b_i \;|\; \phi_i}} \rightarrow \tau_x, \;H_{\alpha}) \quad \forall i, \; \Gamma \vdash u_i : (\under{v: b_i \;|\; \phi_i \;|\; \psi_i}, \; H_{u_i}) \\
                \Gamma, x{:}\tau_x[\overline{a_i \mapsto u_i}] \vdash e : (\tau, \; H_e) \quad \Gamma \vdash^{\textbf{WF}} (\tau, \;H_\alpha \cdot ( \underset{i}{\bullet} \; H_{u_i} ) \cdot \alpha(\overline{b_i{:}\psi_i}) \cdot H_e)
            \end{gathered}
        }{
            \Gamma \vdash \texttt{\small{let}}\;x = \alpha \; \overline{u_i}\;\texttt{\small{in}}\;e : (\tau, \; H_\alpha \cdot ( \underset{i}{\bullet} \; H_{u_i} ) \cdot \alpha(\overline{b_i{:}\psi_i}) \cdot H_e)
        }
        \tag{\textsc{TLetAction}}
        \label{eq:tletaction}
    \end{equation}
    
    \begin{equation}
        \frac{
            \begin{gathered}
                \Gamma \vdash f: \under{v: api \;|\; \phi \;|\; \psi} \quad \forall i, \; \Gamma \vdash y_i: (\under{v: b_i \;|\; \phi_i \;|\; \psi_i}, \; H_{y_i}) \\ \varnothing \vdash^{\textbf{WF}} \under{v: api \;|\; \phi \;|\; \psi} \quad Api = \{F \;|\; \phi[v \mapsto F] \;\land\; \Delta(F)\downarrow\} \\
                \forall F \in Api, \; \Delta(F) = \tau_F \;\land\; \Gamma, f'{:}\tau_F \vdash f': \overline{a_i{:}\overa{v: b_i \;|\; \phi_i}} \rightarrow (\tau_{F_x}, \; H_{F_x}) \;\land\; \\
                \Gamma \vdash \underset{F \in Api}{\bigvee} \tau_{F_x} = \tau_x \quad \Gamma, x{:}\tau_x\overline{[a_i \mapsto y_i]} \vdash e : (\tau, \; H_e) \quad \Gamma \vdash^{\textbf{WF}} (\tau, \; ( \underset{i}{\bullet} \; H_{y_i} ) \cdot call(\psi; \;\overline{a_i{:}(b_i{:}\psi_i}) \cdot H_e)
            \end{gathered}
        }{
            \Gamma \vdash \texttt{\small{let}}\;x = f\;\overline{y_i}\; \texttt{\small{in}}\;e : (\tau, \; ( \underset{i}{\bullet} \; H_{y_i} ) \cdot call(\psi; \;\overline{a_i{:}(b_i{:}\psi_i})) \cdot H_e)
        }
        \tag{\textsc{TAppAPI}}
        \label{eq:tappapi}
    \end{equation}
    \caption{Selected type system rules - Part II.}
    \label{fig:type-system-2}
\end{figure}

The rule \textsc{TMerge} also performs the disjunction between the two histories assignable to the expression. To emphasize the difference between over-approximated and under-approximated types, it is worth noting how, despite taking the disjunction of both the return type and the behaviour type, the relationship between the new type found via disjunction and the two types that are part of it is reversed: as seen in the auxiliary rules for disjunction, in Coverage Types we take the intersection of the denotations between the two types, because the type found is a subtype of them; while in History Expressions we take the union of the denotations, finding instead a supertype of the two.

The rule previously called \textsc{TApp}, for applying functions, has been split into two separate rules: \textsc{TAppMulti} and \textsc{TAppLast}, the former is used when the function represented by $v_1$ has more than one argument, and the latter when $v_1$ has exactly one argument. This choice is due to having to distinguish the moment when the last argument of a function is applied to a parameter, and thus the latent effect of the function can become active.

In the rule \textsc{TAppMulti} we specify in the premises that the return type of the function, $\kappa_x$ must be different from any pair type $\pi$, and so there will be at least one other parameter of the function. In the rule \textsc{TAppLast}, on the other hand, we require that immediately after the parameter $a$ there must be a pair type; in this case, we will have to make the latent effect $H_{\tau_{v_1}}$ active, i.e., include it in the type of the behaviour of the final expression. To do this, we must first address a problem: this concerns the creation of resources within the latent effects; in fact, if the function were applied a second time, and resources were created within it, the second application would use the same identifiers as the first, but this is wrong! So, we will first need to find all the identifiers that belong to a $new_r$ inside the latent effect, then we go on to create replacements for each of these identifiers, taking new ones.

Here, a new function called \nextd $\space$ comes into play, which takes care of taking the next free identifier according to any ordering of the set of identifiers: on the merits, the set $Rid$ of resources is a finite set (in fact in concrete the resources are always in finite quantity) so it is possible to put it in bi-univocal correspondence with the set of natural numbers $\mathbb{N}$ and therefore it is possible to define on the identifiers an ordering by indexing them. It is good to note, how the ordering that is taken does not matter to us, in fact we can think of the \nextd $\space$ as a function that finds an ordering of the identifiers (a permutation of them) such that if in $\Delta$ there are $n$ identifiers, these are the first $n$ in the ordering and $Y$ is the $n + 1$-th.

Once computed in the set $\Theta$, these substitutions will be applied to both the return type $\tau_{v_1}$ and the latent effect itself, in addition to the traditional substitution of the argument $a$ with the passed parameter $v_2$. 

The \textsc{TNew} rule allows a new resource to be created and associated with a variable. Here too, the procedure \nextd $\space$ is used to ensure that the next free identifier in $\Delta$ is precisely $X$. A Coverage Type with base type $r$ and exactly the same value as $X$ will be associated with $x$.

The \textsc{TGet} rule, on the other hand, allows an external function to be associated with a variable, in which case the identifier associated with the API must have already been defined in $\Delta$ by the programmer himself (intuitively, this is a correct choice since it is he himself who imports the external libraries).

The rule \textsc{TAction} allows the application of an operator involving resources. In fact, the rule is very similar to the \textsc{TOp} found in Appendix \ref{app:type-system}, the only difference being that in this case the event indicating the presence of the action $\alpha$ is also added to the active effect. We note that while the set of qualifiers $\phi_i$ is the one used to match the type of the parameters passed with those of the action arguments; the qualifiers $\psi_i$ are those encapsulated in the event present in the history.

The rule \textsc{TAppApi}, for the application of an external function, requires that the variable $f$ be of base type $api$, and that the Extended Coverage Type associated with it be closed with respect to the empty context; this allows us to be able to compute all the identifiers that satisfy the condition $\phi$, finding the APIs that will definitely be applied. For each of these, we take the associated type from the resource context $\Delta$, associate it with a ghost variable $f'$, and check whether it is possible to type the arguments of this API with the same type as the parameters passed $y_i$. Finally we perform the disjunction on the various return types of each API in the set $API$, and in the active effect to be associated with the expression an event of type $call$ should be added in which we encapsulate the qualifier of all APIs that will be executed ($\psi$) and also those associated with the parameters ($\psi_i$).

The inclusion of type qualifiers in actions and API calls enables a more refined over-approximation. This is because, when inferring the latent effect of functions, if a resource (or value) used in actions (or API calls) is part of the function's arguments, we can use a smaller set of values than those defined in the function's type during its application. This results in a more precise approximation of the function's behaviour (remember how in Coverage Types the type of function arguments is an over-approximation type and in the subtype relation the notion of contravariance on the type of function arguments applies \cite{coverage}). This is the same reason why the syntax for $call$ was also introduced and not just the syntax for the single application of $API$ is used.

\subsection{Operational Semantic}

In Figure \ref{fig:semantic-op} we find the small-step rules of the operation semantics. These do not require comment as they are very self-explanatory. The only notable rule is the \textsc{StAppApi} where an oracle $\tilde{O}$ is used to predict the value returned by an external function, since we do not know the body of the function but only its signature.

\begin{figure}[H]
    \begin{equation}
        \frac{
            op \; \overline{v} \equiv v_y
        }{
            (\texttt{let}\;y = op\;\overline{v}\;\texttt{in}\;e, \; H) \hookrightarrow (e[y \mapsto v_y], \; H \cdot \epsilon)
        }
        \tag{\textsc{StAppOp}}
        \label{eq:stappop}
    \end{equation}
    
    \begin{equation}
        \frac{
            (e_1, \;\epsilon) \hookrightarrow (e_1', \;H_{e_1'})
        }{
            (\texttt{let}\;y = e_1\;\texttt{in}\;e_2, \; H) \hookrightarrow (\texttt{let}\;y = e_1'\;\texttt{in}\;e_2, \; H \cdot H_{e_1}')
        }
        \tag{\textsc{StLetE1}}
        \label{eq:stlete1}
    \end{equation}
    
    \begin{equation}
        \frac{
            \space
        }{
            (\texttt{let}\;y = v\;\texttt{in}\;e, \; H) \hookrightarrow (e[y \mapsto v], \; H \cdot \epsilon)
        }
        \tag{\textsc{StLetE2}}
        \label{eq:stlete2}
    \end{equation}
    
    \begin{equation}
        \frac{
            \space
        }{
            (\texttt{let}\;y = \lambda x{:}t.e_1 \; v_x \;\texttt{in}\; e_2, \; H) \hookrightarrow (\texttt{let}\;y = e_1[x \mapsto v_x]\;\texttt{in}\; e_2, \; H \cdot \epsilon)
        }
        \tag{\textsc{StLetAppLam}}
        \label{eq:stletapplam}
    \end{equation}
    
    \begin{equation}
        \frac{
            \space
        }{
            \begin{gathered}
                (\texttt{let}\;y = \;\texttt{fix}\;f{:}t.\lambda x{:}t_x.e_1 \; v_x \;\texttt{in}\; e_2, \;H) \hookrightarrow \\
                (\texttt{let}\;y = (\lambda f{:}t.e_1[x \mapsto v_x])\;(\texttt{fix}\;f{:}t.\lambda x{:}t_x.e_1) \;\texttt{in}\;e_2, \;H \cdot \epsilon)
            \end{gathered}
        }
        \tag{\textsc{StLetAppFix}}
        \label{eq:stletappfix}
    \end{equation}

    \begin{equation}
        \frac{
            \space
        }{
            (\texttt{match}\;d_i\;\overline{v_j}\;\texttt{with}\;\overline{d_i\;\overline{y_j} \rightarrow e_i}, \; H) \hookrightarrow (e_i[\overline{y_j \mapsto v_j}], \;H \cdot \epsilon)
        }
        \tag{\textsc{StMatch}}
        \label{eq:stmatch}
    \end{equation}

    \begin{equation}
        \frac{
            \alpha \; \overline{v} \equiv v_y
        }{
            (\texttt{let}\;y = \alpha\;\overline{v}\;\texttt{in}\;e, \; H) \hookrightarrow (e[y \mapsto v_y], \; H \cdot \alpha(\overline{v}))
        }
        \tag{\textsc{StAppAction}}
        \label{eq:stappaction}
    \end{equation}

    \begin{equation}
        \frac{
            \nextd(\Delta) = X
        }{
            (\texttt{let}\;y = new_r ()\;\texttt{in}\;e, \; H) \hookrightarrow (e[y \mapsto X], \; H \cdot new_r(X))
        }
        \tag{\textsc{StNew}}
        \label{eq:stnew}
    \end{equation}

    \begin{equation}
        \frac{
            \Delta(F) = \tau_F
        }{
            (\texttt{let}\;y = \texttt{get}\;F\;\texttt{in}\;e, \; H) \hookrightarrow (e[y \mapsto F], \; H \cdot get(F))
        }
        \tag{\textsc{StGet}}
        \label{eq:stget}
    \end{equation}

    \begin{equation}
        \frac{
            \Delta(F) = \tau_F
        }{
            (\texttt{let}\;y = F\;\overline{v_i}\;\texttt{in}\;e, \; H) \hookrightarrow (e[y \mapsto \tilde{O}(F, \;\overline{v_i})], \; H \cdot F(\overline{v_i}))
        }
        \tag{\textsc{StAppApi}}
        \label{eq:stappapi}
    \end{equation}
    \caption{Semantics Operation of $\lambda^{\textbf{TG}}$ extended with History Expressions.}
    \label{fig:semantic-op}
\end{figure}

\subsection{A Case Study}

An example code generated by the language presented above is shown in Figure \ref{fig:exmp3}. Let us assume that we wish to construct a generator of all palindromic words of length less than $n$. Instead of implementing it from scratch, in our resource context we have an external API that generates all palindromes words exactly $n$ long, so we decided to implement our generator in terms of an existing one. Let us therefore assume that we have the following resource context (in reality $\Delta$ will also contain all the resources already present in the reference system, for simplicity and because it is of no interest to us we will not list them all):

\begin{equation}
    \begin{gathered}
        \Delta \equiv \{(F, \;n{:}\overa{v: nat \;|\; \top_{nat}} \rightarrow p{:}\overa{v: string \;|\; \top_{string}} \rightarrow \\
        (\under{v: string \;|\; len(v, \;n) \;\land\; is\_palindrome(v)}, \\
        open(string: v = p) \cdot write(string: v = p, \;string: len(v, \;n) \;\land\; is\_palindrome(v))))\}
    \end{gathered}
\end{equation}

For simplicity, we will assume that we also have polymorphic file primitives in our language $\lambda^{\textbf{TG}}$. In this case since the $F$ API works on file paths we will use primitives with the same name but which take as input a string instead of a resource identifier. This is intuitively correct since a path also uniquely identifies a file (or generally a resource).

Obviously in our language it is more advantageous to work with resource identifiers as we can know whether a resource is referable or not, but with paths this is not possible, we cannot know statically whether a path is correct or not, as our code could dynamically modify folders and files.

The problem, however, is that the API identified by $F$ also takes as a parameter a string representing the absolute path to a file where the function, in addition to returning the generated word, will also write it to this file.

\begin{wrapfigure}{l}{0.65\textwidth}
    \begin{lstlisting}
let palindromes_gen (n: nat) : string =
    let palindrom = ?\typeact{get}? F
    in
    let random = nat_gen ()
    in
    let l = random ?\typeact{mod}? n
    in
    palindrom l "/home/angelopassa/results.txt"
    \end{lstlisting}
    \caption{Generator of \emph{all} palindromic words of length less than $n$.}
    \label{fig:exmp3}
\end{wrapfigure}

However, we would like the function not to be passed sensitive files that should neither be read nor written.

Through the typing rules presented in this Section, we can assign the following type to the function \verb|palindromes_gen| in Figure \ref{fig:exmp3}:

\vspace{25pt}

\begin{equation}
    \begin{gathered}
         n{:}\overa{n: nat \;|\; \top_{nat}} \rightarrow (\under{v: string \;|\; is\_palindrome(v) \;\land\; \forall m. \;len(v, \;m) \Longrightarrow m < n}, \\
         get(F) \cdot F(n{:}(nat: v < n), \;p{:}(string: v = "/home/angelopassa/results.txt")))
    \end{gathered}
\end{equation}

So as mentioned, we would like to check that no strings representing sensitive file paths, such as \verb|/etc/shadow|, are passed to the $F$ API. To do this, we simply define a new policy as follows (denoting by $H$ the latent effect of the function \verb|palindromes_gen|):

\begin{equation}
    \forall \eta \in \llbracket H \rrbracket. \;\forall n{:}nat. \;\forall p{:}string. \;F(n, \;p) \in \eta \Longrightarrow \neg(p = "/etc/shadow")
\end{equation}

We can reuse the same logical formula to define the same type of policy for each of the sensitive files that concern us.

\subsection{Soundness}

Surprisingly, in this part we will not present any denotation of either the full type pair or the Extended Coverage Types. As mentioned above, there is no need to introduce new denotations since structurally it comes easy to separate types as needed. In the following we will therefore present only three theorems that in the aggregate will allow us to be able to assert the correctness of our type system:
\begin{enumerate*}[label=(\roman*)]
    \item First we will show that the correctness of the original type system remains unchanged in this extension. In essence, Theorem 4.3 presented in the original work \cite{coverage} is preserved.
    \item Next we will show that the qualifier $\psi$ presented in the Extended Coverage Types actually represents \textbf{\emph{all and alone}} the values returned by the expression, or function as the case may be.
    \item Finally, we present the type system correctness demonstration for the part concerning History Expressions.
\end{enumerate*}

\begin{theorem}[Correctness of Extension on Coverage Types]
    For each expression $e$, if one can infer a type for $e$ through the type system presented, i.e. $\Gamma \vdash e : (\under{v: b \;|\; \phi \;|\; \psi}, \;H)$, then through the original type system of $\lambda^{\textbf{TG}}$ it is possible to infer $\Gamma_{\phi} \vdash e : \under{v: b \;|\; \phi}$.
\end{theorem}

\begin{proof}
    Let us proceed rule by rule:
    \begin{itemize}
        \item For the rules \textsc{TErrr}, \textsc{TConst}, \textsc{TOp}, \textsc{TAction}, \textsc{TVarBase} and \textsc{TVarFun} the correspondence with the original rules is one-to-one. For \textsc{TErr} and \textsc{TVarBase} the demonstration is immediate. For rules using the function $Ty$ the proof depends on the construction of the function, but we have already mentioned how the qualifier $\phi$ is simply replicated a second time, consequently the theorem holds in these cases as well. For the rule \textsc{TVarFun}, the filtering operator $\phi$ applied to the context $\Gamma$, simply eliminates the histories and the $\psi$ qualifier while unchanging the contained variables, the Coverage Types qualifier and the function types qualifier.
        \item Also for the rules \textsc{TSub} and \textsc{TEq} the correspondence with the original rules is bi-univocal, here the premises are structurally identical, only the relation $<:$ that has been redefined changes. In the relation, however, only two new premises have been added, one on the denotation of the History Expressions and the other on the equality of the second qualifiers $\psi$. In essence, it is trivially true that:
        \begin{equation}
            \begin{gathered}
                \llbracket \under{v: b \;|\; \phi_1} \rrbracket_{\Gamma_{\phi}} \subseteq \llbracket \under{v: b \;|\; \phi_2} \rrbracket_{\Gamma_{\phi}} \;\land\; \llbracket H_1 \rrbracket_{\Gamma_{\psi}} \subseteq \llbracket H_2 \rrbracket_{\Gamma_{\psi}} \;\land\; \llbracket \under{v: b \;|\; \psi_1}\rrbracket_{\Gamma_{\psi}} = \llbracket \under{v: b \;|\; \psi_2} \rrbracket_{\Gamma_{\psi}} \\
                \Longrightarrow \\
                \llbracket \under{v: b \;|\; \phi_1} \rrbracket_{\Gamma_{\phi}} \subseteq \llbracket \under{v: b \;|\; \phi_2} \rrbracket_{\Gamma_{\phi}}
            \end{gathered}
        \end{equation}
        \item The rules \textsc{TFun} and \textsc{TFunFlat} do not alter the premises, moreover by eliminating the effects the rule \textsc{TFunFlat} is identical to the former.
        \item Even the rule \textsc{TMatch}, short of universal quantification on each branch, does not alter the premises, for although the type of $v_i$ and $d_i(\overline{y})$ are different, the addition of the condition using the operator $\langle \cdot \rangle$ places the constraint on the equality of the first qualifiers ($\phi$). Disjunction, as with the subtyping relationship, is also not altered for Coverage Types and only constraints are added to the original ones. The only difference is the fact that while in our type system it is possible to type a \verb|pattern-matching| in its entirety with one rule taking into account each branch, in the original one it is necessary to apply cascading \textsc{TMerge} on each pair of branches.
        The \textsc{TMerge} rule also does not alter the assumptions about the original Coverage Types.
        \item All rules on \verb|let-binding| constructs do not alter the structure of the original rules since the type of the function argument is always matched with the first qualifier of the type of the parameter passed, and the same substitutions are made in the contexts.
        \item The last rule, the \textsc{TFix}, is the most cumbersome to prove. Essentially, we add premises concerning History Expressions that are out of our interest in this demonstration, and then in place of encapsulating the expression $e$ in a $\lambda$ that takes as its argument the recursive function itself \cite{coverage}, we add to the environment the first parameter of the recursion $a$ and a variable $f$ of type $api$ that is associated via the resource context $\Delta$ with the very type we wish to infer. We note how we have thus only changed the way we perform the same procedure that leads us to the same result.
    \end{itemize}

    The same property applies for the typing of a function\footnote{We need to apply the operator $\phi$ to the entire function type - and not just remove the qualifier $\psi$ from the return type - since there may be other functions in the parameters.}:
    \begin{equation}
        \Gamma \vdash e : (\overline{\tau_i} \rightarrow (\under{v: b \;|\; \phi \;|\; \psi}, \;H_f), \;H_e) \Longrightarrow \Gamma_{\phi} \vdash e : (\overline{\tau_i} \rightarrow \under{v: b \;|\; \phi \;|\; \psi})_{\phi}
    \end{equation}
\end{proof}

\begin{theorem}[Correctness of Extended Coverage Types]
    For each expression $e$, if $\Gamma \vdash e : (\under{v: b \;|\; \phi \;|\; \psi}, \;H)$ then $\psi$ will represent \textbf{\emph{all and only}} the values returned by $e$.
\end{theorem}

\begin{proof}
    As we have already demonstrated, the function $Ty$ returns a Coverage Type (constants), or a function with a Coverage Type as return type (operators), which respect this property. In fact, the qualifier $\phi$ will always be identical to $\psi$, the latter being merely replicated to adapt to the new language type.
    For this demonstration, we can proceed inductively on the type system rules shown:
    \begin{itemize}
        \item For the rules \textsc{TConst}, \textsc{TAction}, \textsc{TOp} and \textsc{TErr} the demonstration is immediate from the above.
        \item For the rule \textsc{TVarFun}, the property is satisfied since the type of the variable $x$, already present in the context, is assumed to satisfy the property, and this type is the same as that present in the rule's consequences.
        \item Also for the rule \textsc{TVarBase} the property is satisfied since even if we add the same qualifier $v = x$ in all two positions of the type, when computing the denotation of the type of $x$ we will take the context $\Gamma$ by filtering the types through the second qualifier $\psi$, which by definition represents \emph{all and only} the values returned by the expression.
        \item The rule \textsc{TSub} retains the property since, if $\pi$ respects it, according to the relation $<:$ defined in the auxiliary rules, the denotation of the qualifiers $\psi$ also remains unchanged in $\pi'$, which will retain the property. The rule \textsc{TEq} does not alter the semantics of the entire type pair, so the property is retained.
        \item For the rule \textsc{TMatch}, by inductive assumption each branch $e_i$ will have associated an Extended Coverage Type (or a function that returns it), by performing the disjunction this will also be performed on the second qualifier $\psi$, thus guaranteeing for the \verb|pattern-matching| a type that also indicates \emph{all and only} the returned values.
        By the \textsc{TMerge} rule, by inductive hypothesis both $\tau_1$ and $\tau_2$ will respect the property, i.e. the qualifiers in the second field ($\phi$) will be semantically identical. Applying the disjunction between them takes the intersection between two sets (the denotations of $\tau_1$ and $\tau_2$) that are equal. The new qualifier $\phi$ belonging to $\tau$ is practically identical.
        \item The rule, \textsc{TFix}, has the same behaviour as the rules on functions already seen, essentially the type $\tau$ in the consequences is the same as in the premises, and a Coverage Type, $f$, is added to the context $\Gamma$, which trivially respects the property.
        \item The remaining rules are all on \verb|let-bindings| constructs, so the type $\tau$ in the premises is the same in the consequences. Furthermore, in rules such as \textsc{TNew} and \textsc{TGet} Coverage Types respecting the property are added to the context $\Gamma$; in the remaining ones, which are essentially applications of functions, the type of the passed parameters is an extended Coverage Type respecting the property (being in the premises), so replacing this in the return type of the function added to $\Gamma$ with the variable $x$ does not alter the property.
    \end{itemize}

    The same theorem also applies to the typing of a function: $\Gamma \vdash e : (\overline{\tau_i} \rightarrow (\under{v: b \;|\; \phi \;|\; \psi}, \;H_f), \;H_e)$, where $\psi$ represents \emph{\textbf{all and only}} the values returned by the function, based also on the parameters $\tau_i$.
\end{proof}

\begin{theorem}[Correctness of History Expressions]
    For every expression $e$, if $\Gamma \vdash e : (\tau, \;H)$ then for every $\eta$ such that $(e, \;\epsilon) \hookrightarrow^* (v, \;\eta)$ it holds that $\eta \in \llbracket H \rrbracket_{\Gamma}$.
\end{theorem}

\begin{proof}
    As we have already done, we proceed inductively on the rules of the type system, showing that if the property applies to expressions in the premises, it also applies to those in the consequences:
    \begin{itemize}
        \item For the rules \textsc{TConst}, \textsc{TOp}, \textsc{TAction}, \textsc{TVarBase} and \textsc{TVarFun} the verification is immediate as there is no reduction in the rules in Figure \ref{fig:semantic-op} and starting from a history $\epsilon$ this is also maintained in the static semantics.
        \item For the rule \textsc{TSub}, assuming for the sake of argument that the History Expression, which we will call $H$, present in $\pi$ is correct, the subtype relation $<:$ simply enlarges the denotation of H by changing it to an H', formally it holds that:
        \begin{equation}
            \forall \;\Gamma, \;H, \;H', \;\eta. \;\llbracket H \rrbracket_{\Gamma_{\psi}} \subseteq \llbracket H'\rrbracket_{\Gamma_{\psi}} \;\land\; \eta \in \llbracket H \rrbracket_{\Gamma_{\psi}} \Longrightarrow \eta \in \llbracket H' \rrbracket_{\Gamma_{\psi}}
        \end{equation}
        The rule \textsc{TEq} does not alter the type denotation.
        \item For the rules \textsc{TFun} and \textsc{TFunFlat}, the argument is identical to that already made for the first point of the demonstration. Even for the recursive function statement alone (\textsc{TFix}) the empty history $\epsilon$ is the one generated this because according to the reduction rules in Figure \ref{fig:denot_history} the construct $\mu$ adds the appropriate substitution to $\Upsilon$, but reduces to $\epsilon$.
        \item For the rule \textsc{TMatch}, $v$ is a value, consequently as already demonstrated $\epsilon \in \llbracket H_{v_i} \rrbracket_{\Gamma_{\psi}}$ for each $i$. The same applies to the application of the constructors $d_i$ which, being part of the operators, their application will lead to the application of \textsc{TSub} to the empty history only. The disjunction leads to a situation in which for each $i$ it is true that $\llbracket H_{v_i} \cdot H_{y_i} \cdot H_{e_i} \rrbracket_{\Gamma_{\psi}} \subseteq \llbracket H \rrbracket_{\Gamma_{\psi}}$. Having already noted that $\epsilon \in \llbracket H_{v_i} \rrbracket_{\Gamma_{\psi}}$ and $\epsilon \in \llbracket H_{y_i} \rrbracket$ we can affirm that $\llbracket \epsilon \cdot H_{e_i} \rrbracket_{\Gamma_{\psi}} \subseteq \llbracket H \rrbracket_{\Gamma_{\psi}}$ for each $i$. To conclude, the operational rule \textsc{StMatch} reduces in one of the branches $e_i$ producing as effect $\epsilon$ (starting from $H = \epsilon$). But by inductive hypothesis $H_{e_i}$ is a correct over-approximation of $e_i$, and by eliminating the $\epsilon$ through the equality axioms presented in Figure \ref{fig:hist_eq} we arrive at stating that $\llbracket H_{e_i} \rrbracket_{\Gamma_{\psi}} \subseteq \llbracket H \rrbracket_{\Gamma_{\psi}}$ for each $i$. We conclude by noting that one of $H_{e_i}$ is the correct over-approximation of the expression, but each of these belongs to $H$. The \textsc{TMerge} rule is actually similar to the \textsc{TSub} rule, in fact since both $H_1$ and $H_2$ are two correct over-approximations, the disjunction will take the larger of the two.
        \item Now we focus on the rules \textsc{TAppMulti} and \textsc{TAppLast} and $v_1$ will be either a normal function or a recursive function. Accordingly, we will take the rules \textsc{StLetAppLam} and \textsc{StLetAppFix} as reference in operational semantics. The same reasoning presented for the rule \textsc{TMatch} is used here, whereby the effects associated with $v_1$ and $v_2$ will surely contain in their denotations $\epsilon$. If the parameter $v_2$ applied is not the last of the function $v_1$, it will mean that the expression within the $\lambda$-abstraction is another function, and being a value it will produce no effect (relevant, $\epsilon$ excluded) by proceeding in the reduction immediately with $e_2$ (\textsc{StLetAppLam} or \textsc{StLetAppFix} and then \textsc{StLetE2}). For the rule \textsc{TAppMulti} we can end here since the effect $H_e$ will be correct by inductive hypothesis. By ending with the \textsc{TAppLast}, on the other hand, we can place the latent effect within the type of the behaviour since in this case $e_1$ will not be a value and it will be necessary to apply the rule \textsc{StLetE1} a given number of times. The history generated by the first reduction of $e_1$ until we proceed with $e_2$ will surely belong to the denotation of $H_{\tau_{v_1}}$ by inductive hypothesis.
        \item For the rules \textsc{TNew} and \textsc{TGet} the proof is straightforward since the concatenated history before $H_e$ is exactly the one found in the rules of semantic operation, where the premises on $\Delta$ also match.
        \item For the rule \textsc{TLetAction}, remembering the demonstration previously made on $\epsilon$ which in this case will definitely belong to the denotation of $H_{\alpha}$ and of each $H_{u_i}$ it is necessary to verify that $\llbracket \alpha(\overline{v_i}) \rrbracket_{\Gamma_{\psi}} \subseteq \llbracket \alpha(\overline{b_i{:}\psi_i}) \rrbracket_{\Gamma_{\psi}}$. This is straightforward because we have already proved in the previous theorem that $\psi_i$ represents \emph{all and only} the values into which each $u_i$ could reduce, there will be no others. Consequently, for each $i$, $\psi[v \mapsto v_i]$ holds, and in relation to the rules in Figure \ref{fig:denot_history} we can state that the set-inclusion relation is verified.
        \item Finally, we conclude with the rule \textsc{TAppAPI}. Here we proceed in a mirrored manner to the demonstration on the rule \textsc{TLetAction}, only the relation to be proved changes, which becomes $\llbracket F(\overline{v_i}) \rrbracket_{\Gamma_{\psi}} \subseteq \llbracket call(\psi; \;\overline{a_i{:}(b_i{:}\psi_i})) \rrbracket_{\Gamma_{\psi}}$. But even here, thanks to the reduction relation of the History Expressions and the qualifier $\psi$ of the Extended Coverage Types it is easy to verify that it is satisfied.
    \end{itemize}
\end{proof}

\section{Algorithmic Constructions \& Properties}\label{sec:algorithm}
In the original work on Coverage Types, a number of algorithms for typing were also presented in order to have semantics-independent rules that are easy to implement structurally.

\textcolor{Green}{Type synthesis} and \textcolor{purple}{type check} rules were presented as substitutes for type system rules, and algorithms were also proposed to achieve well-formedness, to check subtyping relationships, and to obtain disjointed types.

\subsection{Auxiliary Typing Algorithms}

The algorithm to obtain well-formedness for History Expressions, to be used in conjunction with \textsc{\textcolor{blue}{Ex}} and \textsc{\textcolor{blue}{Forall}}, does not need to be implemented, it is sufficient to use \bind, as the context-dependent denotation is defined by this exact function.

Even for the \textsc{\textcolor{blue}{Disj}} algorithm, which implements the disjunction, we have already seen how the intersection between the denotations of two History Expressions leads to a third history corresponding to the non-deterministic choice of the previous ones.

\begin{equation}
    H_1 + H_2 = H_3 \Longrightarrow \llbracket H_1 \rrbracket_{\Gamma_{\psi}} \cup \llbracket H_2 \rrbracket_{\Gamma_{\psi}} = \llbracket H_3 \rrbracket_{\Gamma_{\psi}}
\end{equation}

Thus, the only algorithms that need to be implemented from zero for History Expressions concern only the subtyping relation, to be used with the \textsc{\textcolor{blue}{Query}}, and the conjunction \textsc{\textcolor{blue}{Conj}}.

In order to structurally verify the subtyping relationship for histories, it is necessary that the histories themselves comply with general construction rules. We therefore introduce the Normal Form for History Expressions.

\begin{definition}[History Expression in Normal Form]
    A History Expression $H^{\sharp}$ is said to be in Normal Form if it respects the following properties:
    \begin{itemize}
        \item $H^{\sharp} = \underset{i}{\bigoplus} \; H^c_i, \; \forall i. \;H^c_i = \underset{j}{\bullet} \; H^T_j$
        \item $\forall H^c_i, \;\exists H,\; H^c_i = H \cdot \epsilon \;\land\; \epsilon \notin H$
        \item $\forall H^c_i, \;\forall \mu F(\overline{a{:}(b{:}\phi)})(H_{rec}) \in H^c_i, \; \mathcal{NF}(H_{rec}) \;\land\; \exists H', H'', H''', \\
        H^c_i = H' \cdot \mu F(\overline{a{:}(b{:}\phi)})(H_{rec}) \cdot H'' \cdot H''' \;\land\; (H'' = \epsilon \;\lor\; (H'' = call(\psi; \;\overline{c{:}(b{:}\theta})) \;\land\; \phi[v \mapsto F]))$
    \end{itemize}
    With $H^T$ defined as:
    \begin{equation}
        H^T := \epsilon \;|\; \alpha(\overline{r{:}\phi}) \;|\; new_r(X) \;|\; get(F) \;|\; call(\phi; \;\overline{a{:}(b{:}\psi)}) \;|\; \mu F(\overline{a{:}(b{:}\phi)})(H_F)
    \end{equation}
    \label{def:nf}
\end{definition}

That is, 
\begin{enumerate*}[label=(\roman*)]
    \item $H^{\sharp}$ must be written in extended form as a series of deterministic choices within which only the expressions defined in $H^T$ must be present in concatenation.
    \item Next, we find that for each of these concatenations, $\epsilon$ must be present \textbf{only} at the end.
    \item Finally, for each statement of a recursive function within each concatenation, we want the effect associated with the recursive function to be in Normal Form, and for the statement to be shifted as far to the right as possible, so that either it is the penultimate expression, before $\epsilon$, or immediately afterwards there is a potential call to the function itself.
\end{enumerate*}

We can now show the \subhist\ procedure for checking the subtyping relationship between two History Expressions $H_1$ and $H_2$ which must fulfil two conditions:

\begin{itemize}
    \item Being in Normal Form.
    \item Be in well-formedness with respect to the empty context, i.e. all logical formulae within them must not contain free variables.
\end{itemize}

The \subhist\ procedure is shown in Algorithm \ref{alg:subhist}, and is essentially concerned with checking for each concatenation present in $H_1$ whether there is another one belonging to $H_2$, a subtype of the first. The verification of the subtype relation between concatenations is shown in the Algorithm \ref{alg:subconc} called \subconc. Here we proceed by running the two concatenations on each expression in parallel, checking whether the qualifiers in the first one imply the counterparts in the second. In other words, it is checked whether the values satisfying the qualifiers in the expression of the first history are a subset of those satisfying the qualifiers of the second.

It is good to make an observation about the case of subtyping recursive functions: it is not necessary that the identifier of the two declared functions be the same since, as we have seen during the calculation of denotation, the application of a recursive function is not \emph{"saved"}, so essentially we check that, once we have bound all the parameters to the effects of the two functions, the first is subtyped by the second. It is necessary for $H_{res_2}$ and $H_{tail_2}$ to replace the occurrences of $G$ with $F$ otherwise the subtype relation of the $call$ inside may be skipped.

\begin{algorithm}[ht]
    \caption{History Expression Subtyping}\label{alg:subhist}
    \begin{algorithmic}[1]
        \Procedure{\textcolor{Mahogany}{SubHist}}{$H_1$, $H_2$}
            \Match{$H_1$}
                \Case{$H^c$}
                    \For{$H^c_{2_i} \in H_2$}
                        \If{\Call{\textcolor{Mahogany}{SubConc}}{$H^c, \;H^c_{2_i}$}}
                            \State \Return \True
                        \EndIf                       
                    \EndFor
                    \State \Return \False
                \EndCase
                \Case{$H^c + H_{tail}$}
                    \For{$H^c_{2_i} \in H_2$}
                        \If{\Call{\textcolor{Mahogany}{SubConc}}{$H^c, \;H^c_{2_i}$}}
                            \State \Return \Call{\textcolor{Mahogany}{SubHist}}{$H_{tail}$, $H_2$}
                        \EndIf                       
                    \EndFor
                    \State \Return \False
                \EndCase
            \EndMatch
        \EndProcedure
    \end{algorithmic}
\end{algorithm}

\begin{algorithm}[H]
    \caption{Concatenation Subtyping}\label{alg:subconc}
    \begin{algorithmic}[1]
        \Procedure{\textcolor{Mahogany}{SubConc}}{$H_1$, $H_2$}
            \Match{$H_1$, $H_2$}
                \Case{$\epsilon$, $\epsilon$}
                    \State \Return \True
                \EndCase
                \Case{$\alpha(\overline{r_i{:}\phi_i}) \cdot H_{tail_1}$, $\alpha(\overline{r_i{:}\psi_i}) \cdot H_{tail_2}$}
                    \State $Res \gets \forall i, \;\forall u{:}r_i, \; \phi[v \mapsto u] \Longrightarrow \psi[v \mapsto u]$
                    \State \Return $Res \;\land\; \Call{\textcolor{Mahogany}{SubConc}}{H_{tail_1}, \;H_{tail_2}}$
                \EndCase
                \Case{$new_r(X) \cdot H_{tail_1}$, $new_r(X) \cdot H_{tail_2}$}
                    \State \Return $\Call{\textcolor{Mahogany}{SubConc}}{H_{tail_1}, \;H_{tail_2}}$
                \EndCase
                \Case{$get(F) \cdot H_{tail_1}$, $get(F) \cdot H_{tail_2}$}
                    \State \Return $\Call{\textcolor{Mahogany}{SubConc}}{H_{tail_1}, \;H_{tail_2}}$
                \EndCase
                \Case{$call(\phi; \;\overline{a_i{:}(b_i{:}\theta_i})) \cdot H_{tail_1}$, $call(\psi; \;\overline{c_i{:}(b_i{:}\sigma_i})) \cdot H_{tail_2}$}
                    \State $Api \gets \forall u{:}api, \;\phi[v \mapsto u] \Longrightarrow \psi[v \mapsto u]$
                    \State $Params \gets \forall i, \;\forall u{:}b_i, \;\theta[v \mapsto u] \Longrightarrow \sigma[v \mapsto u]$
                    \State \Return $Api \;\land\; Params \;\land\; \Call{\textcolor{Mahogany}{SubConc}}{H_{tail_1}, \;H_{tail_2}}$
                \EndCase
                \Case{$\mu F(\overline{a_i{:}(b_i{:}\phi_i)})(H_{res_{1}}) \cdot H_{tail_1}$, $\mu G(\overline{c_j{:}(t_j{:}\psi_j)})(H_{res_2}) \cdot H_{tail_2}$}
                    \State $Body \gets \textsc{\textcolor{Mahogany}{SubHist}}(\bind(\overline{a_i{:}\overa{v{:}b_i \;|\; \phi_i}}, H_{res_1}),$
                    \State \hspace{7.5em}$\bind(\overline{c_j{:}\overa{v{:}t_j \;|\; \psi_j}}, H_{res_2}[F/G]))$
                    \State \Return $Body \;\land\; \Call{\textcolor{Mahogany}{SubConc}}{H_{tail_1}, \;H_{tail_2}[F/G]}$
                \EndCase
                \Default
                    \State \Return \False
                \EndCase
            \EndMatch
        \EndProcedure
    \end{algorithmic}
\end{algorithm}

\begin{theorem}[Correctness of Procedure \subhist]
    \begin{equation*}
        \forall \;H_1, H_2. \;\subhist(H_1, \;H_2) \Longrightarrow \varnothing \vdash H_1 <: H_2
    \end{equation*}
\end{theorem}

\begin{proof}
    We first check that the \subconc\ procedure is correct.

    In general, we can state that if we have for two logical formulae a relation of the kind:
    \begin{equation}
        \phi \Longrightarrow \psi
    \end{equation}
    It means that all values satisfying the first qualifier are also contained in the second:
    \begin{equation}
        \{u \;|\; \phi[v \mapsto u]\} \subseteq \{u \;|\; \psi[v \mapsto u]\}
    \end{equation}
    Now let us call the first set $A$ and the second set $B$, since the first is contained in the second, there will exist a third set $C$ such that $B = A \cup C$, and thus there will also exist a third qualifier $theta$ such that the following relations are satisfied:
    \begin{equation}
        \begin{gathered}
            A \subseteq A \cup C \\
            \phi \Longrightarrow \phi \;\lor\; \theta
        \end{gathered}
    \end{equation}
    The subtype relationship can also be rewritten as:
    \begin{equation}
        \varnothing \vdash \alpha(b{:}\phi) <: \alpha(b{:}\phi \;\lor\; \theta)
    \end{equation}
    But the reduction relation for actions presented in Figure \ref{fig:denot_history} calculates exactly all values that satisfy the predicate and place each value in non-deterministic choice, so we can rewrite the previous relation as:
    \begin{equation}
        \varnothing \vdash \alpha(b{:}\phi) <: \alpha(b{:}\phi) + \alpha(b{:}\theta)
    \end{equation}
    And here the demonstration ends as by switching to denotations the relationship is correct:
    \begin{equation}
        \llbracket \alpha(b{:}\phi) \rrbracket \subseteq \llbracket \alpha(b{:}\phi) \rrbracket \cup \llbracket \alpha(b{:}\theta) \rrbracket
    \end{equation}
    For the other cases, the demonstration is trivial or mirrors the one for actions as for external function calls.

    Let us finish by showing the correctness then of the \subhist\ procedure, in this case $H_1$ and $H_2$ being in Normal Form we can rewrite the subtyping relation as:
    \begin{equation}
        \varnothing \vdash H_1^{c_1} + H_1^{c_2} + \dots + H_1^{c_n} <: H_2^{c_1} + H_2^{c_2} + \dots + H_2^{c_m}
    \end{equation}
    And thus prove that:
    \begin{equation}
        \llbracket H_1^{c_1} \rrbracket \cup \llbracket H_1^{c_2} \rrbracket \cup \dots \cup \llbracket H_1^{c_n} \rrbracket \subseteq \llbracket H_2^{c_1} \rrbracket \cup \llbracket H_2^{c_2} \rrbracket \cup \dots \cup \llbracket H_2^{c_m} \rrbracket
    \end{equation}
    But having already proved the correctness of the procedure \subconc, we know that \textbf{for all} concatenation of $H_1$ \textbf{exists} at least one concatenation $H_2$ supertype of this one:
    \begin{equation}
        \forall i. 1 \leq i \leq n. \;\exists j. 1 \leq j \leq m. \;\llbracket H_1^{c_i} \rrbracket \subseteq \llbracket H_2^{c_j} \rrbracket
    \end{equation}
    Using the set operator of union on both the right and the left side of the subset relation we have that the denotation of $H_1$ will be contained in the denotation of another history which will in turn be contained in the denotation of $H_2$ (this is because there is the existential quantifier on the concatenations of $H_2$). Thus $H_1$ will be a subtype of $H_2$.
\end{proof}

We cannot, however, assert the \emph{completeness} of the algorithm \subconc, as a property of minimality would have to be guaranteed and satisfied, both with regard to the unfolding of the recursion $\mu$, but also with regard to the qualifiers present in the $call$ and the $\alpha$ actions.

However, it would be possible to further refine the Normal Form by introducing other constraints, but for our purposes this is not necessary at present.

Figure \ref{fig:alg-subtyping} shows the auxiliary inference rules for subtyping with the appropriate modifications.

\begin{figure}[H]
    \begin{multicols}{2}
        \begin{equation*}
            \frac{
                \models \Call{\textcolor{blue}{Query}}{\Gamma, \;\under{v: b \;|\; \bot}, \;\under{v: b \;|\; \phi}}
            }{
                \texttt{err} \notin \llbracket \under{v: b \;|\; \phi \;|\; \psi} \rrbracket_{\Gamma}
            }
        \end{equation*}
    
        \begin{equation*}
            \frac{
                \models \Call{\textcolor{blue}{Query}}{\Gamma, \;\under{v: b \;|\; \phi_2}, \;\under{v: b \;|\; \phi_1}}
            }{
                \Gamma \vdash \overa{v: b \;|\; \phi_1 \;|\; \psi_1} <: \overa{v: b \;|\; \phi_2}
            }
        \end{equation*}
    \end{multicols}
    
    \begin{equation*}
        \frac{
            \begin{gathered}
                \models \Call{\textcolor{blue}{Query}}{\Gamma, \;\under{v: b \;|\; \phi_1}, \;\under{v: b \;|\; \phi_2}} \\
                \models \Call{\textcolor{blue}{Query}}{\Gamma, \;\under{v: b \;|\; \psi_1}, \;\under{v: b \;|\; \psi_2}} \quad \models \Call{\textcolor{blue}{Query}}{\Gamma, \;\under{v: b \;|\; \psi_2}, \;\under{v: b \;|\; \psi_1}} \\
                H_1^{\star} = \bind(\Gamma_{\psi}, \;H_1) \quad H_2^{\star} = \bind(\Gamma_{\psi}, \;H_2) \quad H^{\sharp}_1 = H^{\star}_1 \quad H^{\sharp}_2 = H^{\star}_2 \quad \models \subhist(H^{\sharp}_1, \;H^{\sharp}_2) \\
            \end{gathered}
        }{
            \Gamma \vdash (\under{v: b \;|\; \phi_1 \;|\; \psi_1}, \;H_1) <: (\under{v: b \;|\; \phi_2 \;|\; \psi_2}, \;H_2)
        }
    \end{equation*}
    \caption{Auxiliary subtyping rules.}
    \label{fig:alg-subtyping}
\end{figure}

Finally, we present the Algorithm \ref{alg:conj-hist}, called \histconj, which allows us to find a History Expression common to two histories $H_1$ and $H_2$. This algorithm makes it possible to implement the conjunction of histories, used when a latent effect belongs to a function that represents an argument of another outer function.

The algorithm fixes each concatenation of $H_1$ and for each other concatenation of $H_2$, it verifies which of the two fixed concatenations is contained in the other (via the \subconc\ procedure) by taking the smaller of the two. All these common histories are each time placed in non-deterministic choice since both $H_1$ and $H_2$ will already be in Normal Form.

The special symbol $\bot$ is used to indicate when a History Expression common to the two could not be found. This may occur because there is no universal subtype (let alone supertype) for the History Expressions presented.

When the history $H_{conj}$ returned by the procedure is equal to $\bot$, the typing procedure will fail.

\begin{algorithm}[H]
    \caption{Histories Conjunction}\label{alg:conj-hist}
    \begin{algorithmic}[1]
        \Procedure{\textcolor{Mahogany}{HistConj}}{$\Gamma$, $H_1$, $H_2$}
            \State $H_{conj} \gets \bot$
            \State $H_1 \gets \bind(\Gamma, \;H_1)$
            \State $H_1 \gets \bind(\Gamma, \;H_2)$
            \For{$H^c_{1_i} \in H_1$}
                \State $H_{common} \gets \bot$
                \For{$H^c_{2_j} \in H_2$}
                    \If{\Call{\textcolor{Mahogany}{SubConc}}{$H^c_{1_i}, \;H^c_{2_j}$}}
                        \If{$H_{common} = \bot$}
                            \State $H_{common} \gets H^c_{1_i}$
                        \Else
                            \State $H_{common} \gets H_{common} + H^c_{1_i}$
                        \EndIf
                    \ElsIf{\Call{\textcolor{Mahogany}{SubConc}}{$H^c_{2_j}, \;H^c_{1_i}$}}
                        \If{$H_{common} = \bot$}
                            \State $H_{common} \gets H^c_{2_j}$
                        \Else
                            \State $H_{common} \gets H_{common} + H^c_{2_j}$
                        \EndIf
                    \EndIf                       
                \EndFor
                \If{$H_{common} \neq \bot$}
                    \If{$H_{conj} \neq \bot$}
                        \State $H_{conj} \gets H_{conj} + H_{common}$
                    \Else
                        \State $H_{conj} \gets H_{common}$
                    \EndIf
                \EndIf
            \EndFor
            \State \Return $H_{conj}$
        \EndProcedure
    \end{algorithmic}
\end{algorithm}

\begin{theorem}[Correctness of Procedure \histconj]
    \begin{equation*}
        \forall \;\Gamma, H_1, H_2. \;\histconj(H_1, \;H_2) = H_3 \;\land\; \neg(H_3 = \bot) \Longrightarrow \Gamma \vdash H_1 \;\land\;H_2 = H_3
    \end{equation*}
\end{theorem}

\begin{proof}
    Since both $H_1$ and $H_2$ are in Normal Form, they will have the following structure:
    \begin{equation}
        \begin{gathered}
            H_1 = H_{1_1}^c + H_{1_2}^c + \dots + H_{1_n}^c \\
            H_2 = H_{2_1}^c + H_{2_2}^c + \dots + H_{2_m}^c
        \end{gathered}
    \end{equation}
    Having fixed a concatenation $H_{1_i}^c$ of $H_1$ we will go on to find the history with the widest denotation that is a subtype of both $H_{1_i}^c$ and $H_2$.
    In fact, at each $j$-th iteration of the second for loop, from the correctness of the procedure \subconc\, if by inductive hypothesis it results that:
    \begin{equation}
        \begin{gathered}
            H_{common} <: H_{1_i}^c \\
            H_{common} <: H_{2_1}^c + H_{2_2}^c + \dots + H_{2_j}^c
        \end{gathered}
    \end{equation}
    At the $j+1$-th iteration it holds that:
    \begin{equation}
        \begin{gathered}
            H^{\star} <: H_{1_i}^c \\
            H^{\star} <: H_{2_{j+1}}^c
        \end{gathered}
    \end{equation}
    Where $H^{\star}$ is the concatenation found to be common between the two fixed. And with the transition to denotations, it is trivial to arrive at the following result:
    \begin{equation}
        \begin{gathered}
            H_{common} + H^{\star} <: H_{1_i}^c \\
            H_{common} + H^{\star} <: H_{2_1}^c + H_{2_2}^c + \dots + H_{2_j}^c + H_{2_{j+1}}^c
        \end{gathered}
    \end{equation}
    By making the transition to denotations, we can state that the denotation of $H_{common}$ is contained in the denotation of the intersection of the two histories:
    \begin{equation}
        \llbracket H_{common} \rrbracket \subseteq \llbracket H_{1_i}^c \rrbracket \cap \llbracket H_{2_1}^c + H_{2_2}^c + \dots + H_{2_j}^c \rrbracket
    \end{equation}
    And that therefore $H_{common}$ is a subtype of the history having as denotation the exact intersection of the two. But there cannot be a history with a denotation larger than $H_{common}$ that is contained in the intersection, this because each time we add a concatenation to $H_{common}$ this is the largest possible common fixed a $H_{1_i}^c$ and a $H_{2_j}^c$! And thus we can conclude that $H_{common}$ corresponds to the intersection.

    Proceeding iteratively on each concatenation $H_{1_i}^c$ of $H_1$, we can conclude that the history contained in $H_{common}$ will correspond to the intersection of $H_1$ and $H_2$.
\end{proof}

The extensions to the original auxiliary algorithms and the \textcolor{Green}{type synthesis} and \textcolor{purple}{type check} rules are given in Appendix \ref{app:alg} as the changes made result at this point quite intuitive.

\section{Discussion}\label{sec:discussion}
This paper addressed the issues related to the use of \emph{Coverage Type} \cite{coverage}. These, based on a dual logic to that of \emph{Hoare}, the \emph{Incorrectness Logic}, and relying on a type structure similar to \emph{Refinement Types} \cite{refinement}, allow the definition of a type that provides guarantees of completeness.

These guarantees may be defined through the use of \emph{method predicates}, i.e. uninterpreted functions that take as input one or more values associated with a given base type, also structured as lists or trees, and return a truth value based on the satisfiability of the property by the values passed. The \emph{method predicates} are used within the type qualifiers, which are encoded in first-order logical formulae. The advantage of using logical formulae is the possibility of verifying their satisfiability through the use of \emph{theorem prover} such as \emph{Z3} \cite{z3}; and they also make it easy to implement algorithms in a structured manner that allow for the more procedural use of typing rules; and they also allow for the verification of certain properties between types, such as subtype relations, in a faster and more intuitive manner.

The usefulness of this new notion of type was presented in the context of the \emph{Input Test Generator} functions used in \emph{Property-Based Testing} \cite{PBT1, PBT2}. In this way, statically, we can determine whether an input value generator for a programme, in addition to being \emph{correct}, is also \emph{complete}, and thus allows for the random generation of \emph{all} possible inputs for the programme.
The advantage of having a generator, which is complete, and which randomly generates an input on which to test the programme, not only allows the testing phase to be automated, but also increases the probability of finding errors and borderline cases.

The language $\lambda^{\textbf{TG}}$, however, did not provide for the use, and especially the management, of resources. The concept of \emph{History Expression} \cite{history2, history3, history} was therefore first presented: we have seen how these expressions represent a sequence of events involving resources, and being able to define ordering ($<_{\eta}$) and belonging ($\in$) relations on these, we can go on to define and verify, again by means of first-order logical formulae, the policies that we want to be respected, generally for each of the History Expression $\eta$ that will belong to the denotation of a History Expression $H$.

Thus, firstly, \emph{History Expressions} were formally introduced, with their grammar, semantics, the operations that can be performed on these, and properties of equality and $\alpha$-conversion.

Secondly, we extended the language $\lambda^{\textbf{TG}}$ with resource types, resource identifiers, and actions, i.e., special operators that will be tracked within the History Expressions, as they act on resources. Consequently, new terms have also been added to the language for creating resources, obtaining \emph{API} and invoking them, and for the use of actions.
The most important extension, certainly, concerns the changes made to the types of the language, with the introduction of \textbf{Extended Coverage Types}, which allow \emph{\textbf{all and only}} the values returned by a term to be tracked, enabling History Expressions to be correctly over-approximated; \textbf{History} $\pi$ types, and the addition of latent effects in \textbf{Function Types}.
Thus, we have created the pair $(\tau, \;H)$ in which $\tau$ is a \emph{Extended Coverage Type}, while $H$ is a History Expression: it is worth emphasising that while $\tau$ (by extracting the \emph{Coverage Type} it defines) is a \emph{Under-approximation Type}, $H$ is a \emph{Over-approximation Type}. This reflects the objective we want to focus on, i.e. we want to try to statically establish the \emph{completeness} of an \emph{Input Test Generator}, but we also want to check the \emph{correctness} of resource usage. This justified the use of \emph{over-approximation} logic in the definition of \emph{History Expressions}, as opposed to \emph{Coverage Types}.

After the extension of $\lambda^{\textbf{TG}}$, we presented the appropriate extensions to the \emph{type system} associated with $\lambda^{\textbf{TG}}$, which was shown not to be completely \emph{syntax-driven}, but uses algorithmic procedures such as \nextd, to obtain the next free identifier to be used; and \bind, since, as History Expression have been defined, they are independent of the type logic (over- or under-approximated) associated with the variables within them, and therefore the use of this procedure is necessary to bind the precise qualifier associated with a value or variable on the fly.

Subsequently, we also presented the extension concerning the algorithms for the \emph{\textcolor{Green}{type synthesis}}, for the \emph{\textcolor{purple}{type check}}, and for the auxiliary functions. In the latter case, it was seen that thanks to the logical separation between the return type (\emph{Coverage Type}) and the behaviour type (\emph{History Expression}), the changes made to the algorithms are minimal. Three algorithms, \subhist, \subconc\ and \histconj, were also introduced in order to procedurally verify the subtype relationship between two History Expression $H_1$ and $H_2$. In order to do this, it was necessary for the \emph{history} to be defined in a more structured and common way to facilitate the implementation of the algorithms, hence the \emph{Normal Form} for History Expressions was introduced, which makes use of the equality rules presented in Figure \ref{fig:hist_eq} and the $\alpha$-conversion rules presented in \ref{subsec:alpha}.

The main advantage gained from the expressive power of the History Expressions, especially in recursive functions, is the availability of qualifiers within the \emph{history}: this has made it possible to infer types that represent a very refined \emph{over-approximation}, which is close to the real behaviour of the programme.

\section{Conclusion Remarks}\label{sec:conclusion}
Possible future developments concern the implementation and use of our static analysis technique:
\begin{enumerate}
    \item The practical application of the type system through the implementation of typing algorithms in an experimental programming language.
    \item The application of our techniques in a actual Property-Based Testing system.
    \item Also provide for an under-approximated interpretation of History Expressions: while not intended for verifying resource correctness against policies, it can help identify bugs in resource usage by analysing behaviours that can be statically guaranteed within the program.
\end{enumerate}


\bibliographystyle{ACM-Reference-Format}
\bibliography{bibliography}


\begin{thebibliography}{30}


\ifx \showCODEN    \undefined \def \showCODEN     #1{\unskip}     \fi
\ifx \showDOI      \undefined \def \showDOI       #1{#1}\fi
\ifx \showISBNx    \undefined \def \showISBNx     #1{\unskip}     \fi
\ifx \showISBNxiii \undefined \def \showISBNxiii  #1{\unskip}     \fi
\ifx \showISSN     \undefined \def \showISSN      #1{\unskip}     \fi
\ifx \showLCCN     \undefined \def \showLCCN      #1{\unskip}     \fi
\ifx \shownote     \undefined \def \shownote      #1{#1}          \fi
\ifx \showarticletitle \undefined \def \showarticletitle #1{#1}   \fi
\ifx \showURL      \undefined \def \showURL       {\relax}        \fi
\providecommand\bibfield[2]{#2}
\providecommand\bibinfo[2]{#2}
\providecommand\natexlab[1]{#1}
\providecommand\showeprint[2][]{arXiv:#2}

\bibitem[Albert et~al\mbox{.}(2016)]%
        {static1-costa}
\bibfield{author}{\bibinfo{person}{Elvira Albert}, \bibinfo{person}{Richard Bubel}, \bibinfo{person}{Samir Genaim}, \bibinfo{person}{Reiner H{\"a}hnle}, \bibinfo{person}{Germ{\'a}n Puebla}, {and} \bibinfo{person}{Guillermo Rom{\'a}n-D{\'\i}ez}.} \bibinfo{year}{2016}\natexlab{}.
\newblock \showarticletitle{A formal verification framework for static analysis}.
\newblock \bibinfo{journal}{\emph{Software \& Systems Modeling}} \bibinfo{volume}{15}, \bibinfo{number}{4} (\bibinfo{year}{2016}), \bibinfo{pages}{987--1012}.
\newblock
\showISBNx{1619-1374}
\urldef\tempurl%
\url{https://doi.org/10.1007/s10270-015-0476-y}
\showDOI{\tempurl}


\bibitem[Baeten and Weijland(1991)]%
        {bpa}
\bibfield{author}{\bibinfo{person}{J.~C.~M. Baeten} {and} \bibinfo{person}{W.~P. Weijland}.} \bibinfo{year}{1991}\natexlab{}.
\newblock \bibinfo{booktitle}{\emph{Process algebra}}.
\newblock \bibinfo{publisher}{Cambridge University Press}, \bibinfo{address}{USA}.
\newblock
\showISBNx{0521400430}


\bibitem[Bartoletti et~al\mbox{.}(2009a)]%
        {history3}
\bibfield{author}{\bibinfo{person}{Massimo Bartoletti}, \bibinfo{person}{Pierpaolo Degano}, {and} \bibinfo{person}{Gian{-}Luigi Ferrari}.} \bibinfo{year}{2009}\natexlab{a}.
\newblock \showarticletitle{Planning and verifying service composition}.
\newblock \bibinfo{journal}{\emph{J. Comput. Secur.}} \bibinfo{volume}{17}, \bibinfo{number}{5} (\bibinfo{year}{2009}), \bibinfo{pages}{799--837}.
\newblock
\urldef\tempurl%
\url{https://doi.org/10.3233/JCS-2009-0357}
\showDOI{\tempurl}


\bibitem[Bartoletti et~al\mbox{.}(2008)]%
        {history2}
\bibfield{author}{\bibinfo{person}{Massimo Bartoletti}, \bibinfo{person}{Pierpaolo Degano}, \bibinfo{person}{Gian{-}Luigi Ferrari}, {and} \bibinfo{person}{Roberto Zunino}.} \bibinfo{year}{2008}\natexlab{}.
\newblock \showarticletitle{Semantics-Based Design for Secure Web Services}.
\newblock \bibinfo{journal}{\emph{{IEEE} Trans. Software Eng.}} \bibinfo{volume}{34}, \bibinfo{number}{1} (\bibinfo{year}{2008}), \bibinfo{pages}{33--49}.
\newblock
\urldef\tempurl%
\url{https://doi.org/10.1109/TSE.2007.70740}
\showDOI{\tempurl}


\bibitem[Bartoletti et~al\mbox{.}(2007)]%
        {static3-types-effects}
\bibfield{author}{\bibinfo{person}{Massimo Bartoletti}, \bibinfo{person}{Pierpaolo Degano}, \bibinfo{person}{Gian~Luigi Ferrari}, {and} \bibinfo{person}{Roberto Zunino}.} \bibinfo{year}{2007}\natexlab{}.
\newblock \showarticletitle{Types and Effects for Resource Usage Analysis}. In \bibinfo{booktitle}{\emph{Foundations of Software Science and Computational Structures}}, \bibfield{editor}{\bibinfo{person}{Helmut Seidl}} (Ed.). \bibinfo{publisher}{Springer Berlin Heidelberg}, \bibinfo{address}{Berlin, Heidelberg}, \bibinfo{pages}{32--47}.
\newblock
\showISBNx{978-3-540-71389-0}


\bibitem[Bartoletti et~al\mbox{.}(2009b)]%
        {history}
\bibfield{author}{\bibinfo{person}{Massimo Bartoletti}, \bibinfo{person}{Pierpaolo Degano}, \bibinfo{person}{Gian-Luigi Ferrari}, {and} \bibinfo{person}{Roberto Zunino}.} \bibinfo{year}{2009}\natexlab{b}.
\newblock \showarticletitle{Local policies for resource usage analysis}.
\newblock \bibinfo{journal}{\emph{ACM Trans. Program. Lang. Syst.}} \bibinfo{volume}{31}, \bibinfo{number}{6}, Article \bibinfo{articleno}{23} (\bibinfo{date}{aug} \bibinfo{year}{2009}), \bibinfo{numpages}{43}~pages.
\newblock
\showISSN{0164-0925}
\urldef\tempurl%
\url{https://doi.org/10.1145/1552309.1552313}
\showDOI{\tempurl}


\bibitem[Carlini and Wagner(2014)]%
        {rop}
\bibfield{author}{\bibinfo{person}{Nicholas Carlini} {and} \bibinfo{person}{David Wagner}.} \bibinfo{year}{2014}\natexlab{}.
\newblock \showarticletitle{{ROP} is Still Dangerous: Breaking Modern Defenses}. In \bibinfo{booktitle}{\emph{23rd USENIX Security Symposium (USENIX Security 14)}}. \bibinfo{publisher}{USENIX Association}, \bibinfo{address}{San Diego, CA}, \bibinfo{pages}{385--399}.
\newblock
\showISBNx{978-1-931971-15-7}
\urldef\tempurl%
\url{https://www.usenix.org/conference/usenixsecurity14/technical-sessions/presentation/carlini}
\showURL{%
\tempurl}


\bibitem[Colcombet and Fradet(2000)]%
        {dynamic3-trace-prop}
\bibfield{author}{\bibinfo{person}{Thomas Colcombet} {and} \bibinfo{person}{Pascal Fradet}.} \bibinfo{year}{2000}\natexlab{}.
\newblock \showarticletitle{Enforcing trace properties by program transformation}. In \bibinfo{booktitle}{\emph{Proceedings of the 27th ACM SIGPLAN-SIGACT Symposium on Principles of Programming Languages}} (Boston, MA, USA) \emph{(\bibinfo{series}{POPL '00})}. \bibinfo{publisher}{Association for Computing Machinery}, \bibinfo{address}{New York, NY, USA}, \bibinfo{pages}{54–66}.
\newblock
\showISBNx{1581131259}
\urldef\tempurl%
\url{https://doi.org/10.1145/325694.325703}
\showDOI{\tempurl}


\bibitem[de~Moura and Bj{\o}rner(2008)]%
        {z3}
\bibfield{author}{\bibinfo{person}{Leonardo de Moura} {and} \bibinfo{person}{Nikolaj Bj{\o}rner}.} \bibinfo{year}{2008}\natexlab{}.
\newblock \showarticletitle{Z3: An Efficient SMT Solver}. In \bibinfo{booktitle}{\emph{Tools and Algorithms for the Construction and Analysis of Systems}}, \bibfield{editor}{\bibinfo{person}{C.~R. Ramakrishnan} {and} \bibinfo{person}{Jakob Rehof}} (Eds.). \bibinfo{publisher}{Springer Berlin Heidelberg}, \bibinfo{address}{Berlin, Heidelberg}, \bibinfo{pages}{337--340}.
\newblock
\showISBNx{978-3-540-78800-3}


\bibitem[DeLine and F{\"a}hndrich(2001)]%
        {enforcing}
\bibfield{author}{\bibinfo{person}{Robert DeLine} {and} \bibinfo{person}{Manuel F{\"a}hndrich}.} \bibinfo{year}{2001}\natexlab{}.
\newblock \showarticletitle{Enforcing high-level protocols in low-level software}. In \bibinfo{booktitle}{\emph{Proceedings of the ACM SIGPLAN 2001 conference on Programming language design and implementation}}. \bibinfo{pages}{59--69}.
\newblock


\bibitem[Dunfield and Krishnaswami(2021)]%
        {bidirectional-typing}
\bibfield{author}{\bibinfo{person}{Jana Dunfield} {and} \bibinfo{person}{Neel Krishnaswami}.} \bibinfo{year}{2021}\natexlab{}.
\newblock \showarticletitle{Bidirectional Typing}.
\newblock \bibinfo{journal}{\emph{ACM Comput. Surv.}} \bibinfo{volume}{54}, \bibinfo{number}{5}, Article \bibinfo{articleno}{98} (\bibinfo{date}{May} \bibinfo{year}{2021}), \bibinfo{numpages}{38}~pages.
\newblock
\showISSN{0360-0300}
\urldef\tempurl%
\url{https://doi.org/10.1145/3450952}
\showDOI{\tempurl}


\bibitem[Eichelberger and Schmid(2014)]%
        {dynamic2-java}
\bibfield{author}{\bibinfo{person}{Holger Eichelberger} {and} \bibinfo{person}{Klaus Schmid}.} \bibinfo{year}{2014}\natexlab{}.
\newblock \showarticletitle{Flexible resource monitoring of Java programs}.
\newblock \bibinfo{journal}{\emph{Journal of Systems and Software}}  \bibinfo{volume}{93} (\bibinfo{year}{2014}), \bibinfo{pages}{163--186}.
\newblock
\showISSN{0164-1212}
\urldef\tempurl%
\url{https://doi.org/10.1016/j.jss.2014.02.022}
\showDOI{\tempurl}


\bibitem[Fink and Bishop(1997)]%
        {PBT1}
\bibfield{author}{\bibinfo{person}{George Fink} {and} \bibinfo{person}{Matt Bishop}.} \bibinfo{year}{1997}\natexlab{}.
\newblock \showarticletitle{Property-based testing: a new approach to testing for assurance}.
\newblock \bibinfo{journal}{\emph{SIGSOFT Softw. Eng. Notes}} \bibinfo{volume}{22}, \bibinfo{number}{4} (\bibinfo{date}{July} \bibinfo{year}{1997}), \bibinfo{pages}{74–80}.
\newblock
\showISSN{0163-5948}
\urldef\tempurl%
\url{https://doi.org/10.1145/263244.263267}
\showDOI{\tempurl}


\bibitem[Foster et~al\mbox{.}(2002)]%
        {flow-sensitive}
\bibfield{author}{\bibinfo{person}{Jeffrey~S. Foster}, \bibinfo{person}{Tachio Terauchi}, {and} \bibinfo{person}{Alex Aiken}.} \bibinfo{year}{2002}\natexlab{}.
\newblock \showarticletitle{Flow-sensitive type qualifiers}.
\newblock \bibinfo{journal}{\emph{SIGPLAN Not.}} \bibinfo{volume}{37}, \bibinfo{number}{5} (\bibinfo{date}{May} \bibinfo{year}{2002}), \bibinfo{pages}{1–12}.
\newblock
\showISSN{0362-1340}
\urldef\tempurl%
\url{https://doi.org/10.1145/543552.512531}
\showDOI{\tempurl}


\bibitem[Freeman and Pfenning(1991)]%
        {refinement}
\bibfield{author}{\bibinfo{person}{Tim Freeman} {and} \bibinfo{person}{Frank Pfenning}.} \bibinfo{year}{1991}\natexlab{}.
\newblock \showarticletitle{Refinement types for ML}. In \bibinfo{booktitle}{\emph{Proceedings of the ACM SIGPLAN 1991 conference on Programming language design and implementation}}. \bibinfo{pages}{268--277}.
\newblock


\bibitem[Igarashi and Kobayashi(2002)]%
        {resource-usage}
\bibfield{author}{\bibinfo{person}{Atsushi Igarashi} {and} \bibinfo{person}{Naoki Kobayashi}.} \bibinfo{year}{2002}\natexlab{}.
\newblock \showarticletitle{Resource usage analysis}.
\newblock \bibinfo{journal}{\emph{SIGPLAN Not.}} \bibinfo{volume}{37}, \bibinfo{number}{1} (\bibinfo{date}{Jan.} \bibinfo{year}{2002}), \bibinfo{pages}{331–342}.
\newblock
\showISSN{0362-1340}
\urldef\tempurl%
\url{https://doi.org/10.1145/565816.503303}
\showDOI{\tempurl}


\bibitem[Kang et~al\mbox{.}(2005)]%
        {path-sensitive}
\bibfield{author}{\bibinfo{person}{Hyun-Goo Kang}, \bibinfo{person}{Youil Kim}, \bibinfo{person}{Taisook Han}, {and} \bibinfo{person}{Hwansoo Han}.} \bibinfo{year}{2005}\natexlab{}.
\newblock \showarticletitle{A Path Sensitive Type System for Resource Usage Verification of C Like Languages}. In \bibinfo{booktitle}{\emph{Programming Languages and Systems}}, \bibfield{editor}{\bibinfo{person}{Kwangkeun Yi}} (Ed.). \bibinfo{publisher}{Springer Berlin Heidelberg}, \bibinfo{address}{Berlin, Heidelberg}, \bibinfo{pages}{264--280}.
\newblock
\showISBNx{978-3-540-32247-4}


\bibitem[Kobayashi(2003)]%
        {time-regions}
\bibfield{author}{\bibinfo{person}{Naoki Kobayashi}.} \bibinfo{year}{2003}\natexlab{}.
\newblock \showarticletitle{Time regions and effects for resource usage analysis}.
\newblock \bibinfo{journal}{\emph{SIGPLAN Not.}} \bibinfo{volume}{38}, \bibinfo{number}{3} (\bibinfo{date}{Jan.} \bibinfo{year}{2003}), \bibinfo{pages}{50–61}.
\newblock
\showISSN{0362-1340}
\urldef\tempurl%
\url{https://doi.org/10.1145/640136.604182}
\showDOI{\tempurl}


\bibitem[Kutsia and Schreiner(2014)]%
        {dynamic1-logicguard}
\bibfield{author}{\bibinfo{person}{Temur Kutsia} {and} \bibinfo{person}{Wolfgang Schreiner}.} \bibinfo{year}{2014}\natexlab{}.
\newblock \showarticletitle{Verifying the soundness of resource analysis for LogicGuard monitors (revised version)}.
\newblock \bibinfo{journal}{\emph{RISC Report Series, TR-14-08, JKU, Austria}} (\bibinfo{year}{2014}).
\newblock


\bibitem[Le et~al\mbox{.}(2022)]%
        {IL2}
\bibfield{author}{\bibinfo{person}{Quang~Loc Le}, \bibinfo{person}{Azalea Raad}, \bibinfo{person}{Jules Villard}, \bibinfo{person}{Josh Berdine}, \bibinfo{person}{Derek Dreyer}, {and} \bibinfo{person}{Peter~W. O'Hearn}.} \bibinfo{year}{2022}\natexlab{}.
\newblock \showarticletitle{Finding real bugs in big programs with incorrectness logic}.
\newblock \bibinfo{journal}{\emph{Proc. ACM Program. Lang.}} \bibinfo{volume}{6}, \bibinfo{number}{OOPSLA1}, Article \bibinfo{articleno}{81} (\bibinfo{date}{April} \bibinfo{year}{2022}), \bibinfo{numpages}{27}~pages.
\newblock
\urldef\tempurl%
\url{https://doi.org/10.1145/3527325}
\showDOI{\tempurl}


\bibitem[Lhee and Chapin(2003)]%
        {buffer}
\bibfield{author}{\bibinfo{person}{Kyung-Suk Lhee} {and} \bibinfo{person}{Steve~J. Chapin}.} \bibinfo{year}{2003}\natexlab{}.
\newblock \showarticletitle{Buffer overflow and format string overflow vulnerabilities}.
\newblock \bibinfo{journal}{\emph{Software: Practice and Experience}} \bibinfo{volume}{33}, \bibinfo{number}{5} (\bibinfo{year}{2003}), \bibinfo{pages}{423--460}.
\newblock
\urldef\tempurl%
\url{https://doi.org/10.1002/spe.515}
\showDOI{\tempurl}
\showeprint{https://onlinelibrary.wiley.com/doi/pdf/10.1002/spe.515}


\bibitem[Lopez-Garcia et~al\mbox{.}(2010)]%
        {static2}
\bibfield{author}{\bibinfo{person}{Pedro Lopez-Garcia}, \bibinfo{person}{Luthfi Darmawan}, {and} \bibinfo{person}{Francisco Bueno}.} \bibinfo{year}{2010}\natexlab{}.
\newblock \showarticletitle{A Framework for Verification and Debugging of Resource Usage Properties: Resource Usage Verification}. In \bibinfo{booktitle}{\emph{Technical Communications of the 26th International Conference on Logic Programming}}. Schloss Dagstuhl-Leibniz-Zentrum fuer Informatik.
\newblock


\bibitem[Marino and Millstein(2009)]%
        {types-effects-general}
\bibfield{author}{\bibinfo{person}{Daniel Marino} {and} \bibinfo{person}{Todd Millstein}.} \bibinfo{year}{2009}\natexlab{}.
\newblock \showarticletitle{A generic type-and-effect system}. In \bibinfo{booktitle}{\emph{Proceedings of the 4th International Workshop on Types in Language Design and Implementation}} (Savannah, GA, USA) \emph{(\bibinfo{series}{TLDI '09})}. \bibinfo{publisher}{Association for Computing Machinery}, \bibinfo{address}{New York, NY, USA}, \bibinfo{pages}{39–50}.
\newblock
\showISBNx{9781605584201}
\urldef\tempurl%
\url{https://doi.org/10.1145/1481861.1481868}
\showDOI{\tempurl}


\bibitem[Marriott et~al\mbox{.}(2003)]%
        {dynamic4-marriot}
\bibfield{author}{\bibinfo{person}{Kim Marriott}, \bibinfo{person}{Peter~J. Stuckey}, {and} \bibinfo{person}{Martin Sulzmann}.} \bibinfo{year}{2003}\natexlab{}.
\newblock \showarticletitle{Resource Usage Verification}. In \bibinfo{booktitle}{\emph{Programming Languages and Systems}}, \bibfield{editor}{\bibinfo{person}{Atsushi Ohori}} (Ed.). \bibinfo{publisher}{Springer Berlin Heidelberg}, \bibinfo{address}{Berlin, Heidelberg}, \bibinfo{pages}{212--229}.
\newblock
\showISBNx{978-3-540-40018-9}


\bibitem[Mirkovic et~al\mbox{.}(2004)]%
        {dos}
\bibfield{author}{\bibinfo{person}{Jelena Mirkovic}, \bibinfo{person}{Sven Dietrich}, \bibinfo{person}{David Dittrich}, {and} \bibinfo{person}{Peter Reiher}.} \bibinfo{year}{2004}\natexlab{}.
\newblock \bibinfo{booktitle}{\emph{Internet Denial of Service: Attack and Defense Mechanisms (Radia Perlman Computer Networking and Security)}}.
\newblock \bibinfo{publisher}{Prentice Hall PTR}, \bibinfo{address}{USA}.
\newblock
\showISBNx{0131475738}


\bibitem[O'Hearn(2019)]%
        {IL1}
\bibfield{author}{\bibinfo{person}{Peter~W. O'Hearn}.} \bibinfo{year}{2019}\natexlab{}.
\newblock \showarticletitle{Incorrectness logic}.
\newblock \bibinfo{journal}{\emph{Proc. ACM Program. Lang.}} \bibinfo{volume}{4}, \bibinfo{number}{POPL}, Article \bibinfo{articleno}{10} (\bibinfo{date}{Dec.} \bibinfo{year}{2019}), \bibinfo{numpages}{32}~pages.
\newblock
\urldef\tempurl%
\url{https://doi.org/10.1145/3371078}
\showDOI{\tempurl}


\bibitem[Paraskevopoulou et~al\mbox{.}(2015)]%
        {PBT2}
\bibfield{author}{\bibinfo{person}{Zoe Paraskevopoulou}, \bibinfo{person}{C{\u{a}}t{\u{a}}lin Hri{\c{T}}cu}, \bibinfo{person}{Maxime D{\'e}n{\`e}s}, \bibinfo{person}{Leonidas Lampropoulos}, {and} \bibinfo{person}{Benjamin~C. Pierce}.} \bibinfo{year}{2015}\natexlab{}.
\newblock \showarticletitle{Foundational Property-Based Testing}. In \bibinfo{booktitle}{\emph{Interactive Theorem Proving}}, \bibfield{editor}{\bibinfo{person}{Christian Urban} {and} \bibinfo{person}{Xingyuan Zhang}} (Eds.). \bibinfo{publisher}{Springer International Publishing}, \bibinfo{address}{Cham}, \bibinfo{pages}{325--343}.
\newblock
\showISBNx{978-3-319-22102-1}


\bibitem[Strom and Yemini(1986)]%
        {typestate}
\bibfield{author}{\bibinfo{person}{Robert~E. Strom} {and} \bibinfo{person}{Shaula Yemini}.} \bibinfo{year}{1986}\natexlab{}.
\newblock \showarticletitle{Typestate: A programming language concept for enhancing software reliability}.
\newblock \bibinfo{journal}{\emph{IEEE Transactions on Software Engineering}} \bibinfo{volume}{SE-12}, \bibinfo{number}{1} (\bibinfo{year}{1986}), \bibinfo{pages}{157--171}.
\newblock
\urldef\tempurl%
\url{https://doi.org/10.1109/TSE.1986.6312929}
\showDOI{\tempurl}


\bibitem[Xu et~al\mbox{.}(2018)]%
        {taint}
\bibfield{author}{\bibinfo{person}{Zhiwu Xu}, \bibinfo{person}{Cheng Wen}, {and} \bibinfo{person}{Shengchao Qin}.} \bibinfo{year}{2018}\natexlab{}.
\newblock \showarticletitle{State-taint analysis for detecting resource bugs}.
\newblock \bibinfo{journal}{\emph{Science of Computer Programming}}  \bibinfo{volume}{162} (\bibinfo{year}{2018}), \bibinfo{pages}{93--109}.
\newblock
\showISSN{0167-6423}
\urldef\tempurl%
\url{https://doi.org/10.1016/j.scico.2017.06.010}
\showDOI{\tempurl}
\newblock
\shownote{Special Issue on TASE 2016}.


\bibitem[Zhou et~al\mbox{.}(2023)]%
        {coverage}
\bibfield{author}{\bibinfo{person}{Zhe Zhou}, \bibinfo{person}{Ashish Mishra}, \bibinfo{person}{Benjamin Delaware}, {and} \bibinfo{person}{Suresh Jagannathan}.} \bibinfo{year}{2023}\natexlab{}.
\newblock \showarticletitle{Covering All the Bases: Type-Based Verification of Test Input Generators}.
\newblock \bibinfo{journal}{\emph{Proc. ACM Program. Lang.}} \bibinfo{volume}{7}, \bibinfo{number}{PLDI}, Article \bibinfo{articleno}{157} (\bibinfo{date}{jun} \bibinfo{year}{2023}), \bibinfo{numpages}{24}~pages.
\newblock
\urldef\tempurl%
\url{https://doi.org/10.1145/3591271}
\showDOI{\tempurl}


\end{thebibliography}

\newpage

\appendix

\section{Typing Rules}\label{app:type-system}

The following rules are in addition to those presented in Section \ref{sec:type_system}.

\begin{figure}[ht]
    \centering
    \noindent
    \begin{minipage}{0.3\textwidth}
        \begin{equation}
            \frac{
                \Gamma \vdash^{\textbf{WF}} (\text{Ty}(op), \; \epsilon)
            }{
                \Gamma \vdash op : (\text{Ty}(op), \; \epsilon)
            }
            \tag{\textsc{TOp}}
        \end{equation}
    \end{minipage}
    \hfill
    \begin{minipage}{0.6\textwidth}
        \begin{equation}
            \frac{
                \begin{gathered}
                    \Gamma \vdash e_x : (\tau_x, \; H_{e_x}) \quad \Gamma, x{:}\tau_x \vdash e : (\tau, \; H_e) \\
                    \Gamma \vdash^{\textbf{WF}} (\tau, \; H_{e_x} \cdot H_e)
                \end{gathered}  
            }{
                \Gamma \vdash \texttt{\small{let}}\;x = e_x\;\texttt{\small{in}}\;e : (\tau, \; H_{e_x} \cdot H_e)
            }
            \tag{\textsc{TLetE}}
        \end{equation}
    \end{minipage}

    \vspace{20pt}

    \begin{equation}
        \frac{
            \begin{gathered}
                \Gamma \vdash op : (\overline{a_i{:}\overa{v: b_i \;|\; \phi_i}} \rightarrow \tau_x, \; H_{op} ) \quad \forall i, \Gamma \vdash u_i : (\under{v: b_i \;|\; \phi_i}, \; H_{u_i}) \\
                \Gamma, x{:}\tau_x\overline{[a_i \mapsto u_i]} \vdash e : (\tau, \; H_e) \quad \Gamma \vdash^{\textbf{WF}} (\tau, \; H_{op} \cdot ( \underset{i}{\bullet} \; H_{u_i} ) \cdot H_e)
            \end{gathered}
        }{
            \Gamma \vdash \texttt{\small{let}}\;x = op\;\overline{u_i}\;\texttt{\small{in}}\;e : (\tau, \; H_{op} \cdot ( \underset{i}{\bullet} \; H_{u_i} ) \cdot H_e)
        }
        \tag{\textsc{TAppOp}}
    \end{equation}

    \vspace{20pt}

    \begin{equation}
        \frac{
            \begin{gathered}
                \neg(\kappa_x = \pi) \quad \Gamma \vdash v_1 : ((\tau_1 \rightarrow \kappa_1) \rightarrow \kappa_x, \; H_{v_1}) \quad \Gamma \vdash v_2: (\tau_1 \rightarrow \kappa_1, \; H_{v_2}) \\
                \Gamma, x : \kappa_x \vdash e: (\tau, \; H_e) \quad \Gamma \vdash^{\textbf{WF}} (\tau, \; H_{v_1} \cdot H_{v_2} \cdot H_e)
            \end{gathered}
        }{
            \Gamma \vdash \texttt{\small{let}}\;x = v_1\;v_2\;\texttt{\small{in}}\;e : (\tau, \; H_{v_1} \cdot H_{v_2} \cdot H_e)
        }
        \tag{\textsc{TAppFunMulti}}
    \end{equation}

    \vspace{20pt}

    \begin{equation}
        \frac{
            \begin{gathered}
                \Gamma \vdash v_1 : ((\tau_1 \rightarrow \kappa_1) \rightarrow (\tau_{v_1}, \; H_{\tau_{v_1}}), \; H_{v_1}) \quad \Gamma \vdash v_2: (\tau_1 \rightarrow \kappa_1, \; H_{v_2}) \\
                \Gamma, x : \kappa_x \vdash e: (\tau, \; H_e) \quad \Gamma \vdash^{\textbf{WF}} (\tau, \; H_{v_1} \cdot H_{v_2} \cdot H_{\tau_{v_1}} \cdot H_e)
            \end{gathered}
        }{
            \Gamma \vdash \texttt{\small{let}}\;x = v_1\;v_2\;\texttt{\small{in}}\;e : (\tau, \; H_{v_1} \cdot H_{v_2} \cdot H_{\tau_{v_1}} \cdot H_e)
        }
        \tag{\textsc{TAppFunLast}}
    \end{equation}
\end{figure}
\section{Typing Algorithms}\label{app:alg}

\subsection{Type Synthesis}

\begin{figure}[H]
    \begin{multicols}{2}
        \begin{equation}
            \frac{
                \Gamma \vdash^{\textbf{WF}} (\under{v: b \;|\; \bot \;|\; \bot}, \; err())
            }{
                \Gamma \vdash err \;\textcolor{Green}{\Rightarrow}\; (\under{v: b \;|\; \bot \;|\; \bot}, \; err())
            }
            \tag{\textsc{SynErr}}
            \label{eq:synerr}
        \end{equation}
        
        \columnbreak
    
        \begin{equation}
            \frac{
                \Gamma \vdash^{\textbf{WF}} (\text{Ty}(c), \; \epsilon)
            }{
                \Gamma \vdash c \;\textcolor{Green}{\Rightarrow}\; (\text{Ty}(c), \; \epsilon)
            }
            \tag{\textsc{SynConst}}
            \label{eq:synconst}
        \end{equation}
    \end{multicols}

    \begin{multicols}{2}
        \begin{equation}
            \frac{
                \Gamma \vdash^{\textbf{WF}} (\text{Ty}(op), \; \epsilon)
            }{
                \Gamma \vdash op \;\textcolor{Green}{\Rightarrow}\; (\text{Ty}(op), \; \epsilon)
            }
            \tag{\textsc{SynOp}}
            \label{eq:synop}
        \end{equation}

        \columnbreak

        \begin{equation}
            \frac{
                \Gamma \vdash^{\textbf{WF}} (\text{Ty}(\alpha), \; \epsilon)
            }{
                \Gamma \vdash \alpha \;\textcolor{Green}{\Rightarrow}\; (\text{Ty}(\alpha), \; \epsilon)
            }
            \tag{\textsc{SynAction}}
            \label{eq:synaction}
        \end{equation}
    \end{multicols}

    \begin{multicols}{2}
        \begin{equation}
            \frac{
                \Gamma \vdash^{\textbf{WF}} (\under{v: b \;|\; v = x \;|\; v = x}, \; \epsilon)
            }{
                \Gamma \vdash x \;\textcolor{Green}{\Rightarrow}\; (\under{v: b \;|\; v = x \;|\; v = x}, \; \epsilon)
            }
            \tag{\textsc{SynVarBase}}
            \label{eq:synvarb}
        \end{equation}

        \columnbreak

        \begin{equation}
            \frac{
                \Gamma(x) = (a: \tau_a \rightarrow \kappa) \quad \Gamma \vdash^{\textbf{WF}} (a:\tau_a \rightarrow \kappa, \; \epsilon)
            }{
                \Gamma \vdash x \;\textcolor{Green}{\Rightarrow}\; (a:\tau_a \rightarrow \kappa, \; \epsilon)
            }
            \tag{\textsc{SynVarFun}}
            \label{eq:synvarf}
        \end{equation}
    \end{multicols}

    \begin{equation}
        \frac{
            \begin{gathered}
                \neg(\kappa_x = \pi) \quad \Gamma \vdash v_1 \textcolor{Green}{\;\Rightarrow\;} (a{:}\overa{v: b \;|\; \phi} \rightarrow \kappa_x, \; H_{v_1}) \quad \Gamma \vdash v_2 \textcolor{Green}{\;\Rightarrow\;} (\under{v: b \;|\; \psi \;|\; \theta}, \; H_{v_2}) \\
                \Gamma' = a:\under{v: b \;|\; v = v_2 \;\land\; \phi \;|\; v = v_2 \;\land\; \phi}, \; x:\kappa_x \quad \Gamma, \Gamma' \vdash e \textcolor{Green}{\;\Rightarrow\;} (\tau, \; H_e) \\
                \tau' = \textsc{\textcolor{blue}{Ex}}(\Gamma', \;\tau) \quad H_e' = \bind(\Gamma'_{\psi}, \;H_e) \quad \Gamma \vdash^{\textbf{WF}} (\tau', \; H_{v_1} \cdot H_{v_2} \cdot H_e')
            \end{gathered}
        }{
            \Gamma \vdash \;\texttt{let}\; x = v_1 \; v_2 \;\texttt{in}\; e \;\textcolor{Green}{\Rightarrow}\; (\tau', \;H_{v_1} \cdot H_{v_2} \cdot H_e')
        }
        \tag{\textsc{SynAppBaseMulti}}
        \label{eq:synappbasemulti}
    \end{equation}

    \begin{equation}
        \frac{
            \begin{gathered}
                \Gamma \vdash v_1 \textcolor{Green}{\;\Rightarrow\;} (a{:}\overa{v: b \;|\; \phi} \rightarrow (\tau_{v_1}, \;H_{\tau_{v_1}}), \; H_{v_1}) \quad \Gamma \vdash v_2 \textcolor{Green}{\;\Rightarrow\;} (\under{v: b \;|\; \psi \;|\; \theta}, \; H_{v_2}) \\
                \Theta = \{[Y/X] \;|\; new(X) \in H_{\tau_{v_1}} \;\land\; \nextd(\Delta) = Y\} \\
                H_{\tau_{v_1}}^{\star} = \bind(a: \under{v: b \;|\; v = v_2 \;\land\; \phi}, \; H_{\tau_{v_1}}(\Theta)) \\
                \Gamma' = a: \under{v: b \;|\; v = v_2 \;\land\; \phi \;|\; v = v_2 \;\land\; \phi}, \; x:\tau_{v_1}(\Theta) \quad \Gamma, \Gamma' \vdash e \textcolor{Green}{\;\Rightarrow\;} (\tau, \;H_e) \\
                \tau' = \textsc{\textcolor{blue}{Ex}}(\Gamma', \;\tau) \quad H_e' = \bind(\Gamma'_{\psi}, \;H_e) \quad \Gamma \vdash^{\textbf{WF}} (\tau', \;H_{v_1} \cdot H_{v_2} \cdot H_{\tau_{v_1}}^{\star} \cdot H_e')
            \end{gathered}
        }{
            \Gamma \vdash \;\texttt{let}\; x = v_1 \; v_2 \;\texttt{in}\; e \textcolor{Green}{\;\Rightarrow\;} (\tau', \; H_{v_1} \cdot H_{v_2} \cdot H_{\tau_{v_1}}^{\star} \cdot H_e')
        }
        \tag{\textsc{SynAppBaseLast}}
        \label{eq:synappbaselast}
    \end{equation}

    \begin{equation}
        \frac{
            \begin{gathered}
                \Gamma \vdash e_x \textcolor{Green}{\;\Rightarrow\;} (\tau_x, \; H_{e_x}) \quad \Gamma' = x:\tau_x \quad \Gamma, \Gamma' \vdash e \textcolor{Green}{\;\Rightarrow\;} (\tau, \; H_e) \\
                \tau' = \textsc{\textcolor{blue}{Ex}}(\Gamma', \;\tau) \quad H_e' = \bind(\Gamma'_{\psi}, \;H_e) \quad \Gamma \vdash^{\textbf{WF}} (\tau', \; H_{e_x} \cdot H_e')
            \end{gathered}  
        }{
            \Gamma \vdash \texttt{\small{let}}\;x = e_x\;\texttt{\small{in}}\;e \textcolor{Green}{\;\Rightarrow\;} (\tau', \; H_{e_x} \cdot H_e')
        }
        \tag{\textsc{SynLetE}}
        \label{eq:synlete}
    \end{equation}

    \begin{equation}
        \frac{
            \begin{gathered}
                \Gamma \vdash op \textcolor{Green}{\;\Rightarrow\;} (\overline{a_i{:}\overa{v: b_i \;|\; \phi_i}} \rightarrow \tau_x, \; H_{op} ) \quad \forall i, \Gamma \vdash u_i \textcolor{Green}{\;\Rightarrow\;} (\under{v: b_i \;|\; \psi_i \;|\; \theta_i}, \;H_{u_i}) \\
                \Gamma' = \overline{a_i : \under{v: b_i \;|\; v = u_i \;\land\; \phi_i \;|\; u_i \;\land\; \phi_i}} \quad \Gamma, \Gamma' \vdash e : (\tau, \; H_e) \\
                \tau' = \textsc{\textcolor{blue}{Ex}}(\Gamma', \;\tau) \quad 
                H_e' = \bind(\Gamma'_{\psi}, \;H_e) \quad 
                \Gamma \vdash^{\textbf{WF}} (\tau', \; H_{op} \cdot ( \underset{i}{\bullet} \; H_{u_i} ) \cdot H_e')
            \end{gathered}
        }{
            \Gamma \vdash \texttt{\small{let}}\;x = op\;\overline{u_i}\;\texttt{\small{in}}\;e \textcolor{Green}{\;\Rightarrow\;} (\tau', \; H_{op} \cdot ( \underset{i}{\bullet} \; H_{u_i} ) \cdot H_e')
        }
        \tag{\textsc{SynAppOp}}
        \label{eq:synappop}
    \end{equation}
    \caption{Extension of type synthesis algorithms - Part I.}
    \label{fig:type-synthesis-1}
\end{figure}

\begin{figure}[H]
    \begin{equation}
        \frac{
            \begin{gathered}
                \Gamma \vdash u_a \textcolor{Green}{\;\Rightarrow\;} (\tau_{u_a}, \;H_{u_a}) \quad \forall i, \Gamma'_i = \overline{y: \tau_y} \quad \Gamma, \Gamma' \vdash d_i(\overline{y}) \textcolor{Green}{\;\Rightarrow\;} (\tau_{d_i}, \;H_{d_i}) \quad H_{d_i}' = \bind(\Gamma_{i'_{\psi}}, \;H_{d_i}) \\
                \forall i, \text{Ty}(d_i) = \overline{y{:}\overa{v: b_y \;|\; \theta_y}} \rightarrow \under{v: b \;|\; \psi_i \;|\; \psi_i} \\
                \Gamma''_{i} = \overline{y{:}\under{v: b_y \;|\; \theta_y \;|\; \theta_y}}, \;a{:}\under{v: b \;|\; v = u_a \;\land\; \psi_i \;|\; v = u_a \;\land\; \psi_i} \\
                \Gamma, \Gamma''_{i} \vdash e_i \textcolor{Green}{\;\Rightarrow\;} (\tau_i, \; H_i) \quad \tau'_i = \textcolor{blue}{\text{Ex}}(\Gamma''_{i}, \;\tau_i) \quad H_i' = \bind(\Gamma_{i''_{\psi}}, \;H_i) \\
                \Gamma \vdash^{\textbf{WF}} (\textcolor{blue}{\text{Disj}}(\overline{\tau'_i}), \;H_{u_a} \cdot (\underset{i}{\bullet} \;H_{d_i}) \cdot \underset{i}{\bigoplus} \;H_i')
            \end{gathered}
        }{
            \Gamma \vdash \texttt{match} \;u_a\; \texttt{with}  \; \overline{d_i\overline{y} \rightarrow e_i} \textcolor{Green}{\;\Rightarrow\;} (\textcolor{blue}{\text{Disj}}(\overline{\tau'_i}), \; H_{u_a} \cdot (\underset{i}{\bullet} \; H_{d_i}') \cdot \underset{i}{\bigoplus} \; H_i')
        }
        \tag{\textsc{SynMatch}}
        \label{eq:synmatch}
    \end{equation}

    \begin{equation}
        \frac{
            \begin{gathered}
                \neg(\kappa_x = \pi) \quad \Gamma \vdash v_1 \textcolor{Green}{\;\Rightarrow\;} ((\tau_1 \rightarrow \kappa_1) \rightarrow \kappa_x, \; H_{v_1}) \quad \Gamma \vdash v_2 \textcolor{purple}{\;\Leftarrow\;} (\tau_1 \rightarrow \kappa_1, \; H_{v_2}) \quad \Gamma' = x : \kappa_x \\
                \Gamma, \Gamma' \vdash e \textcolor{Green}{\;\Rightarrow\;} (\tau, \; H_e) \quad \tau' = \textsc{\textcolor{blue}{Ex}}(\Gamma', \;\tau) \quad H_e' = \bind(\Gamma'_{\psi}, \;H_e) \quad \Gamma \vdash^{\textbf{WF}} (\tau', \; H_{v_1} \cdot H_{v_2} \cdot H_e')
            \end{gathered}
        }{
            \Gamma \vdash \texttt{\small{let}}\;x = v_1\;v_2\;\texttt{\small{in}}\;e : (\tau', \; H_{v_1} \cdot H_{v_2} \cdot H_e')
        }
        \tag{\textsc{SynAppFunMulti}}
        \label{eq:synappfunmulti}
    \end{equation}
    
    \begin{equation}
        \frac{
            \begin{gathered}
                \Gamma \vdash v_1 \textcolor{Green}{\;\Rightarrow\;} ((\tau_1 \rightarrow \kappa_1) \rightarrow (\tau_{v_1}, \; H_{\tau_{v_1}}), \; H_{v_1}) \quad \Gamma \vdash v_2 \textcolor{purple}{\;\Leftarrow\;} (\tau_1 \rightarrow \kappa_1, \; H_{v_2}) \quad \Gamma' = x : \tau_{v_1} \\
                \Gamma, \Gamma' \vdash e \textcolor{Green}{\;\Rightarrow\;} (\tau, \; H_e) \quad \tau' = \textsc{\textcolor{blue}{Ex}}(\Gamma', \;\tau) \quad H_e' = \bind(\Gamma'_{\psi}, \;H_e) \quad \Gamma \vdash^{\textbf{WF}} (\tau', \; H_{v_1} \cdot H_{v_2} \cdot H_{\tau_{v_1}} \cdot H_e')
            \end{gathered}
        }{
            \Gamma \vdash \texttt{\small{let}}\;x = v_1\;v_2\;\texttt{\small{in}}\;e : (\tau', \; H_{v_1} \cdot H_{v_2} \cdot H_{\tau_{v_1}} \cdot H_e')
        }
        \tag{\textsc{SynAppFunLast}}
        \label{eq:synappfunlast}
    \end{equation}

    \begin{equation}
        \frac{
            \begin{gathered}
                \Gamma' = x: \under{v: r \;|\; v = X \;|\; v = X} \quad \Gamma, \Gamma' \vdash e \textcolor{Green}{\;\Rightarrow\;} (\tau, \; H_e) \quad \nextd(\Delta) = X \\
                \tau' = \textsc{\textcolor{blue}{Ex}}(\Gamma', \; \tau) \quad H_e' = \bind(\Gamma'_{\psi}, \;H_e) \quad \Gamma \vdash^{\textbf{WF}} (\tau', \; \texttt{\small{new}}_r(X) \cdot H_e')
            \end{gathered}
        }{
            \Gamma \vdash \texttt{\small{let}}\;x = \texttt{\small{new}}_r\;()\;\texttt{\small{in}}\;e : (\tau, \;\texttt{\small{new}}_r(X) \cdot H_e)
        }
        \tag{\textsc{SynNew}}
        \label{eq:synew}
    \end{equation}

    \begin{equation}
        \frac{
            \begin{gathered}
                \Gamma' = x : \under{v: api \;|\; v = F \;|\; v = F} \quad \Gamma, \Gamma' \vdash e \textcolor{Green}{\;\Rightarrow\;} (\tau, \; H_e) \quad \Delta(F) = \tau_{F} \\
                \tau' = \textsc{\textcolor{blue}{Ex}}(\Gamma', \;\tau) \quad H_e' = \bind(\Gamma'_{\psi}, \;H_e) \quad \Gamma \vdash^{\textbf{WF}} (\tau', \; \texttt{\small{get}}(F) \cdot H_e')
            \end{gathered}
        }{
            \Gamma \vdash \texttt{\small{let}}\;x = \texttt{\small{get}}\;F\;\texttt{\small{in}}\;e : (\tau, \; \texttt{\small{get}}(F) \cdot H_e)
        }
        \tag{\textsc{SynGet}}
        \label{eq:syngetapi}
    \end{equation}

    \begin{equation}
        \frac{
            \begin{gathered}
                \Gamma \vdash \alpha \textcolor{Green}{\;\Rightarrow\;} (\overline{a_i{:}\overa{v: b_i \;|\; \phi_i}} \rightarrow \tau_x, \;H_{\alpha}) \\
                \forall i, \; \Gamma \vdash u_i \textcolor{Green}{\;\Rightarrow\;} (\under{v: b_i \;|\; \psi_i \;|\; \theta_i}, \; H_{u_i}) \quad \Gamma' = \overline{a_i{:}\under{v: b_i \;|\; v = u_i \;\land\; \phi_i \;|\; v = u_i \;\land\; \phi_i}} \quad \Gamma'' = \Gamma', x:\tau_x \\
                \Gamma, \Gamma'' \vdash e \textcolor{Green}{\;\Rightarrow\;} (\tau, \; H_e) \quad \tau' = \textsc{\textcolor{blue}{Ex}}(\Gamma'', \;\tau) \quad H_e' = \bind(\Gamma''_{\psi}, \;H_e) \\
                H^{\star} = \bind(\Gamma'_{\psi}, \;\alpha(\overline{b_i{:}v = a_i})) \quad \Gamma \vdash^{\textbf{WF}} (\tau', \;H_\alpha \cdot ( \underset{i}{\bullet} \; H_{u_i} ) \cdot H^{\star} \cdot H_e')
            \end{gathered}
        }{
            \Gamma \vdash \texttt{\small{let}}\;x = \alpha \; \overline{u_i}\;\texttt{\small{in}}\;e \textcolor{Green}{\;\Rightarrow\;} (\tau', \; H_\alpha \cdot ( \underset{i}{\bullet} \; H_{u_i} ) \cdot H^{\star} \cdot H_e')
        }
        \tag{\textsc{SynLetAction}}
        \label{eq:synletaction}
    \end{equation}
    \caption{Extension of type synthesis algorithms - Part II.}
    \label{fig:type-synthesis-2}
\end{figure}

\begin{figure}[H]
    \begin{equation}
        \frac{
            \begin{gathered}
                \Gamma \vdash f \textcolor{Green}{\;\Rightarrow\;} \under{v: api \;|\; \phi \;|\; \psi} \quad \forall i, \; \Gamma \vdash y_i \textcolor{Green}{\;\Rightarrow\;} (\under{v: b_i \;|\; \phi_i \;|\; \psi_i}, \;H_{y_i}) \\
                \under{v: api \;|\; \phi' \;|\; \psi'} = \textsc{\textcolor{blue}{Ex}}(\Gamma, \;\under{v: api \;|\; \phi \;|\; \psi}) \quad Api = \{F \;|\; \phi'[v \mapsto F] \;\land\; \Delta(F)\downarrow\} \\
                \forall F \in Api, \;\Delta(F) = \overline{a_i{:}\overa{v: b_i \;|\; \theta_i}} \rightarrow (\tau_{F_x}, \; H_{F_x}) \;\land\; \Gamma \vdash^{\textbf{WF}} \overline{a_i{:}\overa{v: b_i \;|\; \theta_i}} \rightarrow (\tau_{F_x}, \; H_{F_x}) \\
                \textsc{\textcolor{blue}{Disj}}(\overline{\tau_{F_x}}) = \tau_x \quad \Gamma' = \overline{a_i: \under{v: b_i \;|\; v = y_i \;\land\; \theta_i \;|\; v = y_i \;\land\; \theta_i}}, \;x:\tau_x \\
                \Gamma, \Gamma' \vdash e \textcolor{Green}{\;\Rightarrow\;} (\tau, \; H_e) \quad \tau' = \textsc{\textcolor{blue}{Ex}}(\Gamma', \;\tau) \quad H_e' = \bind(\Gamma', \;H_e) \\
                \Gamma \vdash^{\textbf{WF}} (\tau', \; ( \underset{i}{\bullet} \; H_{y_i} ) \cdot call(\psi; \;\overline{a_i{:}(b_i{:}v = y_i \;\land\; \theta_i})) \cdot H_e')
            \end{gathered}
        }{
            \Gamma \vdash \texttt{\small{let}}\;x = f\;\overline{y_i}\; \texttt{\small{in}}\;e \textcolor{Green}{\;\Rightarrow\;} (\tau', \; ( \underset{i}{\bullet} \; H_{y_i} ) \cdot call(\psi; \;\overline{a_i{:}(b_i{:}v = y_i \;\land\; \theta_i})) \cdot H_e')
        }
        \tag{\textsc{SynAppAPI}}
        \label{eq:synappapi}
    \end{equation}
    \caption{Extension of type synthesis algorithms - Part III.}
    \label{fig:type-synthesis-3}
\end{figure}

\subsection{Type Check}

\begin{figure}[ht]
    \begin{multicols}{2}
        \begin{equation}
            \frac{
                \varnothing \vdash e \textcolor{Green}{\;\Rightarrow\;} \pi \quad \Gamma \vdash \pi <: \pi' \quad \Gamma \vdash^{\textbf{WF}} \pi'
            }{
                \Gamma \vdash e \textcolor{purple}{\;\Leftarrow\;} \pi'
            }
            \tag{\textsc{ChkSub}}
            \label{eq:chksub}
        \end{equation}

        \columnbreak
    
        \begin{equation}
            \frac{
                \Gamma, x: \tau_x \vdash e \textcolor{purple}{\;\Leftarrow\;} \pi \quad \Gamma \vdash^{\textbf{WF}} (x: \tau_x \rightarrow \pi, \; \epsilon)
            }{
                \Gamma \vdash \lambda x : \lfloor\tau_x\rfloor. e \textcolor{purple}{\;\Leftarrow\;} (x: \tau_x \rightarrow \pi, \; \epsilon)
            }
            \tag{\textsc{ChkFun}}
            \label{eq:chkfun}
        \end{equation}
    \end{multicols}
    
    \begin{equation}
        \frac{
            \Gamma, x: \tau_x \vdash e \textcolor{purple}{\;\Leftarrow\;} (y: \tau_y \rightarrow \kappa, \; \epsilon) \quad \Gamma \vdash^{\textbf{WF}} (x: \tau_x \rightarrow y: \tau_y \rightarrow \kappa, \; \epsilon)
        }{
            \Gamma \vdash \lambda x : \lfloor\tau_x\rfloor. e \textcolor{purple}{\;\Leftarrow\;} (x: \tau_x \rightarrow y: \tau_y \rightarrow \kappa, \; \epsilon)
        }
        \tag{\textsc{ChkFunFlat}}
        \label{eq:chkflat}
    \end{equation}

    \begin{equation}
        \frac{
            \begin{gathered}
                \nextd(\Delta, \;x{:}\overa{v: b \;|\; v \prec x \;\land\; \phi} \rightarrow \overline{\tau_i} \rightarrow (\tau, \;\epsilon)) = F \\
                \Gamma, x{:}\overa{v: b \;|\; \phi}, f{:}\under{v: api \;|\; v = F \;|\; v = F} \vdash e \textcolor{purple}{\;\Leftarrow\;} \overline{\tau_i} \rightarrow (\tau, \;H) \\
                A = \{x{:}(b{:}\phi)\} \cup \{a_i{:}(b_i{:}\psi_i) \;|\; \tau_i = a_i{:}\overa{v: b_i \;|\; \psi_i}\} \quad \Gamma \vdash^{\textbf{WF}} (x{:}\overa{v: b \;|\; \phi} \rightarrow \overline{\tau_i} \rightarrow (\tau, \;H), \;\mu F(A)(H))
            \end{gathered}
        }{
            \Gamma \vdash \text{fix}\;f{:}(b \rightarrow \lfloor\kappa\rfloor).\lambda x{:}b.\;e \textcolor{purple}{\;\Leftarrow\;} (x{:}\overa{v: b \;|\; \phi} \rightarrow \overline{\tau_i} \rightarrow (\tau, \;H), \;\mu F(A)(H))
        }
        \tag{\textsc{ChkFix}}
        \label{eq:chkfix}
    \end{equation}
    \caption{Extension of type check algorithms.}
    \label{fig:type-check}
\end{figure}

\newpage

\subsection{Auxiliary Functions}

\begin{algorithm}[ht]
    \caption{Exists}\label{alg:ex}
    \begin{algorithmic}[1]
        \Procedure {\textcolor{blue}{Ex}}{$x$, $\under{v: t \;|\; \phi_x \;|\; \psi_x}$, $\tau$}
            \Match {$\tau$}
                \Case {$\under{v: t \;|\; \phi \;|\; \psi}$}
                    \State \Return $\under{v: t \;|\; \exists x: t, \phi_x[v \mapsto x] \;\land\; \phi \;|\; \exists x: t, \psi_x[v \mapsto x] \;\land\; \psi}$
                \EndCase
                \Case {$(\tau_H, \;H)$}
                    \State \Return $(\Call{\textcolor{blue}{Ex}}{x, \under{v: t \;|\; \phi_x \;|\; \psi_x}, \;\tau_H}, \;\bind(x:\under{v: t \;|\; \psi_x}, \;H))$
                \EndCase
                \Case {$\overa{v: t \;|\; \phi}$}
                    \State \Return $\overa{v: t \;|\; \forall x: t, \phi_x[v \mapsto x] \Longrightarrow \phi}$
                \EndCase
                \Case {$a{:}\tau_a \rightarrow \kappa$}
                    \State $\tau_a' \gets \Call{\textcolor{blue}{Fa}}{x, \under{v: t \;|\; \phi_x \;|\; \psi_x}, \;\tau_a}$
                    \State \Return {$a{:}\tau'_a \rightarrow \Call{\textcolor{blue}{Ex}}{x, \under{v: t \;|\; \phi_x \;|\; \psi_x}, \;\kappa}$}
                \EndCase
            \EndMatch
        \EndProcedure
    \end{algorithmic}
\end{algorithm}

\begin{algorithm}[ht]
    \caption{ForAll}\label{alg:forall}
    \begin{algorithmic}[1]
          \Procedure {\textcolor{blue}{Fa}}{$x$, $\under{v: t \;|\; \phi_x \;|\; \psi_x}$, $\tau$}
            \Match {$\tau$}
                \Case {$\under{v: t \;|\; \phi \;|\; \psi}$}
                    \State \Return $\under{v: t \;|\; \forall x: t, \phi_x[v \mapsto x] \Longrightarrow \phi \;|\; \forall x: t, \psi_x[v \mapsto x] \Longrightarrow \psi}$
                \EndCase
                \Case {$(\tau_H, \;H)$}
                    \State \Return $(\Call{\textcolor{blue}{Fa}}{x, \under{v: t \;|\; \phi_x \;|\; \psi_x}, \;\tau_H}, \;\bind(x:\under{v: t \;|\; \psi_x}, \;H))$
                \EndCase
                \Case {$\overa{v: t \;|\; \phi}$}
                    \State \Return $\overa{v: t \;|\; \exists x: t, \phi_x[v \mapsto x] \;\land\; \phi}$
                \EndCase
                \Case {$a{:}\tau_a \rightarrow \kappa$}
                    \State $\tau_a' \gets \Call{\textcolor{blue}{Ex}}{x, \under{v: t \;|\; \phi_x \;|\; \psi_x}, \;\tau_a}$
                    \State \Return {$a{:}\tau'_a \rightarrow \Call{\textcolor{blue}{Fa}}{x, \under{v: t \;|\; \phi_x \;|\; \psi_x}, \;\kappa}$}
                \EndCase
            \EndMatch
        \EndProcedure
    \end{algorithmic}
\end{algorithm}

\begin{algorithm}[ht]
    \caption{Disjunction}\label{alg:disj}
    \begin{algorithmic}[1]
        \Procedure {\textcolor{blue}{Disj}}{$\Gamma$, $\tau_1$, $\tau_2$}
            \Match {$\tau_1$, $\tau_2$}
                \Case {$\under{v: t \;|\; \phi_1 \;|\; \psi_1}$, $\under{v: t \;|\; \phi_2 \;|\; \psi_2}$}
                    \State \Return $\under{v: t \;|\; \phi_1 \;\lor\; \phi_2 \;|\; \psi_1 \;\lor\; \psi_2}$
                \EndCase
                \Case {$(\tau_{H_1}, \;H_1)$, $(\tau_{H_2}, \;H_2)$}
                    \State \Return $(\Call{\textcolor{blue}{Disj}}{\Gamma, \;\tau_{H_1}, \;\tau_{H_2}}, \;H_1 + H_2)$
                \EndCase
                \Case {$\overa{v: t \;|\; \phi_1}$, $\overa{v: t \;|\; \phi_2}$}
                    \State \Return $\overa{v: t \;|\; \phi_1 \;\land\; \phi_2}$
                \EndCase
                \Case {$a{:}\tau_{a_1} \rightarrow \kappa_1$, $a{:}\tau_{a_2} \rightarrow \kappa_2$}
                    \State $\tau_a \gets \Call{\textcolor{blue}{Conj}}{\Gamma, \;\tau_{a_1}, \;\tau_{a_2}}$
                    \State \Return {$a{:}\tau_a \rightarrow \Call{\textcolor{blue}{Disj}}{\Gamma, \;\kappa_1, \;\kappa_2}$}
                \EndCase
            \EndMatch
        \EndProcedure
    \end{algorithmic}
\end{algorithm}

\begin{algorithm}[ht]
    \caption{Conjunction}\label{alg:conj}
    \begin{algorithmic}[1]
        \Procedure {\textcolor{blue}{Conj}}{$\Gamma$, $\tau_1$, $\tau_2$}
            \Match {$\tau_1$, $\tau_2$}
                \Case {$\under{v: t \;|\; \phi_1 \;|\; \psi_1}$, $\under{v: t \;|\; \phi_2 \;|\; \psi_2}$}
                    \State \Return $\under{v: t \;|\; \phi_1 \;\land\; \phi_2 \;|\; \psi_1 \;\land\; \psi_2}$
                \EndCase
                \Case {$(\tau_{H_1}, \;H_1)$, $(\tau_{H_2}, \;H_2)$}
                    \State $H_3 \gets \Call{\textcolor{Mahogany}{HistConj}}{\Gamma, \;H_1, \;H_2}$
                    \If{$H_3 = \bot$}
                        \State \Return \Failure
                    \Else
                        \State \Return $(\Call{\textcolor{blue}{Conj}}{\Gamma, \;\tau_{H_1}, \;\tau_{H_2}}, \;H_3)$
                    \EndIf
                \EndCase
                \Case {$\overa{v: t \;|\; \phi_1}$, $\overa{v: t \;|\; \phi_2}$}
                    \State \Return $\overa{v: t \;|\; \phi_1 \;\lor\; \phi_2}$
                \EndCase
                \Case {$a{:}\tau_{a_1} \rightarrow \kappa_1$, $a{:}\tau_{a_2} \rightarrow \kappa_2$}
                    \State $\tau_a \gets \Call{\textcolor{blue}{Disj}}{\Gamma, \;\tau_{a_1}, \;\tau_{a_2}}$
                    \State \Return {$a{:}\tau_a \rightarrow \Call{\textcolor{blue}{Conj}}{\Gamma, \;\kappa_1, \;\kappa_2}$}
                \EndCase
            \EndMatch
        \EndProcedure
    \end{algorithmic}
\end{algorithm}

\end{document}